\documentclass{article}

\usepackage{cite}
\usepackage[utf8]{inputenc}
\usepackage{geometry}
\usepackage{authblk}
\usepackage{subcaption}
\usepackage[utf8]{inputenc} 
\usepackage[T1]{fontenc}
\usepackage{url}
\usepackage{ifthen}
\usepackage{cite}
\usepackage[cmex10]{amsmath} % Use the [cmex10] option to ensure complicance
\usepackage{amssymb,lipsum, amsthm} 
\usepackage{ bm}
\allowdisplaybreaks
\usepackage{amsthm}                       % with IEEE Xplore (see bare_conf.tex)
\usepackage{tikz}
\usetikzlibrary{patterns}
\usepackage{pgfplots}
\pgfplotsset{compat=1.9}

\usepgfplotslibrary{fillbetween}
\usepackage{graphicx}
\usepackage{pgfplots}
\usepackage{caption}
\usepackage{float, stfloats}
\usepackage{latexsym}
\usepackage{nicefrac}
\usepackage{Notations}

\usepackage{hyperref}
\usepackage{sidecap}
\sidecaptionvpos{figure}{t}

\usepackage{booktabs}       % professional-quality tables
\usepackage{amsfonts}       % blackboard math symbols
\usepackage{microtype}      % microtypography
\usepackage{xcolor}         % colors

\newtheorem{theorem}{Theorem}
\newtheorem{statement}[theorem]{\textbf{Statement}}

\newtheorem{lemma}{Lemma}[section]
\newtheorem{proposition}{Proposition}

\theoremstyle{definition}
\newtheorem{remark}{Remark}

\newtheorem{assumption}{Assumption}
\title{\bf{Rectangular Rotational Invariant Estimator for High-Rank Matrix Estimation}}

\author{Farzad Pourkamali,
Nicolas Macris}
\affil{\emph{School of Computer and Communication Science, Ecole Polytechnique F{\'e}d{\'e}rale de Lausanne}}
\date{}

\begin{document}

% \DOI{DOI HERE}
% \copyrightyear{2021}
% \vol{00}
% \pubyear{2021}
% \access{Advance Access Publication Date: Day Month Year}
% \appnotes{Paper}
% \copyrightstatement{Published by Oxford University Press on behalf of the Institute of Mathematics and its Applications. All rights reserved.}
% \firstpage{1}

% %\subtitle{Subject Section}

% \title[Rectangular RIE]{Rectangular Rotational Invariant Estimator for High-Rank Matrix Estimation}

% \author{Farzad Pourkamali, Nicolas Macris
% \address{\orgdiv{School of Computer and Communication Sciences}, \orgname{EPFL}, \country{Switzerland}}}

% \authormark{F. Pourkamali and N. Macris}

% \corresp[*]{Corresponding author: \href{email:farzad.pourkamali@epfl.ch}{farzad.pourkamali@epfl.ch}}

% \received{Date}{0}{Year}
% \revised{Date}{0}{Year}
% \accepted{Date}{0}{Year}

\maketitle

\begin{abstract}
    We consider estimating a matrix from noisy observations coming from an arbitrary additive bi-rotational invariant perturbation. We propose an estimator which is optimal among the class of \textit{rectangular rotational invariant estimators} and can be applied irrespective of the prior on the signal. For the particular case of Gaussian noise, we prove the optimality of the proposed estimator, and we find an explicit expression for the MMSE in terms of the limiting singular value distribution of the observation matrix. Moreover, we prove a formula linking the asymptotic mutual information and the limit of a log-spherical integral of rectangular matrices. We also provide numerical checks for our results for general bi-rotational invariant noise, as well as Gaussian noise, which match our theoretical predictions.
\end{abstract}

\section{Introduction}

\textit{Matrix denoising} is the problem of removing or reducing  noise from a given data matrix while preserving important features or structure of the signal. This problem is a fundamental constituent of modern data analysis which aims to extract insightful information from noisy high-dimensional datasets, which are often presented as large matrices. Given its fundamental role, this problem  has attracted a lot of attention from both theoretical and algorithmic point of views, and its analysis involves modern mathematical tools from high-dimensional probability theory, statistics, and random matrix theory. In its simplest form, the problem can be formulated as follows. Let the data matrix be 
\begin{equation*}
    \bY = \bS + \bZ
\end{equation*}
where $\bS$ is the hidden signal of interest, and $\bZ$ a noise matrix. The general problem addressed in this paper is to establish the fundamental limits of Bayesian optimal and algorithmic estimations, minimizing the squared error, for $\bS$ given $\bY$, assuming knowledge of the priors on the signal and noise. 

Suppose that the hidden signal matrix has rank $P$ and has eigenvalue (if $\bS=\bS^\intercal$) or singular value (if $\bS\neq\bS^\intercal$) decomposition
 \begin{equation*}
     \bS = \sum_{j = 1}^P \lambda_j \bu_j \bu_j^\intercal, \quad {\rm or} \quad \bS = \sum_{j = 1}^P \lambda_j \bu_j \bv_j^\intercal \quad \bu_j \in \bR^{N}, \bv_j \in \bR^M
 \end{equation*}
Most existing rigorous results concern the regime of fixed rank $P$ as $N, M \to \infty$. In regimes, where $P$ grows with the dimensions, which are the ones studied in this paper, the problem is much harder and there are much fewer attempts to study the denoising problem. We briefly summarize the state of the art in the literature, and then summarize our main contributions.  

\paragraph{Finite rank.}
For finite-rank, with $P$ fixed as $N, M \to \infty$, the denoising problem and its more involved counterpart, the factorization problem, are well studied. The behavior of eigenvalues/singular values and eigenvector/singular vectors of finite-rank perturbations of a Gaussian matrix is studied in \cite{baik2005phase, benaych2011eigenvalues, benaych2012singular} which leads to spectral estimators for  low-rank signals $\bS$ when the noise matrix $\bZ$ is Gaussian distributed. For the case of entry-wise factorized prior on the vectors $\bu_j, \bv_j$, and Gaussian noise, closed form expressions have been established for the asymptotic Bayes-optimal estimation error \cite{lesieur2017constrained, dia2016mutual, lelarge2019fundamental, miolane2017fundamental, barbier2019adaptive}, and iterative algorithms based on approximate message passing has been proposed \cite{montanari2021estimation, fletcher2018iterative}. The low-rank matrix denoising problem has been addressed in various other settings, such as structured noise matrices \cite{barbier2023fundamental, fan2022approximate}, mismatched estimation problems \cite{pourkamali2021mismatched, barbier2022price, pourkamali2022mismatched, guionnet2023estimating}, and estimation in the regime with diverging aspect-ratio of  matrices ($\nicefrac{N}{M} \to \infty$ or $\nicefrac{N}{M} \to 0$) \cite{montanari2022fundamental}.

\paragraph{Sub-linear rank regime.} For {\it symmetric} signals with factorized prior $\bS = \sum_{j = 1}^P \lambda_j \bu_j \bu_j^\intercal$ with $P = N^{\beta}$ for any $\beta \in (0,1)$, it is shown in \cite{pourkamali2023matrix} that under Gaussian noise the rank-one formula for the mutual information and the Bayes-optimal error is still valid, and a decimation algorithm based on \cite{camilli2023matrix,camilli2023decimation} is proposed to estimate the factors. Moreover, under rotational invariance of the signal prior, the same phenomenon is rigorously established in \cite{husson2022spherical}, namely the mutual information and the Bayes-optimal error can be deduced from the rank-one formula. 

\paragraph{Linear rank regime.} When the noise is Gaussian, the problem has been studied in \cite{kabashima2016phase, barbier2022statistical, maillard2022perturbative, troiani2022optimal, pourkamali2023matrix}, and various algorithms are proposed in \cite{troiani2022optimal, camilli2023matrix, bodin2023gradient}. A more general class of noise priors is considered in \cite{bun2016rotational, Pourkamali2023RectangularRI, landau2023singular}, where the noise is assumed to be rotational invariant. For this kind of noise, a class of estimators called \textit{Rotational Invariant Estimator} (RIE) are proposed. These estimators are constructed from the observation matrix by modifying the singular values/eigenvalues without changing the singular vectors/eigenvectors. For the denoising problem, these estimators have been studied in \cite{bun2016rotational, pourkamali2023matrix, semerjian2024matrix} for symmetric matrices, and generalized to rectangular matrices in \cite{Pourkamali2023RectangularRI, landau2023singular}. Other applications of RIEs include matrix factorization \cite{pourkamali2023bayesian} and covariance estimation \cite{stein1975estimation,takemura1984orthogonally, ledoit2011eigenvectors, bun2017cleaning, benaych2023optimal}.\\

\paragraph{Main contributions.} 
We consider denoising a {\it non-symmetric rectangular} matrix under additive bi-rotational invariant noise. Our contributions are summarized below:
\begin{itemize}
    \item We extend the rotational invariant estimators to rectangular matrices. We conjecture that the proposed estimator is optimal among the RIE class under general bi-rotational invariant noise. 

    \item For the particular case of Gaussian noise:
        \begin{itemize}
            \item We prove a trace relation which gives a solid justification for the optimality of the proposed RIE.
            \item Using the optimality of the RIE, we derive the asymptotic Bayes-optimal error in terms of the limiting singular value distribution of the observation.
            \item We prove by independent methods that the mutual information between $\bS$ and $\bY$ is linked to the asymptotic log-spherical integral, an object which has been studied in the theoretical physics and mathematics literature \cite{guionnet2021large}.
    \end{itemize}

    \item We provide numerical simulations under various settings, which
        \begin{itemize}
            \item support the optimality of the proposed (general) RIE, as suggested by the derivation based on (non-rigorous) methods from statistical physics.
            \item suggest that RIE is not limited to the rotational invariant signals, and can be applied regardless of of the prior to get non-trivial (although non-optimal) estimates. 
        \end{itemize}
\end{itemize}

The paper is organized as follows. In section \ref{rect-RIE-section}, we introduce the model in more details and present an explicit RIE with its associated algorithm for the general class of bi-rotational invariant noise. In section \ref{Gaussian Noise}, we focus on Gaussian noise and prove the optimality of the RIE, and study the Bayes-optimal error and mutual information. Section \ref{numerics} is devoted to numerical simulations, followed by proof/derivation details in section \ref{Proof-details}.

A part of this work was presented in the conference ISIT 2023 \cite{Pourkamali2023RectangularRI}.

\textit{Notation}: For a sequence of matrices $\bA$ of growing size, we denote the limiting empirical singular value distribution (ESD) by $\mu_A$, and the limiting eigenvalue distribution $\bA \bA^\intercal$ by $\rho_A$. The free rectangular convolution \cite{benaych2009rectangular} with ratio $\alpha \in [0,1]$  of two probability distributions is denoted 
 by $\boxplus_\alpha$.

% 
% \newpage

\section{Denoising Model and Rotational Invariant Estimators} \label{rect-RIE-section}
Let $\bS \in \bR^{N \times M}$ be the signal matrix that we aim to estimate from the observation matrix $\bY$:
\begin{equation}\label{model}
    \bY = \sqrt{\lambda} \bS + \bZ
\end{equation}
where $\bZ \in \bR^{N \times M}$ is a bi-rotationally invariant matrix, i.e. $P_{Z}(\bZ) = P_{Z}( \bU \bZ \bV^\intercal)$ for any orthogonal matrices $\bU, \bV$, and $\lambda \in \bR_+$ is proportional to the signal-to-noise-ratio (SNR).  We assume that $M$ scales like $N$, and $N/M \to \alpha$. Moreover, we assume that the empirical singular value distributions (ESD) of $\bZ$ and $\bY$  have well-defined limiting measures as $N \to \infty$. We denote them $\mu_Z, \mu_Y$ respectively and refer to them as {\it limiting ESD}. Studying the problem for the case $\alpha \in (0,1]$ suffices. Indeed, suppose the observation matrix $\bY \in \bR^{N \times M}$ has dimensions $N>M$ (so $\alpha >1$),  then exchanging the role of $M, N$, we can apply our results to the matrix $\bY^\intercal$ with aspect ratio $1/\alpha \in (0,1)$. 
%This enforces that the entries of $\bS$ and $\bZ$ are of the order $O(\frac{1}{\sqrt{N}})$ 

\subsection{Rectangular RIE Class}
Given the observation $\bY$, the class of \textit{Rotational Invariant Estimators (RIE)} $\bm{\Xi}_S(\bY)$ of $\bS$ have the same singular vectors than $\bY$. More precisely, consider the SVD of $\bY$ to be:
\begin{equation*}
    \bY = \bU_Y \bGam \bV_Y^\intercal, \quad \quad  \bGam =   \left[
\begin{array}{c|c}
\rm{diag}(\gamma_1, \cdots, \gamma_N )  & \mathbf{0}_{N \times (M-N)}
\end{array}
\right] \in \bR^{N \times M}
\end{equation*}
with $\gamma_1, \cdots, \gamma_N \geq 0$ singular values of $\bY$, and orthogonal matrices $\bU_Y \in \bR^{N \times N}, \bV_Y \in \bR^{M \times M}$. RIEs $\bm{\Xi}_S(\bY)$ are constructed by definition as :
\begin{equation}\label{RIE-def}
    \bm{\Xi}_S(\bY) = \sum_{j = 1}^N \xi_j \bu_j \bv_j^{\intercal} 
\end{equation}
where $\bu_j, \bv_j$ are columns of $\bU_Y, \bV_Y$. The goal is to have the minimum squared error, therefore the optimal singular values are the solution to the following optimization problem:
\begin{equation}\label{opt-squared-error}
    \min_{\xi_1, \cdots, \xi_N} \| \bS - \bm{\Xi}_S(\bY) \|_{\rm F}^2
\end{equation}
One can easily see that the solutions to optimization problem \eqref{opt-squared-error} are:
\begin{equation}\label{Oracle-sv}
    \xi^*_j = \bu_j^\intercal \bS \bv_j \quad \quad {\rm for} \quad 1 \leq j \leq N
\end{equation}
The particular estimator constructed with the singular values \eqref{Oracle-sv} is denoted by $\bm{\Xi}^*_S(\bY)$, and is called \textit{oracle estimator} as it involves the signal matrix $\bS$.

\subsection{Algorithmic RIE}
Our main contribution is the derivation of an explicit formula for the optimal singular values \eqref{Oracle-sv} which only involves the observation matrix and the knowledge of spectral measure of the noise. This formula leads to an algorithm for the estimation, which we conjecture, has in the asymptotic limit a performance matching the one of the oracle estimator (in the sense of the mean-square-error \eqref{opt-squared-error}.

The optimal singular values can be approximated for sufficiently large $N$, as:
\begin{equation}
\begin{split}
       \widehat{\xi_j^*} &=\frac{1}{\sqrt{\lambda}} \Bigg[ \gamma_j - \frac{1}{\pi \bar{\mu}_{Y}(\gamma_j)} {\rm Im} \, \mathcal{C}^{(\alpha)}_{\mu_Z}\bigg( \frac{1- \alpha}{\gamma_j} \pi \sH [\bar{\mu}_{Y}](\gamma_j)+ \alpha \big( \pi \sH [\bar{\mu}_{Y}](\gamma_j)\big)^2  - \alpha \big( \pi \bar{\mu}_{Y}(\gamma_j)\big)^2 \\ 
       &\hspace{150pt}+\ci \pi \bar{\mu}_{Y}(\gamma_j) \big(\frac{1-\alpha}{\gamma_j} + 2 \alpha \pi \sH [\bar{\mu}_{Y}](\gamma_j) \big) \bigg) \Bigg]
\end{split}
\label{rect-RIE}
\end{equation}
where $ \bar{\mu}_{Y}(\gamma) = \frac{1}{2}(\mu_Y(\gamma) + \mu_Y(-\gamma))$ is the symmetrization of the limiting ESD of $\bY$, $\mathcal{C}^{(\alpha)}_{\mu_Z}$ is the rectangular R-transform of $\mu_Z$, and $\sH [\bar{\mu}_{Y}]$ is the Hilbert transform of $\bar{\mu}_{Y}$. The definitions of these objects are reviewed in appendix \ref{RMT-transforms}, and the derivation of the estimator \eqref{rect-RIE} is sketched in section \ref{RIE-derivation-sec}.

The algorithm to estimate $\bS$ proceeds as follows:
\begin{enumerate}
    \item Compute the SVD of $\bY$, $\bY = \bU_Y \bGam \bV_Y^\intercal$.

    \item Approximate $\mathcal{G}_{\bar{\mu}_Y}(z)$ from the singular values of $\bY$, from which $\bar{\mu}_{Y}(\gamma)$ and $\sH [\bar{\mu}_{Y}](\gamma)$ can be evaluated using \eqref{Plemelj formulae}.

    \item Compute $\widehat{\xi_j^*}$ as in \eqref{rect-RIE}, and construct the estimator $ \widehat{\bm{\Xi}^*_S}(\bY) = \sum_{j = 1}^N \widehat{\xi_j^*} \bu_j \bv_j^{\intercal}$.
\end{enumerate}

% \subsubsection{Recovering the RIE for Denoising Symmetric Matrices}
% In model \eqref{model}, suppose that both $\bS, \bZ$ are symmetric matrices and the limiting eigenvalues distribution of $\bZ$ is a symmetric distribution, $\rho_Z(x) = \rho_Z(-x)$. Then, using the relation $\mathcal{C}^{(1)}_{\mu_Z}(z) = \sqrt{z} \mathcal{R}_{\rho_Z}(\sqrt{z})$ \textcolor{red}{?????}

\subsection{Bayes Optimality and MMSE}
From the Bayesian estimation point of view, considering a prior distribution for the signal $P_S(\bS)$, one wishes to minimize the average 
mean-squared-error (MSE), which is defined for an estimator $\Theta_S : \bR^{N \times M} \to \bR^{N \times M}$ as
\begin{equation*}
    {\rm MSE}_{\Theta_S}  = \frac{1}{N} \bE \big\| \bS - \Theta_S(\bY) \|_{\rm F}^2 \quad,
\end{equation*}
where the expectation is over $\bS, \bZ$. It is well known that the estimator which has the minimum MSE is the posterior mean estimator $\Theta_S^*(\bY) = \bE[ \bS | \bY ]$. 

Note that for model \eqref{model} the oracle estimator \eqref{Oracle-sv} is the best estimator  among the \textit{RIE class} (in the sense that it minimizes the MSE in this class). Furthermore the derivation of the explicit estimator \eqref{rect-RIE} does not involve Bayesian methodology and does not require any knowledge of the prior of the signal. 

However, if the prior on the signal is bi-rotationally invariant, i.e. $P_{S}(\bS) = P_{S}( \bU \bS \bV^\intercal)$ for any orthogonal matrices $\bU, \bV$, these estimators are intimately related to the Bayesian one. As shown in section \ref{optimality-RIE} for bi-rotationally invariant signal distributions the posterior mean estimator $ \bE[ \bS | \bY ]$ belongs to the RIE class. Since the oracle estimator has minimum MSE among the RIE class, we have that ${\rm MSE}_{\bm{\Xi}^*_S} \leq {\rm MSE}_{\Theta_S^*}$. On the other hand, by definition, we have ${\rm MSE}_{\bm{\Xi}^*_S} \leq {\rm MSE}_{\Theta_S^*}$. Therefore, the oracle estimator is Bayes-optimal under bi-rotational invariant prior and achieves the MMSE. 

Moreover, the "exact" analytical derivation of the explicit estimator \eqref{rect-RIE} suggests that it has the same performance as the oracle estimator as $N \to \infty$. Therefore, the above algorithm should be asymptotically Bayes-optimal with an asymptotic MSE equal to the MMSE. Denoting the rhs in \eqref{rect-RIE} as a {\it function} of singular values of $\bY$, $\widehat{\xi^*} : {\rm supp}(\mu_Y)\to \mathbb{R}$, we are led to the following result.

\begin{statement}[\textbf{MMSE}]\label{General-MMSE-statement}
    Suppose that $\bS, \bZ$ have bi-rotational invariant priors, and Assume their ESDs converge to well-defined measures $\mu_S, \mu_Z$ with bounded second moments. We have:
    \begin{equation}
        \lim_{N \to \infty} {\rm MMSE}_N(\lambda) = \int x^2 \mu_S(x) \, dx - \int {\widehat{\xi^*}(x)}^2 \mu_Y(x) \, dx
        \label{MMSE-eq}
    \end{equation}
    where $\mu_Y$ is the limiting ESD of $\bY$, $\mu_Y = \mu_S \boxplus_{\alpha} \mu_Z$.
        % \vspace{-5mm}
\end{statement}

\begin{remark}
Note that, for non-rotation invariant priors the estimator \eqref{rect-RIE} still can be applied, however it may results in a sub-optimal estimate of the signal. However, this estimate can be used as a "warmed-up" spectral initialization for more efficient algorithms (see for example \cite{mondelli2021approximate, montanari2021estimation}).
\end{remark}

\section{Gaussian noise}\label{Gaussian Noise}

\subsection{RIE}
In this section, we consider the case of Gaussian noise matrix, more precisely we suppose that $\bZ \in \bR^{N \times M}$ has i.i.d. Gaussian entries of variance $\nicefrac{1}{N}$, and $\lambda \in \bR_+$ is proportional to the signal-to-noise-ratio (SNR). We make the following assumptions:

\begin{assumption}
    The operator norm of $\bS$, and the ratio $\nicefrac{M}{N}$ are bounded by some numerical constant $K>0$ independent of $N,M$. 
\end{assumption}

% \subsection{Rotationally Invariant Estimator (RIE)}
% Given the observation $\bY$, a \textit{Rotationally Invariant Estimator (RIE)} $\bm{\Xi}_S(\bY)$ of $\bS$ is constructed by changing singular values of $\bY$ but not changing its singular vectors. Consider the SVD of $\bY$ to be:
% \begin{equation*}
%     \bY = \bU_Y \bGam \bV_Y^\intercal, \quad \quad  \bGam =   \left[
% \begin{array}{c|c}
% \rm{diag}(\gamma_1, \cdots, \gamma_N )  & \mathbf{0}_{N \times (M-N)}
% \end{array}
% \right] \in \bR^{N \times M}
% \end{equation*}
% with $\gamma_1, \cdots, \gamma_N \geq 0$ singular values of $\bY$, and orthogonal matrices $\bU_Y \in \bR^{N \times N}, \bV_Y \in \bR^{M \times M}$. A RIE $\bm{\Xi}_S(\bY)$ is constructed as :
% \begin{equation*}
%     \bm{\Xi}_S(\bY) = \sum_{j = 1}^N \xi_j \bu_j \bv_j^{\intercal} 
% \end{equation*}
% where $\bu_j, \bv_j$ are columns of $\bU_Y, \bV_Y$. The goal is to have the minimum squared error, therefore the optimal singular values are the solution to the following optimization problem:
% \begin{equation}\label{opt-squared-error}
%     \min_{\xi_1, \cdots, \xi_N} \| \bS - \bm{\Xi}_S(\bY) \|_{\rm F}^2
% \end{equation}
% One can easily see that the solutions to optimization problem \eqref{opt-squared-error} are:
% \begin{equation}\label{Oracle-sv}
%     \xi^*_j = \bu_j^\intercal \bS \bv_j \quad \quad {\rm for} \quad 1 \leq j \leq N
% \end{equation}
% The estimator constructed with singular values $\xi^*_j$ is denoted by $\bm{\Xi}^*_S(\bY)$, and is called \textit{oracle estimator} as it involves the signal matrix $\bS$.

Recall that the resolvent of the matrix $\bY \bY^\intercal$, evaluated at $z^2$ is defined as:
\begin{equation*}
    \bG_{YY^\intercal}(z^2) = \big( z^2 \bI - \bY \bY^\intercal \big)^{-1}
\end{equation*}
Now, define two random functions of $z \in \mathbb{C}\backslash \bR$ as:
\begin{equation}\label{traces}
    G(z) = \frac{1}{N} \Tr  \bG_{YY^\intercal} (z^2), \quad L(z) = \frac{1}{N} \Tr \bG_{YY^\intercal} (z^2) \bY \bS^\intercal
\end{equation}
%Eigenvalues of $\cY$ are signed singular values of $\bY$, therefore the limiting eigenvalue distribution of $\cY$ (ignoring zero eigenvalues) is the same as the limiting symmetrized singular value distribution of $\bY$.

\begin{proposition}\label{optimal-singular-values&traces}
    For $1 \leq j \leq N$, and for any $\epsilon > 0$ such that $[\gamma_j -\epsilon, \gamma_j +\epsilon] \cap \{ \gamma_1, \cdots, \gamma_N \} = \{ \gamma_j\} $, the optimal singular values \eqref{Oracle-sv} satisfy
    \begin{equation}\label{exact-RND}
        \xi_j^* =  \lim_{\eta \to 0} \frac{ \int_{\gamma_j - \epsilon}^{\gamma_j + \epsilon} \im  L(x + \ci \eta ) \, dx}{\int_{\gamma_j - \epsilon}^{\gamma_j + \epsilon} \im \big\{ (x + \ci \eta ) G(x + \ci \eta ) \big\} \, dx}
    \end{equation}
\end{proposition}

The proof of the above Proposition is presented in section \ref{Proof-Prop2}. Note that \eqref{exact-RND} is an exact formula for the optimal singular values $\xi_j^*$, but in practice given an explicit expression for $L(z)$ we use the following approximation to evaluate $\xi_j^*$:
\begin{equation}\label{optimal-sv-approx}
    \xi_j^* \approx \frac{\im L(z)}{\im \big\{ z G(z) \big\}} \equiv \widehat{\xi^*_j} \quad \quad {\rm for}\hspace{5pt} z = \gamma_j + \ci \eta, \quad {\rm with} \hspace{5pt} \eta \ll 1
\end{equation}
% The error of this approximation 
% Denote the estimator constructed with $\widehat{\xi^*_j}$ by $\widehat{\bm{\Xi}^*_S}(\bY)$. 
%\subsection{Estimation of $L(z)$}
The definition of the function $L(z)$ in the numerator of the estimator \eqref{optimal-sv-approx}, involves the signal matrix. Therefore, to use the estimator, we need to find a way to estimate this function only from the data. In the following theorem, we give an asymptotic approximation of $L(z)$, which we prove in section \ref{Proof-Thm2}.

\begin{theorem}[Estimation of $L(z)$]\label{trace-relation} 
Let $\alpha_0 = \nicefrac{N}{M}$. For any $z \in \mathbb{C} \backslash \bR$ with $\big| \im z \big| < 1$, the function $L(z)$ defined in \eqref{traces} satisfies
\begin{equation}\label{L(z)-approx}
    L(z) = \frac{1}{\sqrt{\lambda}} \Big[ G(z) \big(z^2 +1- \frac{1}{\alpha_0} \big) - z^2 G^2(z)  -1 \Big] + \epsilon_N
 \end{equation}
where the error term $\epsilon_N$ is bounded as:
\begin{equation*}
    \epsilon_N \leq \frac{C_K + X}{ N \big| \im z \big|^3 }
\end{equation*}
with $C_K$ a constant depending on $K$, and $X$ is a complex sub-Gaussian random variables with finite sub-Gaussian norm depending on $K$.
\end{theorem}
% \begin{proof}
%     Section \ref{Proof-Thm2}.
% \end{proof}
\begin{remark}\label{error-term-imz}
    We believe that adopting the methodology developed in \cite{erdHos2009semicircle, erdHos2017dynamical} the approximation \eqref{L(z)-approx} can be improved, with the error term controlled by $\big( N | \im z | \big)^{-1}$. We numerically verify this conjecture in section \ref{numerical-thm-check}.
\end{remark}

\subsubsection{Algorithm}
%Note that \eqref{exact-RND} is an exact formula for the optimal singular values $\xi_j^*$, but in practice, 
Using the explicit expression \eqref{L(z)-approx} for $L(z)$, we are led to the following approximation to evaluate $\xi_j^*$:
\begin{equation}\label{optimal-sv-approx-alg}
    \widehat{\xi^*_j} = \frac{1}{\sqrt{\lambda}} \frac{\im \Big\{ G(z) \big(z^2 +1- \frac{1}{\alpha_0} \big) - z^2 G^2(z) \Big\} }{\im \big\{ z G(z) \big\}}  \quad \quad {\rm for}\hspace{5pt} z = \gamma_j + \ci \frac{1}{\sqrt{N}}
\end{equation}

\begin{remark}[On imaginary part of $z$ in \eqref{optimal-sv-approx}] Approximating the exact formula \eqref{exact-RND} with the expression \eqref{optimal-sv-approx} is more accurate when $z$ is closer to the real line (small $\eta$). On the other hand, for $z$ close to the real line the error of the approximation \eqref{L(z)-approx} gets worse. Considering that the error in \eqref{L(z)-approx} is controlled by $ \big( N | \im z | \big)^{-1}$, we can see that $\eta = \nicefrac{1}{N^{\epsilon}}$ with any $0<\epsilon<1$ should work properly as $N$ increases. We study the effect of choice of $\epsilon$ in the numerical section, see Fig. \ref{fig:imz}.
\end{remark}

\begin{remark}[Relation to the formula \eqref{rect-RIE}]
First note that $z G(z)$ in the denominator of \eqref{optimal-sv-approx-alg} is the Stieltjes transform of empirical symmetric spectral measure of $\bY$, and in the limit $\eta \to 0$, ${\rm Re} \, \big\{ z G(z)\big\} = \pi \sH [\bar{\mu}_{Y}](x)   $. Therefore, from \eqref{L(z)-approx}, for $z = \gamma_j + \ci \eta$ with $\eta \ll 1$, we have 
\begin{equation}\label{rect-RIE-Gauss}
    \begin{split}
        \frac{\im L(z)}{\im \big\{ z G(z) \big\}} &= \frac{1}{\sqrt{\lambda}} \frac{\rm{Re}\, \big\{ G(z) \big\} 2 \gamma_j \eta + \im \big\{ G(z) \big\} \big(\gamma_j^2 - \eta^2+1- \frac{1}{\alpha_0} \big)}{\gamma_j \im \big\{G(z) \big\} + {\rm Re} \, \big\{G(z)\big\} \eta} -    \frac{1}{\sqrt{\lambda}} 2 {\rm Re} \, \big\{ z G(z)\big\} \\
        &\approx \frac{1}{\sqrt{\lambda}} \Big[ \gamma_j + \big( 1 - \frac{1}{\alpha_0} \big) \frac{1}{\gamma_j} - 2 \pi \sH [\bar{\mu}_{Y}](\gamma_j) \Big]
    \end{split}
\end{equation} 
For $\bZ$ with i.i.d. Gaussian entries of variance $\nicefrac{1}{N}$ and with the assumption that $\alpha_0 \to \alpha$ we have $\mathcal{C}^{(\alpha)}_{\mu_Z}(z) = \frac{1}{\alpha}z$, and thus the expression in \eqref{rect-RIE} reduces to \eqref{rect-RIE-Gauss}.
\end{remark}

\subsection{MMSE}
Given the rather simple expression for the optimal singular values, we can compute the asymptotic MMSE for the particular case of Gaussian noise, see section \ref{com-G-MMSE} for the derivation.

\begin{statement}[\textbf{Gaussian MMSE}]\label{Gaus-MMSE-statement}
    Assume that the prior on $\bS$ is bi-rotational invariant, and the ESD of $\bS$ converges to a well-defined limiting measure $\mu_S$ with compact support and bounded second moment. We have:
    \begin{equation}
    \begin{split}
        \lim_{N \to \infty} {\rm MMSE}_N(\lambda) = \frac{1}{\lambda} \Big[ \frac{1}{\alpha} - &\big(\frac{1}{\alpha} -1 \big)^2 \int \frac{\mu_Y(x)}{x^2}\, dx - \frac{\pi^2}{3} \int {\mu_Y(x)}^3 \, dx \Big]
    \end{split}
        \label{G-MMSE-eq}
    \end{equation}
    where $\mu_Y$ is the limiting ESD of $\bY$ and $\mu_Y = \mu_S \boxplus_{\alpha} \mu_{\rm MP}$.
\end{statement}

\begin{remark} 
    In the symmetric case, the asymptotic MMSE of a Gaussian channel is linked to the free Fisher information of non-commutative random variables \cite{pourkamali2023matrix}. Using this link, we can deduce the continuity of the MMSE as a function of SNR, which rules out the existence of the first-order phase transitions. Moreover, using the I-MMSE relation \cite{guo2005mutual}, this link also implies a rather explicit expression for the asymptotic mutual information. We believe that similar relations hold for the rectangular case and the MMSE should be a continuous function of $\lambda$. However in the rectangular case free probability \cite{benaych2009freerectangular} is much less developed than its symmetric counterpart, and these considerations are beyond the scope of the present paper. 
\end{remark}

\subsection{Mutual information}
In this subsection we prove that the asymptotic mutual information is linked to the asymptotic rectangular spherical integral.  The rectangular spherical integral is defined for two matrices $\bA, \bB \in \bR^{N \times M}$ as:
 \begin{equation*}
    \mathcal{I}_{N,M}(\bA, \bB) := \iint D \bU  \,  D \bV \, e^{N \Tr \bA^\intercal \bU  \bB \bV^\intercal }
 \end{equation*}
where $D \bU, D \bV$ denote the Haar measures over the groups of $N \times N$ and $M \times M$ orthogonal matrices. The asymptotic behavior of these integrals has been studied in \cite{guionnet2021large} which proves that the limit $\lim_{N \to \infty} \frac{1}{N^2} \ln \mathcal{I}_{N,M} (\bA, \bB)$ exists and equals a variational formula given in terms of limiting ESD of $\bA, \bB$ (see Appendix \ref{spherical integral app}).

We make the following assumptions on the prior of $\bS$:
\begin{assumption}\label{assumptions on law}
 The empirical singular value distribution of $\bS$ converges almost surely weakly to a well-defined probability density function $\mu_S(x)$ with compact support in $[C_1, C_2]$ with $C_1, C_2 \in \bR_{\geq 0}$. Moreover, the symmetrization of $\mu_S$ has bounded second moment $\int x^2 \, d \bar{\mu}_S < \infty$, finite non-commutative entropy $\iint \ln | x - y |  \, d \bar{\mu}_S(x) \, d \bar{\mu}_S(y) > -\infty$, and $\int \ln |x| \, d \bar{\mu}_S(x) > - \infty$.
\end{assumption}

\begin{assumption}\label{bounded-mom}
The second moment of $\mu_{S}^{(N)}$ is almost surely bounded.
\end{assumption}

Let 
$$
\mathcal{J}[\mu_{\sqrt{\lambda} S}, \mu_{\sqrt{\lambda} S}\boxplus_{\alpha} \mu_{\rm MP}] = \lim_{N\to +\infty} \frac{1}{N M} \ln \mathcal{I}_{N,M} ( \sqrt{\lambda} \bS, \bY)$$ 
where $\mu_{\sqrt{\lambda} S}$ is the limiting spectral distribution of $\sqrt{\lambda} \bS$ and $\mu_{\rm MP}$ is the  Marchenko-Pastur distribution. 
\begin{theorem}[\textbf{Mutual Information}]\label{theoremMain}
Under assumptions \ref{assumptions on law},\ref{bounded-mom} and $\nicefrac{M}{N} \leq K$, we have:
\begin{equation}
\begin{split}
    \lim_{N \to \infty} \, \frac{1}{M N} I_N(\bS; \bY) = \lambda \alpha \int \! x^2 \mu_S(x) \, dx  - \mathcal{J}[\mu_{\sqrt{\lambda} S}, \mu_{\sqrt{\lambda} S}\boxplus_{\alpha} \mu_{\rm MP}]
\end{split}
    \label{asymp-mI-th}
\end{equation}
\label{MI-th}
\end{theorem}

\begin{remark}\label{asymp-log-rect-sph}
    In \cite{troiani2022optimal}, the asymptotic log-spherical integral $\mathcal{J}[\mu_{\sqrt{\lambda} S}, \mu_{\sqrt{\lambda} S}\boxplus_{\alpha} \mu_{\rm MP}]$ is computed explicitly, which together with Theorem \ref{theoremMain} gives an an explicit expression for the asymptotic mutual information under Gaussian noise.
\end{remark}

\section{Numerical simulations}\label{numerics}

\subsection{General bi-rotational invariant noise}
In Fig. \ref{fig:General-noise}, the performance of the algorithmic RIE based on \eqref{rect-RIE} is compared against the oracle estimator \eqref{Oracle-sv} for two cases of noise distribution and Gaussian signal matrix, i.e. $\bS$ is a matrix with i.i.d. Gaussian entries of variance $\nicefrac{1}{N}$.

\paragraph{Uniform spectral noise.} For this prior, the noise matrix $\bZ \in \bR^{N \times N}$ is constructed as $\bZ = \bU {\rm diag}(r_1, \cdots, r_N) \bV^\intercal$, where $\bU, \bV \in \bR^{N \times N}$ are independent Haar distributed matrices, and the singular values $r_1, \cdots, r_N$ are chosen independently uniformly from $[0,2]$. The limiting spectral measure of the uniform noise distribution $\mathcal{U}_{[0,2]}$ has rectangular R-transform $\mathcal{C}_{\mathcal{U}_{[0,2]}}^{(1)} (z) = 2 \sqrt{z} \coth \big(2\sqrt{z}\big) - 1$.

\paragraph{Sum of rank-one factors.} For the noise matrix we take a sum of rank-one matrices, $\bZ = \sum_{k=1}^L \bu_k \bv_k^\intercal$, where $\bu_k$'s and $\bv_k$'s are independent uniform random vectors of the unit norm in $\bR^N, \bR^M$. Denoting the limiting ratio $\nicefrac{L}{N} \to c$, the liming symmetrized ESD of $\bZ$ is the rectangular analogue of the symmetrized
Poisson distribution with parameter $c$, with rectangular R-transform $\mathcal{C}^{(\alpha)}(z) = \nicefrac{c z}{1 - z}$ (see section 4.3, Proposition 6.1 in \cite{benaych2007infinitely}). 

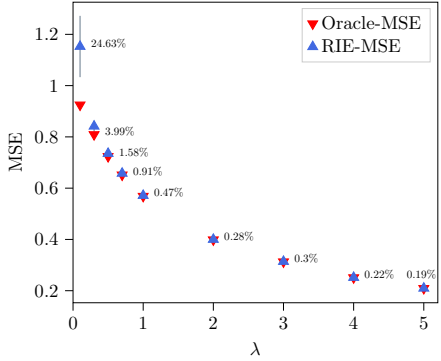
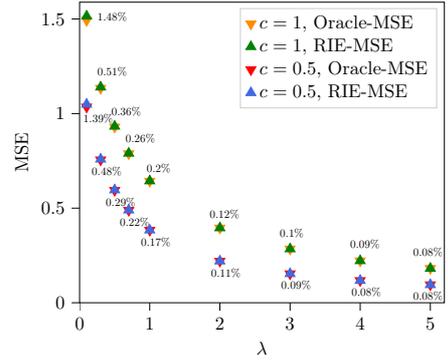
\begin{figure}
\begin{subfigure}[t]{.48\textwidth}
    \centering
    \captionsetup{justification=centering}
    % This file was created with tikzplotlib v0.10.1.
\begin{tikzpicture}[scale = 0.7]

\definecolor{darkgray176}{RGB}{176,176,176}
\definecolor{lightcoral}{RGB}{240,128,128}
\definecolor{lightslategray}{RGB}{119,136,153}
\definecolor{royalblue}{RGB}{65,105,225}

\begin{axis}[
legend cell align={left},
legend style={fill opacity=0.8, draw opacity=1, text opacity=1, draw=white!80!black},
tick align=outside,
tick pos=left,
x grid style={darkgray176},
xlabel={\(\displaystyle \lambda\)},
xmin=0, xmax=5.2,
xtick style={color=black},
y grid style={darkgray176},
ylabel={MSE},
ymin=0.152989628270105, ymax=1.32539033414586,
ytick style={color=black}
]
\path [draw=lightcoral, semithick]
(axis cs:0.1,0.923737564407104)
--(axis cs:0.1,0.925920607197669);

\path [draw=lightcoral, semithick]
(axis cs:0.3,0.804393241100484)
--(axis cs:0.3,0.81386038120777);

\path [draw=lightcoral, semithick]
(axis cs:0.5,0.720281014015007)
--(axis cs:0.5,0.726999508933095);

\path [draw=lightcoral, semithick]
(axis cs:0.7,0.645750561762231)
--(axis cs:0.7,0.657776341270898);

% \path [draw=lightcoral, semithick]
% (axis cs:0.9,0.588136725710247)
% --(axis cs:0.9,0.597341366641143);

\path [draw=lightcoral, semithick]
(axis cs:1,0.563689801272909)
--(axis cs:1,0.573432903917115);

\path [draw=lightcoral, semithick]
(axis cs:2,0.392454165647133)
--(axis cs:2,0.404943625182792);

\path [draw=lightcoral, semithick]
(axis cs:3,0.308128156083209)
--(axis cs:3,0.317630266805713);

\path [draw=lightcoral, semithick]
(axis cs:4,0.247104440654114)
--(axis cs:4,0.254713902511564);

\path [draw=lightcoral, semithick]
(axis cs:5,0.206280569446275)
--(axis cs:5,0.211848418333924);

\path [draw=lightslategray, semithick]
(axis cs:0.1,1.03303033145219)
--(axis cs:0.1,1.27209939296969);

\path [draw=lightslategray, semithick]
(axis cs:0.3,0.81987575197821)
--(axis cs:0.3,0.863033131115337);

\path [draw=lightslategray, semithick]
(axis cs:0.5,0.726273947666152)
--(axis cs:0.5,0.743930384552147);

\path [draw=lightslategray, semithick]
(axis cs:0.7,0.648403717688098)
--(axis cs:0.7,0.666974801994573);

% \path [draw=lightslategray, semithick]
% (axis cs:0.9,0.5901955957818)
% --(axis cs:0.9,0.600516926620797);

\path [draw=lightslategray, semithick]
(axis cs:1,0.565385729356802)
--(axis cs:1,0.577083622597401);

\path [draw=lightslategray, semithick]
(axis cs:2,0.393250011156362)
--(axis cs:2,0.406413307591626);

\path [draw=lightslategray, semithick]
(axis cs:3,0.308871790800908)
--(axis cs:3,0.318785866515836);

\path [draw=lightslategray, semithick]
(axis cs:4,0.247606338085179)
--(axis cs:4,0.255309250887774);

\path [draw=lightslategray, semithick]
(axis cs:5,0.206662737482111)
--(axis cs:5,0.212248083249741);

\addplot [semithick, red, mark=triangle*, mark size=3, mark options={solid,rotate=180}, only marks]
table {%
0.1 0.924829085802387
0.3 0.809126811154127
0.5 0.723640261474051
0.7 0.651763451516565
% 0.9 0.592739046175695
1 0.568561352595012
2 0.398698895414963
3 0.312879211444461
4 0.250909171582839
5 0.2090644938901
};
\addlegendentry{Oracle-MSE}
\addplot [semithick, royalblue, mark=triangle*, mark size=3, mark options={solid}, only marks]
table {%
0.1 1.15256486221094
0.3 0.841454441546774
0.5 0.735102166109149
0.7 0.657689259841335
% 0.9 0.595356261201299
1 0.571234675977102
2 0.399831659373994
3 0.313828828658372
4 0.251457794486477
5 0.209455410365926
};
\addlegendentry{RIE-MSE}
\draw (axis cs:0.2,1.15) node[
  scale=0.6,
  anchor=base west,
  text=black,
  rotate=0.0
]{24.63\%};
\draw (axis cs:0.4,0.809126811154127) node[
  scale=0.6,
  anchor=base west,
  text=black,
  rotate=0.0
]{3.99\%};
\draw (axis cs:0.6,0.723640261474051) node[
  scale=0.6,
  anchor=base west,
  text=black,
  rotate=0.0
]{1.58\%};
\draw (axis cs:0.8,0.651763451516565) node[
  scale=0.6,
  anchor=base west,
  text=black,
  rotate=0.0
]{0.91\%};
% \draw (axis cs:1,0.592739046175695) node[
%   scale=0.6,
%   anchor=base west,
%   text=black,
%   rotate=0.0
% ]{0.44\%};
\draw (axis cs:1.1,0.568561352595012) node[
  scale=0.6,
  anchor=base west,
  text=black,
  rotate=0.0
]{0.47\%};
\draw (axis cs:2.1,0.398698895414963) node[
  scale=0.6,
  anchor=base west,
  text=black,
  rotate=0.0
]{0.28\%};
\draw (axis cs:3.1,0.312879211444461) node[
  scale=0.6,
  anchor=base west,
  text=black,
  rotate=0.0
]{0.3\%};
\draw (axis cs:4.1,0.250909171582839) node[
  scale=0.6,
  anchor=base west,
  text=black,
  rotate=0.0
]{0.22\%};
\draw (axis cs:4.7,0.25) node[
  scale=0.6,
  anchor=base west,
  text=black,
  rotate=0.0
]{0.19\%};
\end{axis}

\end{tikzpicture}
    \caption{Uniform spectral noise matrix,\\ $ N = M = 1000$}
    \label{fig:uniform-noise}
\end{subfigure}
\hfill
\begin{subfigure}[t]{.48\textwidth}
    \centering
    \captionsetup{justification=centering}
    % This file was created with tikzplotlib v0.10.1.
\begin{tikzpicture}[scale = 0.7]

\definecolor{darkgray176}{RGB}{176,176,176}
\definecolor{darkorange}{RGB}{255,140,0}
\definecolor{green}{RGB}{0,128,0}
\definecolor{lightcoral}{RGB}{240,128,128}
\definecolor{lightgreen}{RGB}{144,238,144}
\definecolor{lightslategray}{RGB}{119,136,153}
\definecolor{moccasin}{RGB}{255,228,181}
\definecolor{royalblue}{RGB}{65,105,225}

\begin{axis}[
legend cell align={left},
legend style={fill opacity=0.8, draw opacity=1, text opacity=1, draw=white!80!black},
tick align=outside,
tick pos=left,
x grid style={darkgray176},
xlabel={\(\displaystyle \lambda\)},
xmin=0, xmax=5.2,
xtick style={color=black},
y grid style={darkgray176},
ylabel={MSE},
ymin=-0, ymax=1.58883287640586,
ytick style={color=black}
]
\path [draw=moccasin, semithick]
(axis cs:0.1,1.49154312527954)
--(axis cs:0.1,1.49414092816245);

\path [draw=moccasin, semithick]
(axis cs:0.3,1.13314154953142)
--(axis cs:0.3,1.13492923892317);

\path [draw=moccasin, semithick]
(axis cs:0.5,0.928468554435307)
--(axis cs:0.5,0.93026336225066);

\path [draw=moccasin, semithick]
(axis cs:0.7,0.786929071128695)
--(axis cs:0.7,0.788852871005469);

\path [draw=moccasin, semithick]
(axis cs:1,0.642015481550515)
--(axis cs:1,0.643133467987065);

\path [draw=moccasin, semithick]
(axis cs:2,0.39471665882669)
--(axis cs:2,0.395358793073859);

\path [draw=moccasin, semithick]
(axis cs:3,0.283750376167518)
--(axis cs:3,0.284481621375888);

\path [draw=moccasin, semithick]
(axis cs:4,0.221343339328497)
--(axis cs:4,0.221821689603955);

\path [draw=moccasin, semithick]
(axis cs:5,0.181339396113788)
--(axis cs:5,0.181674567951811);

\path [draw=lightgreen, semithick]
(axis cs:0.1,1.51204107781874)
--(axis cs:0.1,1.51769697873993);

\path [draw=lightgreen, semithick]
(axis cs:0.3,1.13853131645474)
--(axis cs:0.3,1.14108682663884);

\path [draw=lightgreen, semithick]
(axis cs:0.5,0.931657123037624)
--(axis cs:0.5,0.933818927506093);

\path [draw=lightgreen, semithick]
(axis cs:0.7,0.788957837199194)
--(axis cs:0.7,0.79085354010593);

\path [draw=lightgreen, semithick]
(axis cs:1,0.643222840444468)
--(axis cs:1,0.644474235611744);

\path [draw=lightgreen, semithick]
(axis cs:2,0.395169692023641)
--(axis cs:2,0.395865723975969);

\path [draw=lightgreen, semithick]
(axis cs:3,0.284029199118663)
--(axis cs:3,0.284760324251864);

\path [draw=lightgreen, semithick]
(axis cs:4,0.221546791271514)
--(axis cs:4,0.222023356156209);

\path [draw=lightgreen, semithick]
(axis cs:5,0.181493892147687)
--(axis cs:5,0.181821486393159);

\path [draw=lightcoral, semithick]
(axis cs:0.1,1.03345472510892)
--(axis cs:0.1,1.03609425633274);

\path [draw=lightcoral, semithick]
(axis cs:0.3,0.754228492059106)
--(axis cs:0.3,0.756178106741535);

\path [draw=lightcoral, semithick]
(axis cs:0.5,0.593763851164255)
--(axis cs:0.5,0.59493370697408);

\path [draw=lightcoral, semithick]
(axis cs:0.7,0.488018417154994)
--(axis cs:0.7,0.489210968758979);

\path [draw=lightcoral, semithick]
(axis cs:1,0.383221830362838)
--(axis cs:1,0.383790906428633);

\path [draw=lightcoral, semithick]
(axis cs:2,0.219939209336089)
--(axis cs:2,0.220218639158206);

\path [draw=lightcoral, semithick]
(axis cs:3,0.153167996908189)
--(axis cs:3,0.153453261286375);

\path [draw=lightcoral, semithick]
(axis cs:4,0.117292060462972)
--(axis cs:4,0.117516609890495);

\path [draw=lightcoral, semithick]
(axis cs:5,0.0949790254213866)
--(axis cs:5,0.0951502480244167);

\path [draw=lightslategray, semithick]
(axis cs:0.1,1.04756558110824)
--(axis cs:0.1,1.05076687856072);

\path [draw=lightslategray, semithick]
(axis cs:0.3,0.757955495887175)
--(axis cs:0.3,0.759720906710571);

\path [draw=lightslategray, semithick]
(axis cs:0.5,0.595528104822001)
--(axis cs:0.5,0.596639091061839);

\path [draw=lightslategray, semithick]
(axis cs:0.7,0.489098122487344)
--(axis cs:0.7,0.490314940291618);

\path [draw=lightslategray, semithick]
(axis cs:1,0.383889369501796)
--(axis cs:1,0.384422222184902);

\path [draw=lightslategray, semithick]
(axis cs:2,0.220185074956458)
--(axis cs:2,0.220453711720534);

\path [draw=lightslategray, semithick]
(axis cs:3,0.15330372891298)
--(axis cs:3,0.153590668764985);

\path [draw=lightslategray, semithick]
(axis cs:4,0.117382960080505)
--(axis cs:4,0.117608642350864);

\path [draw=lightslategray, semithick]
(axis cs:5,0.0950501900395601)
--(axis cs:5,0.0952217898725559);

\addplot [semithick, darkorange, mark=triangle*, mark size=3, mark options={solid,rotate=180}, only marks]
table {%
0.1 1.49284202672099
0.3 1.1340353942273
0.5 0.929365958342983
0.7 0.787890971067082
1 0.64257447476879
2 0.395037725950275
3 0.284115998771703
4 0.221582514466226
5 0.1815069820328
};
\addlegendentry{$c=1$, Oracle-MSE}
\addplot [semithick, green, mark=triangle*, mark size=3, mark options={solid}, only marks]
table {%
0.1 1.51486902827934
0.3 1.13980907154679
0.5 0.932738025271859
0.7 0.789905688652562
1 0.643848538028106
2 0.395517707999805
3 0.284394761685264
4 0.221785073713861
5 0.181657689270423
};
\addlegendentry{$c=1$, RIE-MSE}
\addplot [semithick, red, mark=triangle*, mark size=3, mark options={solid,rotate=180}, only marks]
table {%
0.1 1.03477449072083
0.3 0.75520329940032
0.5 0.594348779069167
0.7 0.488614692956986
1 0.383506368395735
2 0.220078924247147
3 0.153310629097282
4 0.117404335176733
5 0.0950646367229016
};
\addlegendentry{$c=0.5$, Oracle-MSE}
\addplot [semithick, royalblue, mark=triangle*, mark size=3, mark options={solid}, only marks]
table {%
0.1 1.04916622983448
0.3 0.758838201298873
0.5 0.59608359794192
0.7 0.489706531389481
1 0.384155795843349
2 0.220319393338496
3 0.153447198838982
4 0.117495801215685
5 0.095135989956058
};
\addlegendentry{$c=0.5$, RIE-MSE}
\draw (axis cs:0.2,1.49486902827934) node[
  scale=0.6,
  anchor=base west,
  text=black,
  rotate=0.0
]{1.48\%};
\draw (axis cs:0.2,1.2) node[
  scale=0.6,
  anchor=base west,
  text=black,
  rotate=0.0
]{0.51\%};
\draw (axis cs:0.4,0.99) node[
  scale=0.6,
  anchor=base west,
  text=black,
  rotate=0.0
]{0.36\%};
\draw (axis cs:0.6,0.85) node[
  scale=0.6,
  anchor=base west,
  text=black,
  rotate=0.0
]{0.26\%};
\draw (axis cs:0.9,0.69) node[
  scale=0.6,
  anchor=base west,
  text=black,
  rotate=0.0
]{0.2\%};
\draw (axis cs:1.8,0.45) node[
  scale=0.6,
  anchor=base west,
  text=black,
  rotate=0.0
]{0.12\%};
\draw (axis cs:2.8,0.35) node[
  scale=0.6,
  anchor=base west,
  text=black,
  rotate=0.0
]{0.1\%};
\draw (axis cs:3.8,0.29) node[
  scale=0.6,
  anchor=base west,
  text=black,
  rotate=0.0
]{0.09\%};
\draw (axis cs:4.7,0.25) node[
  scale=0.6,
  anchor=base west,
  text=black,
  rotate=0.0
]{0.08\%};
\draw (axis cs:0,0.954774490720832) node[
  scale=0.6,
  anchor=base west,
  text=black,
  rotate=0.0
]{1.39\%};
\draw (axis cs:0.12,0.67520329940032) node[
  scale=0.6,
  anchor=base west,
  text=black,
  rotate=0.0
]{0.48\%};
\draw (axis cs:0.32,0.514348779069168) node[
  scale=0.6,
  anchor=base west,
  text=black,
  rotate=0.0
]{0.29\%};
\draw (axis cs:0.52,0.408614692956986) node[
  scale=0.6,
  anchor=base west,
  text=black,
  rotate=0.0
]{0.22\%};
\draw (axis cs:0.82,0.303506368395735) node[
  scale=0.6,
  anchor=base west,
  text=black,
  rotate=0.0
]{0.17\%};
\draw (axis cs:1.82,0.140078924247147) node[
  scale=0.6,
  anchor=base west,
  text=black,
  rotate=0.0
]{0.11\%};
\draw (axis cs:2.82,0.0733106290972819) node[
  scale=0.6,
  anchor=base west,
  text=black,
  rotate=0.0
]{0.09\%};
\draw (axis cs:3.82,0.0374043351767332) node[
  scale=0.6,
  anchor=base west,
  text=black,
  rotate=0.0
]{0.08\%};
\draw (axis cs:4.7,0.0150646367229016) node[
  scale=0.6,
  anchor=base west,
  text=black,
  rotate=0.0
]{0.08\%};
\end{axis}

\end{tikzpicture}
    \caption{  Sum of rank-one factors noise matrix,\\ $N=1000, M=2000, L=c N$}
    \label{fig:Poisson-noise}
\end{subfigure}
    \caption{  Performance of the algorithmic RIE based on \eqref{rect-RIE} as compared to the oracle one. Signal matrix $\bS \in \bR^{N \times M}$ has i.i.d. Gaussian entries of variance $\nicefrac{1}{N}$. Results are averaged over 10 runs (error bars are invisible). Average relative error is also reported. In both examples, the Hilbert transform of the observation is computed numerically using Cauchy kernel method in \cite{potters2020first}. }
    \label{fig:General-noise}
\end{figure}

\subsection{Gaussian noise}
\subsubsection{Validity of Theorem \ref{trace-relation}}\label{numerical-thm-check}
In Fig. \ref{fig:Thm-check}, we numerically verify Theorem \ref{trace-relation} and check the behavior of the error term $\epsilon_N$. For simplicity, we set the SNR parameter to one, $\lambda = 1$. In Fig. \ref{rel-error}, the relative error is plotted,
\begin{equation}\label{rel-error-tr-relation}
    \frac{| \epsilon_N | }{| L(z) |} = \frac{ \bigg| L(z) - \Big[ G(z) \big(z^2 +1- \frac{1}{\alpha_0} \big) - z^2 G^2(z)  -1 \Big] \bigg|}{| L(z) |}
\end{equation}
for the case of a signal matrix with i.i.d. Gaussian entries of variance $\nicefrac{1}{N}$. In Fig. \ref{error-term}, the behavior of the error term is depicted, which verifies the conjecture stated in remark \ref{error-term-imz}, namely that the error is controlled by $\big( N | \im z | \big)^{-1}$.

\begin{figure}
\begin{subfigure}[t]{.48\textwidth}
    \centering
    % This file was created with tikzplotlib v0.10.1.
\begin{tikzpicture}[scale=0.7]

\definecolor{darkgray176}{RGB}{176,176,176}
\definecolor{darkorange25512714}{RGB}{255,127,14}
\definecolor{steelblue31119180}{RGB}{31,119,180}

\begin{axis}[
legend cell align={left},
legend style={fill opacity=0.8, draw opacity=1, text opacity=1, draw=white!80!black},
tick align=outside,
tick pos=left,
x grid style={darkgray176},
xlabel={$N$},
xmin=95, xmax=1005,
xtick style={color=black},
y grid style={darkgray176},
ylabel={Relative Error},
scaled y ticks=false,
yticklabel style={
  /pgf/number format/precision=3,
  /pgf/number format/fixed},
ymin=0.0124054446703834, ymax=0.107965314389338,
ytick style={color=black}
]
\path [draw=blue, fill=blue, opacity=0.1]
(axis cs:100,0.0684509735505266)
--(axis cs:100,0.0541992853698559)
--(axis cs:200,0.037191973717933)
--(axis cs:300,0.0310885882145584)
--(axis cs:400,0.02930396131263)
--(axis cs:500,0.0246807710004876)
--(axis cs:550,0.0246038642909233)
--(axis cs:600,0.0227796337693245)
--(axis cs:650,0.0227880768850582)
--(axis cs:700,0.0185331928993448)
--(axis cs:750,0.021162755151048)
--(axis cs:800,0.0201355881598352)
--(axis cs:850,0.018785657095846)
--(axis cs:900,0.0193711265496177)
--(axis cs:950,0.0183327402234124)
--(axis cs:1000,0.016749075112154)
--(axis cs:1000,0.0204627295009764)
--(axis cs:1000,0.0204627295009764)
--(axis cs:950,0.0222790260482745)
--(axis cs:900,0.0236150285237229)
--(axis cs:850,0.0227991284463493)
--(axis cs:800,0.0248086978916992)
--(axis cs:750,0.0257128237612763)
--(axis cs:700,0.0228904557430522)
--(axis cs:650,0.0277675119257918)
--(axis cs:600,0.0277991031443198)
--(axis cs:550,0.0296625022557581)
--(axis cs:500,0.0304651282864515)
--(axis cs:400,0.0359017409440501)
--(axis cs:300,0.038651525901367)
--(axis cs:200,0.0451476454836777)
--(axis cs:100,0.0684509735505266)
--cycle;

\path [draw=red, fill=red, opacity=0.1]
(axis cs:100,0.103621683947567)
--(axis cs:100,0.0839048487493169)
--(axis cs:200,0.0521479339681824)
--(axis cs:300,0.0466507670143497)
--(axis cs:400,0.0382001356067365)
--(axis cs:500,0.0354545843568907)
--(axis cs:550,0.0336866030180983)
--(axis cs:600,0.0293606282010053)
--(axis cs:650,0.028807310920765)
--(axis cs:700,0.0279020573926674)
--(axis cs:750,0.026809714740734)
--(axis cs:800,0.027466926064438)
--(axis cs:850,0.0242512695476647)
--(axis cs:900,0.0271178290770413)
--(axis cs:950,0.0246403243940644)
--(axis cs:1000,0.0255189776370957)
--(axis cs:1000,0.0307844386542427)
--(axis cs:1000,0.0307844386542427)
--(axis cs:950,0.0303112376511195)
--(axis cs:900,0.0329217364524386)
--(axis cs:850,0.0301621788140527)
--(axis cs:800,0.0336841539285233)
--(axis cs:750,0.0329946677290305)
--(axis cs:700,0.0334801542676869)
--(axis cs:650,0.0356719389225122)
--(axis cs:600,0.0359688681230622)
--(axis cs:550,0.0419167135575304)
--(axis cs:500,0.0445564103201718)
--(axis cs:400,0.0470372654703727)
--(axis cs:300,0.0566216443830082)
--(axis cs:200,0.0629219703411487)
--(axis cs:100,0.103621683947567)
--cycle;

\addplot [semithick, steelblue31119180]
table {%
100 0.0613251294601912
200 0.0411698096008054
300 0.0348700570579627
400 0.03260285112834
500 0.0275729496434696
550 0.0271331832733407
600 0.0252893684568222
650 0.025277794405425
700 0.0207118243211985
750 0.0234377894561622
800 0.0224721430257672
850 0.0207923927710976
900 0.0214930775366703
950 0.0203058831358434
1000 0.0186059023065652
};
\addlegendentry{$\alpha_0 = \nicefrac{1}{2}$}
\addplot [semithick, darkorange25512714]
table {%
100 0.0937632663484419
200 0.0575349521546655
300 0.0516362056986789
400 0.0426187005385546
500 0.0400054973385312
550 0.0378016582878144
600 0.0326647481620337
650 0.0322396249216386
700 0.0306911058301771
750 0.0299021912348822
800 0.0305755399964806
850 0.0272067241808587
900 0.0300197827647399
950 0.0274757810225919
1000 0.0281517081456692
};
\addlegendentry{$\alpha_0 = 1$}
\end{axis}

\end{tikzpicture}
    \caption{Relative error}
    \label{rel-error}
\end{subfigure}
\hfill
\begin{subfigure}[t]{.48\textwidth}
    \centering
    % This file was created with tikzplotlib v0.10.1.
\begin{tikzpicture}[scale=0.7]

\definecolor{darkgray176}{RGB}{176,176,176}
\definecolor{darkorange25512714}{RGB}{255,127,14}
\definecolor{steelblue31119180}{RGB}{31,119,180}

\begin{axis}[
legend cell align={left},
legend style={fill opacity=0.8, draw opacity=1, text opacity=1, draw=white!80!black},
tick align=outside,
tick pos=left,
x grid style={darkgray176},
xlabel={$N$},
xmin=95, xmax=1005,
xtick style={color=black},
y grid style={darkgray176},
ylabel={$| \epsilon_N |$},
scaled y ticks=false,
yticklabel style={
  /pgf/number format/precision=3,
  /pgf/number format/fixed},
ymin=0.00975137034649519, ymax=0.0495608900560206,
ytick style={color=black}
]
\path [draw=blue, fill=blue, opacity=0.1]
(axis cs:100,0.0445232168006955)
--(axis cs:100,0.035421876966326)
--(axis cs:200,0.024878054352905)
--(axis cs:300,0.0209361090045763)
--(axis cs:400,0.0199324877595992)
--(axis cs:500,0.0168454078241581)
--(axis cs:550,0.0168492021774531)
--(axis cs:600,0.0155956113848722)
--(axis cs:650,0.0156367624481691)
--(axis cs:700,0.0127490064805524)
--(axis cs:750,0.0145774416573974)
--(axis cs:800,0.0138604468589992)
--(axis cs:850,0.0129386885815661)
--(axis cs:900,0.0133608682330662)
--(axis cs:950,0.0126404341121142)
--(axis cs:1000,0.0115608939696554)
--(axis cs:1000,0.0141285600267789)
--(axis cs:1000,0.0141285600267789)
--(axis cs:950,0.0153528593695068)
--(axis cs:900,0.0162920837551136)
--(axis cs:850,0.0156949720278686)
--(axis cs:800,0.0170709669969439)
--(axis cs:750,0.0177159617834928)
--(axis cs:700,0.0157544846132649)
--(axis cs:650,0.019050365402043)
--(axis cs:600,0.0190192532467848)
--(axis cs:550,0.020306497489097)
--(axis cs:500,0.0207758744112006)
--(axis cs:400,0.0244038936523728)
--(axis cs:300,0.0259822434178814)
--(axis cs:200,0.0301480960814378)
--(axis cs:100,0.0445232168006955)
--cycle;

\path [draw=red, fill=red, opacity=0.1]
(axis cs:100,0.0477513664328604)
--(axis cs:100,0.0387379680767911)
--(axis cs:200,0.0247304088299934)
--(axis cs:300,0.0223149539010048)
--(axis cs:400,0.0184014912730989)
--(axis cs:500,0.0171179764180098)
--(axis cs:550,0.0162770780672544)
--(axis cs:600,0.0142451117135915)
--(axis cs:650,0.0139890390766946)
--(axis cs:700,0.0135632629438915)
--(axis cs:750,0.0130287081028991)
--(axis cs:800,0.0133397719215611)
--(axis cs:850,0.0117939516644908)
--(axis cs:900,0.0132178573498309)
--(axis cs:950,0.0120347795872168)
--(axis cs:1000,0.0124690167417364)
--(axis cs:1000,0.0150457126260265)
--(axis cs:1000,0.0150457126260265)
--(axis cs:950,0.0148170308188641)
--(axis cs:900,0.0160386523623824)
--(axis cs:850,0.0146414776494368)
--(axis cs:800,0.0163341141672044)
--(axis cs:750,0.0160097361936255)
--(axis cs:700,0.0162913708594708)
--(axis cs:650,0.0173766679025505)
--(axis cs:600,0.0174433060362718)
--(axis cs:550,0.0202297806937454)
--(axis cs:500,0.0214687373152963)
--(axis cs:400,0.0226558604402363)
--(axis cs:300,0.0270740894775627)
--(axis cs:200,0.0298734619697845)
--(axis cs:100,0.0477513664328604)
--cycle;

\addplot [semithick, steelblue31119180]
table {%
100 0.0399725468835108
200 0.0275130752171714
300 0.0234591762112289
400 0.022168190705986
500 0.0188106411176794
550 0.018577849833275
600 0.0173074323158285
650 0.0173435639251061
700 0.0142517455469086
750 0.0161467017204451
800 0.0154657069279716
850 0.0143168303047173
900 0.0148264759940899
950 0.0139966467408105
1000 0.0128447269982172
};
\addlegendentry{$\alpha_0 = \nicefrac{1}{2}$}
\addplot [semithick, darkorange25512714]
table {%
100 0.0432446672548258
200 0.027301935399889
300 0.0246945216892837
400 0.0205286758566676
500 0.0192933568666531
550 0.0182534293804999
600 0.0158442088749316
650 0.0156828534896225
700 0.0149273169016811
750 0.0145192221482623
800 0.0148369430443827
850 0.0132177146569638
900 0.0146282548561066
950 0.0134259052030405
1000 0.0137573646838814
};
\addlegendentry{$\alpha_0 = 1$}
\addplot [semithick, black]
table {%
100 0.04
200 0.0282842712474619
300 0.023094010767585
400 0.02
500 0.0178885438199983
550 0.0170560573084488
600 0.0163299316185545
650 0.0156892908110547
700 0.0151185789203691
750 0.0146059348668044
800 0.014142135623731
850 0.0137198868114007
900 0.0133333333333333
950 0.012977713690461
1000 0.0126491106406735
};
\addlegendentry{$0.4 N^{-\nicefrac{1}{2}}$}
\end{axis}

\end{tikzpicture}
    \caption{  Error term}
    \label{error-term}
\end{subfigure}
    \caption{  Validity of the estimation \eqref{L(z)-approx}. Plots are average of 100 experiments and $95\%$ confidence interval is also depicted. The signal matrix $\bS \in \bR^{N \times M}$ has i.i.d. Gaussian entries of variance $\nicefrac{1}{N}$, and $M = \nicefrac{N}{\alpha_0}$. the expressions are evaluated for $z = 1 + \nicefrac{\ci}{\sqrt{N}}$. In the left panel, the relative error \eqref{rel-error-tr-relation} is plotted for various values of $N$. On the same simulations, the error term is plotted in the right panel which behaves as $N^{-\nicefrac{1}{2}}$ which matches with the conjecture of remark \ref{error-term-imz}.}
    \label{fig:Thm-check}
\end{figure}
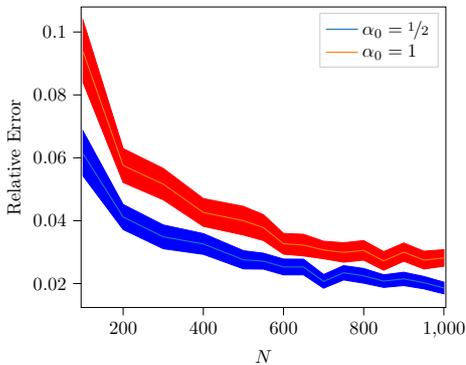
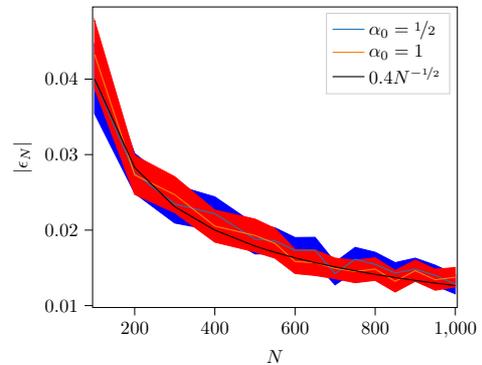

\subsubsection{Gaussian Signal}
If we consider the signal matrix to have i.i.d. Gaussian entries of variance $\nicefrac{1}{N}$, then each entry of $\bY$ can be viewed as an independent scalar AWGN channel. For this scalar channel, the MMSE equals $\frac{1}{N} \frac{1}{1 + \lambda}$ \cite{guo2005mutual}. Therefore, the (normalized) MMSE of the matrix problem is $\frac{M}{N} \frac{1}{1 + \lambda}\to \frac{1}{\alpha} \frac{1}{1 + \lambda}$ for $N \to \infty$. As a sanity check of Statement \ref{Gaus-MMSE-statement}, using the fact that $\mu_Y$ is the Marchenko-Pastur (MP) law rescaled with $\sqrt{\lambda + 1}$, we compute the MMSE analytically (with the help of \textit{Mathematica} \cite{Mathematica}) and find it equal to $\frac{1}{\alpha} \frac{1}{1 + \lambda}$. 
In Fig. \ref{fig:Gsign-Gnois}, MSE of RIE is compared to the theoretical MMSE for $\alpha = 1, \alpha = \nicefrac{1}{2}$. Note that, for this example, we use the RIE \eqref{rect-RIE}, and the Hilbert transform used in RIE is the exact Hilbert transform of the symmetrization of MP law rescaled with $\sqrt{\lambda + 1}$, which is $\sH_{\bar{\mu}_Y}(x) = \frac{x}{2 + 2 \lambda} - \frac{1-\alpha}{2 \alpha x}$.

\begin{figure}
\centering
\begin{minipage}[t]{.48\textwidth}
  \centering
  \input{Figures/G-G-MSE}
  \caption{Performance of the RIE \eqref{rect-RIE} for the Gaussian signal and noise. Signal and noise matrices $\bS, \bZ \in \bR^{N \times M}$ have i.i.d. Gaussian entries of variance $\nicefrac{1}{N}$. The MMSE is plotted for two aspect ratios $\alpha = 1, \nicefrac{1}{2}$, and the RIE \eqref{rect-RIE} is applied to $N = 1000$. Results are averaged over 10 runs.}
  \label{fig:Gsign-Gnois}
\end{minipage}%
\hfill
\begin{minipage}[t]{.48\textwidth}
  \centering
  % This file was created with tikzplotlib v0.10.1.
\begin{tikzpicture}[scale = 0.7]

\definecolor{darkgray176}{RGB}{176,176,176}
\definecolor{lightslategray}{RGB}{119,136,153}
\definecolor{royalblue}{RGB}{65,105,225}

\begin{axis}[
tick align=outside,
tick pos=left,
x grid style={darkgray176},
xlabel={$\epsilon$},
xmin=0.06, xmax=0.94,
xtick style={color=black},
y grid style={darkgray176},
ylabel={RIE-MSE - Oracle-MSE},
scaled y ticks=false,
yticklabel style={
  /pgf/number format/precision=3,
  /pgf/number format/fixed},
ymin=-0.000635500903793554, ymax=0.0272337103259102,
ytick style={color=black}
]
\path [draw=lightslategray, semithick]
(axis cs:0.1,0.0255622413762111)
--(axis cs:0.1,0.0259669279972873);

\path [draw=lightslategray, semithick]
(axis cs:0.2,0.00838840890779667)
--(axis cs:0.2,0.00867243956659257);

\path [draw=lightslategray, semithick]
(axis cs:0.3,0.00274818097981918)
--(axis cs:0.3,0.00289834183909538);

\path [draw=lightslategray, semithick]
(axis cs:0.4,0.000980799474877376)
--(axis cs:0.4,0.00108894189386074);

\path [draw=lightslategray, semithick]
(axis cs:0.5,0.000631281424829345)
--(axis cs:0.5,0.000682622071804239);

\path [draw=lightslategray, semithick]
(axis cs:0.6,0.000819266922011323)
--(axis cs:0.6,0.000922408151656081);

\path [draw=lightslategray, semithick]
(axis cs:0.7,0.00217492200526696)
--(axis cs:0.7,0.00257473153812355);

\path [draw=lightslategray, semithick]
(axis cs:0.8,0.00679902858218962)
--(axis cs:0.8,0.0077849741620709);

\path [draw=lightslategray, semithick]
(axis cs:0.9,0.0218470244135023)
--(axis cs:0.9,0.0240199954916553);

\addplot [semithick, royalblue, mark=triangle*, mark size=3, mark options={solid}, only marks]
table {%
0.1 0.0257645846867492
0.2 0.00853042423719462
0.3 0.00282326140945728
0.4 0.00103487068436906
0.5 0.000656951748316792
0.6 0.000870837536833702
0.7 0.00237482677169526
0.8 0.00729200137213026
0.9 0.0229335099525788
};
\end{axis}

\end{tikzpicture}
  \caption{Performance of the RIE using \eqref{optimal-sv-approx-alg} for the Gaussian signal. The formula \eqref{optimal-sv-approx-alg} is used to estimate the optimal singular values with $z = \gamma_j + \ci N^{-\epsilon}$. RIE is applied to $N = 1000, M=2000, \lambda = 2$, and results are averaged over 10 runs.}
  \label{fig:imz}
\end{minipage}
\end{figure}

In Fig. \ref{fig:imz}, we investigate the performance of the RIE using the relation \eqref{optimal-sv-approx-alg} for various values of the imaginary part of $z$. In this plot, the difference of the MSEs of the RIE and the oracle estimator for the Gaussian signal and noise matrices is depicted. The RIE is applied with $z = \gamma_j + \ci N^{-\epsilon}$. A few remarks about this plot are in order. First, it supports the conjecture stated in remark \ref{error-term-imz} that the error term in \eqref{trace-relation} is controlled by $\big( N | \im z | \big)^{-1}$. Moreover, we can see that as the imaginary part of $z$ increases ($\epsilon$ decreases) the difference increases. For this regime, the error term in \eqref{trace-relation} becomes small, however the approximation \eqref{optimal-sv-approx-alg} of the exact formula \eqref{exact-RND} is inaccurate. On the other hand, for $z$ with small imaginary part, this approximation is more accurate, but the error of the estimation in \eqref{trace-relation} becomes large. 

\subsubsection{Signal with sparse spectrum}
The signal matrix $\bS \in \bR^{N \times M}$ is constructed as $\bS = \bU [ {\rm diag}(\sigma_1, \cdots, \sigma_N), \mathbf{0}_{M-N} ] \bV^\intercal$, where $\bU \in \bR^{N \times N}, \bV \in \bR^{M \times M}$ are independent Haar distributed matrices, and the singular values $\sigma_1, \cdots, \sigma_N$ are independent Bernoulli random variables, $\mu_S = p \delta_0 + (1-p) \delta_{+1}$ for $ 0 \leq p \leq 1$. In Fig. \ref{fig:S-MSE} MSE of RIE is compared to the MSE of oracle estimator with $\alpha = \nicefrac{1}{2}$ for $p=0.2,0.9$. We observe that, in the high-sparsity regime $p=0.9$, the model behaves like finite-rank signal and the MSE is close to the rank-one MMSE computed in \cite{miolane2017fundamental}.

\begin{figure}
  \begin{minipage}[c]{0.6\textwidth}
    \centering
    \input{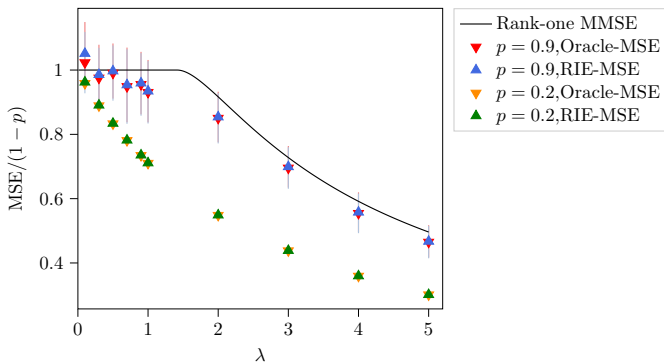}
  \end{minipage}\hfill
  \begin{minipage}[c]{0.35\textwidth}
    \caption{   Signal with Bernoulli spectrum. MSE is normalized by the norm of the signal, $1-p$. The RIE is applied to $N=1000, M =2000$, and the results are averaged over 10 runs (error bars might be invisible).}
    \label{fig:S-MSE}
  \end{minipage}
\end{figure}

% \begin{SCfigure}
%     \centering
%     \input{Figures/S-G-MSE}
%     \hspace{10pt}
%     \caption{  \textnormal{\small Signal with Bernoulli spectrum. MSE is normalized by the norm of the signal, $1-p$. The RIE is applied to $N=1000, M =2000$, and the results are averaged over 10 runs (error bars might be invisible).}}
    
% \end{SCfigure}

\subsection{Non-rotational invariant signal distribution}
We consider $\bS$ to have i.i.d. entries from the Bernoulli-Rademacher distribution,
\begin{equation*}
    S_{i,j} = \begin{cases}
        +\frac{1}{\sqrt{N}} &\text{with probability } \frac{1-p}{2} \\
        0 &\text{with probability } p \\
        -\frac{1}{\sqrt{N}} &\text{with probability } \frac{1-p}{2} 
    \end{cases}, \quad \quad \forall \quad 1\leq i \leq N, \quad 1\leq j \leq M
\end{equation*}
With normalization $\nicefrac{1}{\sqrt{N}}$, the spectrum of $\bS$ does not grow with the dimension and has a finite support, thus we can apply our estimator to reconstruct $\bS$. \textit{Note that the prior of $\bS$ is not rotationally invariant, and neither the oracle estimator nor the RIE are optimal}. In Fig. \ref{fig:nri-signal}, the performance of the RIE is compared with the oracle estimator for two cases of noise priors. Note that under Gaussian noise, the MMSE can be computed simply by considering the MMSE of scalar channel, and the MMSE is also plotted. We can see that RIE, although it is sub-optimal, can give a non-trivial estimate of the signal for non-rotationally invariant priors. 

\begin{figure}
\centering
\begin{subfigure}[t]{.58\textwidth}
    \centering
    \captionsetup{justification=centering}
    \input{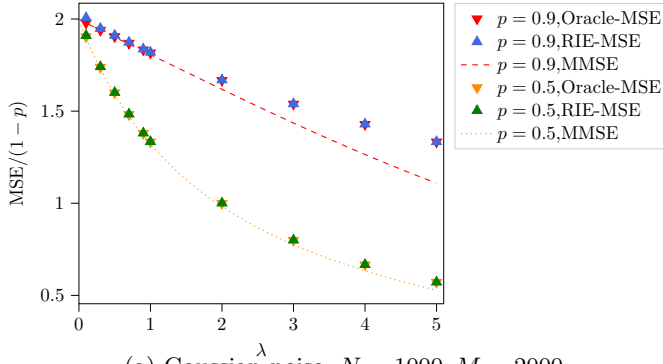}
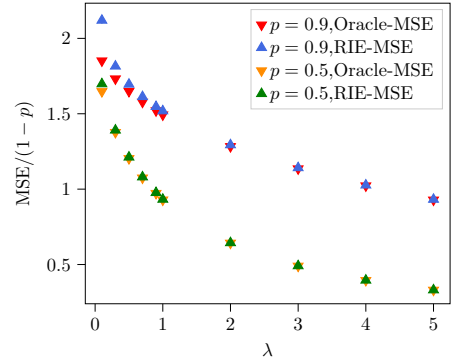
    \vspace{-9pt} 
    % \hspace{-0pt}
   \caption{Gaussian noise, $ N =1000, M = 2000$}
    \label{fig:nri-signal-Gnoise}
\end{subfigure}
\hfill
\begin{subfigure}[t]{.41\textwidth}
    \centering
    \captionsetup{justification=centering}
    % This file was created with tikzplotlib v0.10.1.
\begin{tikzpicture}[scale=0.7]

\definecolor{darkgray176}{RGB}{176,176,176}
\definecolor{darkorange}{RGB}{255,140,0}
\definecolor{green}{RGB}{0,128,0}
\definecolor{lightcoral}{RGB}{240,128,128}
\definecolor{lightgreen}{RGB}{144,238,144}
\definecolor{lightslategray}{RGB}{119,136,153}
\definecolor{moccasin}{RGB}{255,228,181}
\definecolor{royalblue}{RGB}{65,105,225}

\begin{axis}[
legend cell align={left},
legend style={fill opacity=0.8, draw opacity=1, text opacity=1, draw=white!80!black},
tick align=outside,
tick pos=left,
x grid style={darkgray176},
xlabel={$\lambda$},
xmin=-0.145, xmax=5.245,
xtick style={color=black},
y grid style={darkgray176},
ylabel={${\rm MSE}/(1-p)$},
ymin=0.239791347727806, ymax=2.2333098180409,
ytick style={color=black}
]
\path [draw=lightcoral, semithick]
(axis cs:0.1,1.84806313807419)
--(axis cs:0.1,1.85466340159406);

\path [draw=lightcoral, semithick]
(axis cs:0.3,1.72930978403531)
--(axis cs:0.3,1.73504488348262);

\path [draw=lightcoral, semithick]
(axis cs:0.5,1.64469170423863)
--(axis cs:0.5,1.65326703419909);

\path [draw=lightcoral, semithick]
(axis cs:0.7,1.57432986339604)
--(axis cs:0.7,1.58106758885079);

\path [draw=lightcoral, semithick]
(axis cs:0.9,1.51865057196703)
--(axis cs:0.9,1.52212822914458);

\path [draw=lightcoral, semithick]
(axis cs:1,1.49053246258225)
--(axis cs:1,1.49566115300128);

\path [draw=lightcoral, semithick]
(axis cs:2,1.28126585966444)
--(axis cs:2,1.28438483012243);

\path [draw=lightcoral, semithick]
(axis cs:3,1.13423668731568)
--(axis cs:3,1.13660314679699);

\path [draw=lightcoral, semithick]
(axis cs:4,1.0208473796835)
--(axis cs:4,1.02399544422728);

\path [draw=lightcoral, semithick]
(axis cs:5,0.927320354036782)
--(axis cs:5,0.929562655863614);

\path [draw=lightslategray, semithick]
(axis cs:0.1,2.0960909069633)
--(axis cs:0.1,2.14269534211758);

\path [draw=lightslategray, semithick]
(axis cs:0.3,1.8025372810316)
--(axis cs:0.3,1.82788447650402);

\path [draw=lightslategray, semithick]
(axis cs:0.5,1.68784257040138)
--(axis cs:0.5,1.70260045568701);

\path [draw=lightslategray, semithick]
(axis cs:0.7,1.60599394704003)
--(axis cs:0.7,1.61962989287663);

\path [draw=lightslategray, semithick]
(axis cs:0.9,1.54353206943004)
--(axis cs:0.9,1.55053071597202);

\path [draw=lightslategray, semithick]
(axis cs:1,1.51387346247048)
--(axis cs:1,1.5228937241887);

\path [draw=lightslategray, semithick]
(axis cs:2,1.29080631992144)
--(axis cs:2,1.29573060093575);

\path [draw=lightslategray, semithick]
(axis cs:3,1.13997303744835)
--(axis cs:3,1.14268380878163);

\path [draw=lightslategray, semithick]
(axis cs:4,1.02486311791412)
--(axis cs:4,1.02785579866996);

\path [draw=lightslategray, semithick]
(axis cs:5,0.930660640349758)
--(axis cs:5,0.932889603142321);

\path [draw=moccasin, semithick]
(axis cs:0.1,1.64648900216539)
--(axis cs:0.1,1.6493200434523);

\path [draw=moccasin, semithick]
(axis cs:0.3,1.37561288936589)
--(axis cs:0.3,1.3779598069402);

\path [draw=moccasin, semithick]
(axis cs:0.5,1.20236741809203)
--(axis cs:0.5,1.20443373137243);

\path [draw=moccasin, semithick]
(axis cs:0.7,1.07371352337406)
--(axis cs:0.7,1.07562669461878);

\path [draw=moccasin, semithick]
(axis cs:0.9,0.971952946688706)
--(axis cs:0.9,0.973230065728519);

\path [draw=moccasin, semithick]
(axis cs:1,0.927712928050845)
--(axis cs:1,0.930098202443048);

\path [draw=moccasin, semithick]
(axis cs:2,0.642181383202648)
--(axis cs:2,0.643089093915395);

\path [draw=moccasin, semithick]
(axis cs:3,0.489724615745401)
--(axis cs:3,0.490542934792888);

\path [draw=moccasin, semithick]
(axis cs:4,0.394822235552243)
--(axis cs:4,0.395420145097881);

\path [draw=moccasin, semithick]
(axis cs:5,0.330405823651129)
--(axis cs:5,0.330829230280066);

\path [draw=lightgreen, semithick]
(axis cs:0.1,1.69173935441155)
--(axis cs:0.1,1.70358856636673);

\path [draw=lightgreen, semithick]
(axis cs:0.3,1.38939448813347)
--(axis cs:0.3,1.39231260307877);

\path [draw=lightgreen, semithick]
(axis cs:0.5,1.21071608620336)
--(axis cs:0.5,1.21268044198642);

\path [draw=lightgreen, semithick]
(axis cs:0.7,1.07900109638142)
--(axis cs:0.7,1.08122054753846);

\path [draw=lightgreen, semithick]
(axis cs:0.9,0.975426398081232)
--(axis cs:0.9,0.977320388496806);

\path [draw=lightgreen, semithick]
(axis cs:1,0.930903480658417)
--(axis cs:1,0.93340313413334);

\path [draw=lightgreen, semithick]
(axis cs:2,0.643403642367753)
--(axis cs:2,0.64428870992009);

\path [draw=lightgreen, semithick]
(axis cs:3,0.490410319351976)
--(axis cs:3,0.491306255637673);

\path [draw=lightgreen, semithick]
(axis cs:4,0.395321971499603)
--(axis cs:4,0.395910103292868);

\path [draw=lightgreen, semithick]
(axis cs:5,0.330733969032234)
--(axis cs:5,0.331175668501517);

\addplot [semithick, red, mark=triangle*, mark size=3, mark options={solid,rotate=180}, only marks]
table {%
0.1 1.85136326983413
0.3 1.73217733375896
0.5 1.64897936921886
0.7 1.57769872612341
0.9 1.52038940055581
1 1.49309680779176
2 1.28282534489343
3 1.13541991705634
4 1.02242141195539
5 0.928441504950198
};
\addlegendentry{$p=0.9$,Oracle-MSE}
\addplot [semithick, royalblue, mark=triangle*, mark size=3, mark options={solid}, only marks]
table {%
0.1 2.11939312454044
0.3 1.81521087876781
0.5 1.6952215130442
0.7 1.61281191995833
0.9 1.54703139270103
1 1.51838359332959
2 1.29326846042859
3 1.14132842311499
4 1.02635945829204
5 0.93177512174604
};
\addlegendentry{$p=0.9$,RIE-MSE}
\addplot [semithick, darkorange, mark=triangle*, mark size=3, mark options={solid,rotate=180}, only marks]
table {%
0.1 1.64790452280884
0.3 1.37678634815305
0.5 1.20340057473223
0.7 1.07467010899642
0.9 0.972591506208612
1 0.928905565246946
2 0.642635238559021
3 0.490133775269144
4 0.395121190325062
5 0.330617526965597
};
\addlegendentry{$p=0.5$,Oracle-MSE}
\addplot [semithick, green, mark=triangle*, mark size=3, mark options={solid}, only marks]
table {%
0.1 1.69766396038914
0.3 1.39085354560612
0.5 1.21169826409489
0.7 1.08011082195994
0.9 0.976373393289019
1 0.932153307395879
2 0.643846176143922
3 0.490858287494824
4 0.395616037396235
5 0.330954818766875
};
\addlegendentry{$p=0.5$,RIE-MSE}
\end{axis}

\end{tikzpicture}
    \caption{  Sum of rank-one factors noise matrix,\\ $N=1000, M=2000, c = 1$}
    \label{fig:nri-signal-Poisson-noise}
\end{subfigure}
    \caption{  Performance of the  RIE ( \eqref{rect-RIE} and \eqref{optimal-sv-approx-alg})  and oracle estimator for non-rotational invariant signal. Signal matrix $\bS \in \bR^{N \times M}$ has i.i.d. Bernoulli-Rademacher entries (divided by $\nicefrac{1}{\sqrt{N}}$. Results are averaged over 10 runs (error bars are invisible).}
    \label{fig:nri-signal}
\end{figure}

\section{Analytical Derivations and Proofs}\label{Proof-details}
\subsection{Derivation sketch of the rectangular explicit RIE}\label{RIE-derivation-sec}
Let the SVD of the signal be $ \bS = \sum_{k=1}^N \sigma_k \bs_k^{(l)} {\bs_k^{(r)}}^\intercal $ where $\bs_k^{(r)}$/$\bs_k^{(l)}$ is the right/left singular vector of $\bS$ corresponding to the $k$-th singular value $\sigma_k$. From \eqref{Oracle-sv}, the optimal singular values of the oracle RIE can be written as:
\begin{equation*}
    \xi_j^*  = \bu_j^\intercal \bS \bv_j =  \sum_{k=1}^N \sigma_k \big(  \bu_j^\intercal \bs_k^{(l)}\big) \big(  \bv_j^\intercal \bs_k^{(r)} \big)
\end{equation*}
The main assumption is that in the large-$N$ limit, $\xi_j^*$'s can be approximated by the expectation, $\widehat{\xi_i^*} = \sum_{j=1}^N \sigma_j \Big \langle \big(  \bu_j^\intercal \bs_k^{(l)}\big) \big(  \bv_j^\intercal \bs_k^{(r)} \big) \Big \rangle$, where the expectation $\langle -\rangle$  is over the singular vectors of the observation $\bY$. 

Therefore, to compute the optimal singular vales, we need to find the overlap $\Big \langle \big(  \bu_j^\intercal \bs_k^{(l)}\big) \big(  \bv_j^\intercal \bs_k^{(r)} \big) \Big \rangle$ between singular vectors of $\bS$ and singular vectors of $\bY$. In what follows, we will see that (a rescaling of) this quantity can be expressed in terms of $j$-th singular value of $\bY$ and $k$-th singular value of $\bS$ and the limiting measures, indeed. Thus, we will use the notation $O(\gamma_j, \sigma_k) = N \Big \langle \big(  \bu_j^\intercal \bs_k^{(l)}\big) \big(  \bv_j^\intercal \bs_k^{(r)} \big) \Big \rangle$ in the following and write
\begin{equation}
    \widehat{\xi_j^*} = \frac{1}{N} \sum_{k =1}^N \sigma_k O(\gamma_j, \sigma_k)
\end{equation}
In the next section, we discuss how this overlap can be computed from the resolvent of the "Hermitized" version of $\bY$.

\subsubsection{Relation Between Overlap and the Resolvent}
Construct the symmetric matrix $\cY \in \bR^{(N+M) \times (N+M)}$ from the matrix $\bY$,
\begin{equation}\label{hermitization}
    \cY = \left[
\begin{array}{cc}
\mathbf{0}_{N\times N} & \bY \\
\bY^\intercal  & \mathbf{0}_{M\times M}
\end{array}
\right]
\end{equation}
By Theorem 7.3.3 in \cite{horn2012matrix}, the eigen-decomposition of $\cY$ reads: 
\begin{equation}
    \cY = \bW \left[
\begin{array}{ccc}
\rm{diag}(\gamma_1, \cdots, \gamma_N ) & \mathbf{0} & \mathbf{0}\\
\mathbf{0} & -\rm{diag}(\gamma_1, \cdots, \gamma_N ) &  \mathbf{0}\\
\mathbf{0} & \mathbf{0} & \mathbf{0} 
\end{array}
\right]  \bW^\intercal, \quad \bW = \left[
\begin{array}{ccc}
\hat{\bU}_Y & \hat{\bU}_Y &  \mathbf{0}_{N \times (M-N)} \\
\hat{\bV}_Y^{(1)} & -\hat{\bV}_Y^{(1)} &  \bV_Y^{(2)}
\end{array}
\right]
\label{eigen-cY}
\end{equation}
with $\bV_Y = \left[
\begin{array}{cc}
\bV_Y^{(1)} & \bV_Y^{(2)}
\end{array}
\right]$ in which $\bV_Y^{(1)} \in \bR^{M \times N}$, and $\hat{\bV}_Y^{(1)} = \frac{1}{\sqrt{2}} \bV_Y^{(1)}$, $\hat{\bU}_Y= \frac{1}{\sqrt{2}} \bU_Y$. Denote the eigenvectors of $\cY$ by $\bw_i \in \bR^{M+N}$, $i = 1, \dots,M+N$.
Define the resolvent of $\cY$ 
\begin{equation*}
    \bG_{\mathcal{Y}}(z) = \left[
\begin{array}{c}
z \bI - \cY 
\end{array}
\right]^{-1}
\end{equation*}
For $z = x - \ci \eta $, we have:
\begin{equation*}
     \bG_{\mathcal{Y}}(x- \ci \eta) = \sum_{k=1}^{2N} \frac{x + \ci \eta }{(x - \tilde{\gamma}_k)^2+\eta^2} \bw_k \bw^\intercal _k + \frac{x + \ci \eta}{x^2 + \eta^2} \sum_{k=2N+1}^{M+N} \bw_k \bw_k^\intercal
\end{equation*}
where $\tilde{\gamma}_k$ are the non-trivially zero eigenvalues of $\cY$, which are in fact the (signed) singular values of $\bY$, $\tilde{\gamma}_1 = \gamma_1, \hdots, \tilde{\gamma}_N=\gamma_N, \tilde{\gamma}_{N+1}= - \gamma_1, \hdots, \tilde{\gamma}_{2N}=-\gamma_N$.
Define set of vectors $\br_i, \bl_i \in \bR^{N+M}$ for $i=1,\dots,N$ as:
\begin{equation*}
    \br_i = \left[
\begin{array}{c}
\mathbf{0}_N \\
\bs_i^{(r)}
\end{array}
\right] \hspace{1cm}
\bl_i =\left[
\begin{array}{c}
\bs_i^{(l)} \\
\mathbf{0}_M
\end{array}
\right]
\end{equation*}
We have
\begin{equation}\label{RIE-deriv-1}
\begin{split}
    \br_j^\intercal  \big( {\rm Im}\, \bG_{\mathcal{Y}}(x - \ci \eta) \big) \bl_j &= \sum_{k=1}^{2N} \frac{\eta }{(x - \tilde{\gamma}_k)^2+\eta^2} (\br_j^\intercal \bw_k) (\bw^\intercal _k \bl_j)  + \frac{x + \ci \eta}{x^2 + \eta^2} \sum_{k=2N+1}^{M+N} (\br_j^\intercal \bw_k) (\bw^\intercal _k \bl_j) 
\end{split}
\end{equation}
Given the structure of $\bw_k$'s in \eqref{eigen-cY}:
\begin{equation*}
    (\br_j^\intercal \bw_k) (\bw^\intercal _k \bl_j) = \begin{cases}
        \frac{1}{2} \big(  \bu_k^\intercal \bs_j^{(l)}\big) \big(  \bv_k^\intercal \bs_j^{(r)} \big) \quad &{\rm for} \, \, 1 \leq k \leq N \\
        -\frac{1}{2} \big(  \bu_k^\intercal \bs_j^{(l)}\big) \big(  \bv_k^\intercal \bs_j^{(r)} \big) \quad &{\rm for} \, \, N + 1 \leq k \leq 2 N\\
        0 &{\rm for} \, \, 2 N + 1 \leq k \leq M + N
    \end{cases}
\end{equation*}
Taking an average over singular vectors of $\bY$ in \eqref{RIE-deriv-1}, we find:
\begin{equation}\label{RIE-deriv-2}
\begin{split}
    \br_j^\intercal  \big\langle {\rm Im}\, \bG_{\mathcal{Y}}(x - \ci \eta) \big\rangle \bl_j &= \frac{1}{N} \sum_{k=1}^{2N} \frac{\eta }{(x - \tilde{\gamma}_k)^2+\eta^2}  (-1)^{\mathbb{I}(k>N)} O(\gamma_k, \sigma_j) 
\end{split}
\end{equation}
Now, taking the limit $N \to \infty$, we obtain:
\begin{equation*}
    \br_j^\intercal  \big( {\rm Im}\, \bG_{\mathcal{Y}}(x - \ci \eta) \big) \bl_j \xrightarrow[]{N \to \infty} \int_\bR \frac{\eta}{(x - t)^2+ \eta^2} O(t,\sigma_j) \bar{\mu}_{Y}(t) \, dt
\end{equation*}
where $O(t,\sigma_j)$ is extended (continuously) to arbitrary values inside the support of $\bar{\mu}_Y$ (the symmetrized limiting singular value distribution of $\bY$) with the property that $O(-t, \sigma_j) = - O(t, \sigma_j)$. Sending $\eta \to 0$, we find the following formula valid in the large N limit:
\begin{equation}
    \br_j^\intercal  \big\langle {\rm Im}\, \bG_{\mathcal{Y}}(x - \ci \eta) \big\rangle  \bl_j \approx \pi \bar{\mu}_{Y}(x) O(x, \sigma_j) 
    \label{resolvent-overlap}
\end{equation}
Eq. \eqref{resolvent-overlap} is important because it enables us to investigate the overlap through the resolvent of $\cY$. In the next section, we derive a relation between this resolvent and the signal $\bS$ which will allow us to find a formula for the optimal singular values $\xi^*_i$'s in terms of the singular values of the observation matrix $\bY$. 
% Note that we are only interested in the case where $x$ is positive, and in the support of the  limiting singular value distribution of $\bY$, $\mu_{Y}$. Eq. \eqref{resolvent-overlap} gives us to access the overlap through computation of the resolvent of $\cY$. In the next section, we state a relation between this resolvent and the resolvent of the signal which enables us to find $\xi_i^*$'s.

\subsubsection{Resolvent Relation}
To derive a resolvent relation between te observation and the signal, we consider the model
\begin{equation}
   \bY = \bS + \bU \bZ \bV^\intercal 
   \label{RIE-observ-model}
\end{equation}
with $\bZ$ a fixed matrix with limiting singular value distribution $\mu_Z$, and $\bU \in \bR^{N \times N}, \bV \in \bR^{M \times M}$ random orthogonal matrices. Note that for convenience the SNR parameter has been absorbed into $\bS$, so to obtain the estimator for model \eqref{model}, this estimator should be divided by $\sqrt{\lambda}$ eventually.

In Appendix \ref{RIE-der}, we derive the relation \eqref{resolvent-relation} for the resolvent $\bG_{\mathcal{Y}}(z)$, in which $\langle . \rangle$ is the expectation w.r.t. the singular vectors of $\bY$, and $\bG_{S^\intercal  S}$ is the resolvent matrix of $\bS^\intercal  \bS$.

\begin{equation}
\begin{split}
\langle \bG_{\mathcal{Y}}(z) \rangle &= \Bigg \langle \left[
\begin{array}{cc}
z^{-1} \bI_N + z^{-1} \bY \bG_{Y^\intercal Y}(z^2) \bY^\intercal  & \bY \bG_{Y^\intercal Y}(z^2) \\
\bG_{Y^\intercal Y}(z^2) \bY^\intercal  & z \bG_{Y^\intercal Y}(z^2)
\end{array}
\right] \Bigg \rangle\\
&\hspace{-9pt}\approx\left[
\begin{array}{cc}
(z-\zeta_1^*)^{-1} \bI_N + (z-\zeta_1^*)^{-1} \bS \bG_{S^\intercal  S} \big((z-\zeta_2^*)(z - \zeta_1^*)\big) \bS^\intercal  &   \bS \bG_{S^\intercal  S} \big((z-\zeta_2^*)(z - \zeta_1^*)\big)  \\
\bG_{S^\intercal  S}\big((z-\zeta_2^*)(z - \zeta_1^*)\big) \bS^\intercal  & (z - \zeta_1^*) \bG_{S^\intercal  S} \big((z-\zeta_2^*)(z - \zeta_1^*)\big)
\end{array}
\right]
\end{split}
\label{resolvent-relation}
\end{equation}
with
\begin{equation}
    \zeta^*_a = z\frac{Z(z)}{\mathcal{M}_{\mu_Y} \big( \frac{1}{z^2} \big) +1}, \quad
    \zeta^*_b  = \alpha z\frac{Z(z)}{\alpha \mathcal{M}_{\mu_Y}  \big( \frac{1}{z^2} \big) +1}, \quad \quad {\rm and} \quad
    Z(z) = \mathcal{C}^{(\alpha)}_{\mu_Z}\bigg(\frac{1}{z^2} T^{(\alpha)} \Big( \mathcal{M}_{\mu_Y}  \big( \frac{1}{z^2} \big) \Big)\bigg)
    \label{zeta_sol}
\end{equation}

As a sanity check, by considering the normalized trace of the first block on both sides of \eqref{resolvent-relation}, one can recover the free rectangular addition formula $\mathcal{C}^{(\alpha)}_{\mu_S}(u) + \mathcal{C}^{(\alpha)}_{\mu_Z}(u) =  \mathcal{C}^{(\alpha)}_{\mu_Y}(u) $ for $u = \frac{1}{z^2} T^{(\alpha)} \Big( \mathcal{M}_{\mu_Y}  \big( \frac{1}{z^2} \big) \Big)$ (see Appendix \ref{free-add-conv}).

\subsubsection{Overlap and Optimal Singular Values}
From the lower-left block of \eqref{resolvent-relation}, we get:
\begin{equation*}
\begin{split}
    \br_j^\intercal  \, \langle \bG_{\mathcal{Y}}(z) \rangle \, \bl_j &= {\bs_j^{(r)}}^\intercal  \bG_{S^\intercal  S}\big((z-\zeta_2^*)(z - \zeta_1^*)\big) \bS^\intercal  \bs_j^{(l)} \\
&= \frac{\sigma_j}{(z-\zeta_2^*)(z - \zeta_1^*) - {\sigma_j}^2}
\end{split}
\end{equation*}
and using \eqref{resolvent-overlap}, we find:
\begin{equation}
    O(\gamma, \sigma) \approx  \frac{1}{\pi \bar{\mu}_{Y}(\gamma)}  \lim_{z \to \gamma - \ci 0^+} {\rm Im}\, \frac{\sigma}{(z-\zeta_2^*)(z - \zeta_1^*) - {\sigma}^2}
    \label{Overlap-eq}
\end{equation}
where $\sigma$ is in the support of the limiting singular value distribution of $\bS$, $\mu_{S}$. In Fig. \ref{fig:Overlap} we illustrate on an example that theoretical predictions for the overlaps from \eqref{Overlap-eq} are in good agreement with numerical simulations.
% \begin{SCfigure}
%   \centering
%    \input{Figures/Overlap.tex}
%    \caption{\small Computation of the rescaled overlap. Both $\bS$ and $\bZ$ are $N \times M$ matrices with i.i.d. Gaussian entries of variance $1/N$, and $N/M = 1/4$. The simulation results are average of 1000 experiments with fixed $\bS$, and $N = 1000, M =4000$. Some of the simulation points are dropped for clarity.}
%    \label{fig:Overlap}
% \end{SCfigure}
% \begin{SCfigure}
%     \centering
%     \input{Figures/Overlap.tex}
%     \hspace{10pt}
%     \caption{ \textnormal{ Computation of the rescaled overlap. Both $\bS$ and $\bZ$ are $N \times M$ matrices with i.i.d. Gaussian entries of variance $1/N$, and $N/M = 0.25$. The simulation results are average of 1000 experiments with fixed $\bS$, and $N = 1000, M =4000$. Some of the simulation points are dropped for clarity.}}
%     \label{fig:Overlap}
% \end{SCfigure}

\begin{figure}
  \begin{minipage}[c]{0.5\textwidth}
    \centering
     \input{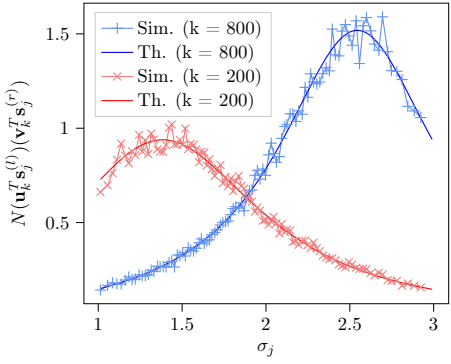}
  \end{minipage}\hfill
  \begin{minipage}[c]{0.5\textwidth}
    \caption{   Computation of the rescaled overlap. Both $\bS$ and $\bZ$ are $N \times M$ matrices with i.i.d. Gaussian entries of variance $1/N$, and $N/M = 0.25$. The simulation results are average of 1000 experiments with fixed $\bS$, and $N = 1000, M =4000$. Some of the simulation points are dropped for clarity.}
   \label{fig:Overlap}
  \end{minipage}
\end{figure}

The optimal estimator for singular values reads:
\begin{equation}\label{optxi}
\begin{split}
    \widehat{\xi_j^*} = \frac{1}{N} \sum_{k=1}^N \sigma_k O(\gamma_j, \sigma_k) &\approx \frac{1}{N \pi \bar{\mu}_{Y}(\gamma_j)} \lim_{z \to \gamma_j - \ci 0^+}  {\rm Im}\,\sum_{k=1}^N \frac{\sigma_k^2}{(z-\zeta_2^*)(z - \zeta_1^*) - \sigma_k^2}  \\
    &=  \frac{1}{N \pi \bar{\mu}_{Y}(\gamma_j)} \lim_{z \to \gamma_j - \ci 0^+}  {\rm Im} \, \Tr \bS \bG_{S^\intercal S}\big((z-\zeta_2^*)(z - \zeta_1^*)\big)\bS^\intercal
\end{split}
\end{equation}
Comparing the left-upper blocks in the first and second lines of \eqref{resolvent-relation} we find
\begin{equation}
\begin{split}
    \bS \bG_{S^\intercal S} &\big((z-\zeta_2^*)(z - \zeta_1^*)\big)\bS^\intercal = \Big\langle - \frac{\zeta_1^*}{z} \bI_N + \big( 1 - \frac{\zeta_1^*}{z} \big) \bY \bG_{Y^\intercal Y}(z^2)\bY^\intercal  \Big \rangle
\end{split}
\label{optimal-sv-eq1}
\end{equation}
The trace of the r.h.s of \eqref{optimal-sv-eq1} is (with multiplication by $1/N$)
\begin{equation*}
    \begin{split}
\frac{1}{N} \sum_{k=1}^N \Big[\frac{{\gamma_k}^2}{z^2 - {\gamma_k}^2 }\big( 1 - \frac{\zeta_1^*}{z} \big) - \frac{\zeta_1^*}{z} \Big] &=- \frac{\zeta_1^*}{z} \frac{1}{N} \sum_{k=1}^N \big[\frac{{\gamma_k}^2}{z^2 - {\gamma_k}^2 }+1 \big] +  \frac{1}{N} \sum_{k=1}^N \frac{{\gamma_k}^2}{z^2 - {\gamma_k}^2 } \\
        &\approx - \zeta_1^* z \mathcal{G}_{\rho_Y}(z^2) +  \mathcal{M}_{\mu_Y}  \big( \frac{1}{z^2} \big) \\
        &= - \zeta_1^* z \mathcal{G}_{\rho_Y}(z^2) +  z^2 \mathcal{G}_{\rho_Y}(z^2) - 1
    \end{split}
\end{equation*}
The last expression on the r.h.s can be expressed in terms of the symmetrized limiting spectral distribution of $\bY$. Indeed if we denote the Stieltjes of $\bar{\mu}_{Y}$ by $\mathcal{G}_{\bar{\mu}_{Y}}(z)$, using the relation $z \mathcal{G}_{\rho_Y}(z^2) = \mathcal{G}_{\bar{\mu}_Y}(z)$, the above trace implies with \eqref{optimal-sv-eq1}:
\begin{equation*}
    \frac{1}{N} \Tr \bS \bG_{S^\intercal S}\big((z-\zeta_2^*)(z - \zeta_1^*)\big)\bS \approx - \zeta_1^* \mathcal{G}_{\bar{\mu}_Y}(z) +  z \mathcal{G}_{\bar{\mu}_Y}(z) - 1
\end{equation*}
Moreover $\zeta_1^*$ in \eqref{zeta_sol} can be written as,
\begin{equation}
    \begin{split}
        \zeta_1^* = \frac{1}{\mathcal{G}_{\bar{\mu}_Y}(z)} \mathcal{C}^{(\alpha)}_{\mu_Z}\bigg(\frac{1}{z}  \mathcal{G}_{\bar{\mu}_Y}(z) \Big(1 - \alpha + \alpha z \mathcal{G}_{\bar{\mu}_Y}(z) \Big) \bigg)
    \end{split}
\end{equation}
Replacing these results in \eqref{optxi} we easily deduce \eqref{optimal-sv-final} for the optimal singular values of the RIE.
\begin{equation}\label{optimal-sv-final} 
    \begin{split}
        \widehat{\xi_j^*} &= \frac{1}{\pi \bar{\mu}_{Y}(\gamma_j)} {\rm Im} \,\Bigg[  \gamma_j \mathcal{G}_{\bar{\mu}_Y}(\gamma_j - \ci 0^+) - \mathcal{C}^{(\alpha)}_{\mu_Z}\bigg(\frac{1}{\gamma_j}  \mathcal{G}_{\bar{\mu}_Y}(\gamma_j - \ci 0^+) \Big(1 - \alpha + \alpha \gamma_j  \mathcal{G}_{\bar{\mu}_Y}(\gamma_j - \ci 0^+) \Big) \bigg) \Bigg] \\
        &=  \gamma_j - \frac{1}{\pi \bar{\mu}_{Y}(\gamma_j)} {\rm Im} \, \mathcal{C}^{(\alpha)}_{\mu_Z}\bigg( \frac{1- \alpha}{\gamma_j} \pi \sH [\bar{\mu}_{Y}](\gamma_j)+ \alpha \big( \pi \sH [\bar{\mu}_{Y}](\gamma_j)\big)^2  - \alpha \big( \pi \bar{\mu}_{Y}(\gamma_j)\big)^2 \\ 
       &\hspace{50pt}+\ci \pi \bar{\mu}_{Y}(\gamma_j) \big(\frac{1-\alpha}{\gamma_j} + 2 \alpha \pi \sH [\bar{\mu}_{Y}](\gamma_j) \big) \bigg)
    \end{split}
\end{equation}

% \begin{remark}
% To derive \eqref{optimal-sv-final}, we have absorbed $\sqrt{\lambda}$ into $\bS$. Therefore, to obtain the estimator for model \eqref{observation-matrix}, it should be divided by $\sqrt{\lambda}$ to get \eqref{rect-RIE}.
% \end{remark}

% \subsection{Gaussian Noise}

\subsection{Optimality of oracle estimator}\label{optimality-RIE}
In this section, we show that for rotational invariant priors, the posterior mean estimator belongs to the RIE class. We proceed by presenting an equivalent definition of the RIE and then show that posterior mean estimator satisfies this definition.

\begin{lemma}\label{RIE-lemma}
Given the observation matrix $\bY$, let $\Theta_S(\bY)$ be an estimator for $\bS$. Then $\Theta_S(\bY)$ is a RIE if and only if for any  orthogonal matrices $\bU \in \bR^{N \times N}, \bV \in \bR^{M \times M}$:
    \begin{equation}
        \Theta_S( \bU \bY \bV^\intercal) = \bU \Theta_S( \bY ) \bV^\intercal
        \label{RIE-property}
    \end{equation}
\end{lemma}
\begin{proof}
If $\Theta_S(\bY)$ is a RIE, then this property clearly follows from the definition \eqref{RIE-def}. Let us now show the converse.

Suppose that an estimator $\Theta_S(\bY)$ satisfies \eqref{RIE-property}. First, we show that if the observation matrix is diagonal, then the estimator is also diagonal. Consider the observation matrix to be $\bY^{\rm diag} = \left[
\begin{array}{c|c}
\rm{diag}(y_1, \dots, y_N) & \mathbf{0}_{N \times (M-N)} 
\end{array}
\right]$. Let $\bI_k^- \in \bR^{N \times N}, \bJ_k^- \in \bR^{M \times M}$ be diagonal matrices with diagonal entries all one except the $k$-th entry which is $-1$. Note that for $1 \leq k \leq N$, we have $\bY^{\rm diag} = \bI_k^- \bY^{\rm diag} \bJ_k^-$. Moreover matrices $\bI_k^-, \bJ_k^-$ are orthogonal thus or any $1 \leq k \leq N$, from \eqref{RIE-property} we have:
\begin{equation}
    \Theta_S(\bY^{\rm diag}) = \Theta_S( \bI_k^- \bY^{\rm diag} \bJ_k^- )= \bI_k^-  \Theta_S(\bY^{\rm diag}) \bJ_k^-
\end{equation}
This implies that all entries on the $k$-th row and $k$-th column of $\Theta_S(\bY^{\rm diag})$ is zero except the $k$-th entry on the diagonal. Since this holds for any $k$, we conclude that $\Theta_S(\bY^{\rm diag})$ is diagonal.

Now, for a given general observation matrix $\bY = \bU_Y \bGam \bV_Y^\intercal$, put $\bU = \bU_Y^\intercal, \bV = \bV_Y^\intercal$ in the property \eqref{RIE-property}. We have:
\begin{equation*}
    \Theta_S( \bGam ) = \bU_Y^\intercal \Theta_S( \bY ) \bV_Y
\end{equation*}
From the argument above, the matrix on the l.h.s is diagonal. Consequently, the matrix $\bU_Y^\intercal \Theta_S( \bY ) \bV_Y$ is diagonal which implies that the columns of $\bU_Y, \bV_Y$ are the left and right singular vectors of $\Theta_S( \bY )$. Therefore, $\Theta_S( \bY ) $ is a RIE.
\end{proof}

Now, we prove that the posterior mean estimator $\Theta_S^*(\bY) = \bE[ \bS | \bY ]$ has the property \eqref{RIE-property}, and thus belongs to the RIE class. For simplicity, we drop the SNR factor $\sqrt{\lambda}$. For any orthogonal matrices $\bU \in \bR^{N \times N}, \bV \in \bR^{M \times M}$, we have:
\begin{equation*}
\begin{split}
     \bE[ \bS |  \bU \bY \bV^\intercal] &= \frac{\int  d \tilde{\bS} \, \tilde{\bS} \, P_S(\tilde{\bS}) P_Z(  \bU \bY \bV^\intercal - \tilde{\bS} ) }  {\int  d \tilde{\bS} \, P_S(\tilde{\bS}) P_Z(  \bU \bY \bV^\intercal - \tilde{\bS} )} \\
    &\stackrel{\text{(a)}}{=} \frac{\int \, d \tilde{\bS} \, \bU \tilde{\bS} \bV^\intercal \,  P_S(\tilde{\bS}) P_Z( \bU \bY \bV^\intercal -  \bU \tilde{\bS} \bV^\intercal ) }  {\int \, d \tilde{\bS} \,  P_S(\tilde{\bS}) P_Z( \bU \bY \bV^\intercal -  \bU \tilde{\bS} \bV^\intercal )}  \\
    &\stackrel{\text{(b)}}{=} \bU  \Big\{ \frac{\int  d \tilde{\bS} \, \tilde{\bS}  \, P_S(\tilde{\bS}) P_Z(  \bY -  \tilde{\bS} ) }  {\int  d \tilde{\bS} \, P_S(\tilde{\bS}) P_Z(  \bY -  \tilde{\bS} )} \Big\} \bV^\intercal \\
    &=  \bU \bE[ \bS |\bY]  \bV^\intercal
\end{split}
\end{equation*}
where in (a), we changed variables $\tilde{\bS} \to \bU \tilde{\bS} \bV^\intercal$, used $|\det \bU | = |\det \bV | = 1$, and bi-rotational invariance of $P_S$, $P_S(\tilde{\bS}) = P_S(  \bU \tilde{\bS} \bV^\intercal )$. In (b), we used the bi-rotational invariance property of $P_Z$, namely $P_Z( \bU \bY \bV^\intercal - \bU \tilde{\bS} \bV^\intercal ) = P_Z( \bY - \tilde{\bS} ) $.

\subsection{Proof of Proposition \ref{optimal-singular-values&traces}}\label{Proof-Prop2}
Define the two measures:
\begin{equation*}
    \nu := \frac{1}{2 N} \sum_{j=1}^{N} \bu_j^\intercal \bS \bv_j \big( \delta_{\gamma_j} - \delta_{- \gamma_j} \big), \quad \quad 
    \tau := \frac{1}{2 N} \sum_{j=1}^{N}  \big( \delta_{\gamma_j} + \delta_{- \gamma_j} \big)
\end{equation*}
Using the Stieltjes inversion formula, for any $\epsilon > 0$ such that $[\gamma_j -\epsilon, \gamma_j +\epsilon] \cap \{ \gamma_1, \cdots, \gamma_N \} = \{ \gamma_j\} $, the optimal singular value $\xi^*_j$ can be expressed as:
 \begin{equation}\label{proof-prop-step2}
        \xi_j^* =  \lim_{\eta \to 0} \frac{ \int_{\gamma_j - \epsilon}^{\gamma_j + \epsilon} \im  \mathcal{G}_{\nu} (x + \ci \eta ) \, dx}{\int_{\gamma_j - \epsilon}^{\gamma_j + \epsilon} \im \big\{\mathcal{G}_{\tau} (x + \ci \eta ) \big\} \, dx}
    \end{equation}
with $\mathcal{G}_{\nu}, \mathcal{G}_{\tau}$ the Stieltjes transforms of $\nu, \tau$. The first Stiltjes transform can be written as:
\begin{equation}\label{proof-prop-step3}
    \begin{split}
        \mathcal{G}_{\nu} (z) &= \frac{1}{2 N} \sum_{j=1}^{N} \bu_j^\intercal \bS \bv_j \big( \frac{1}{z - \gamma_j} - \frac{1}{z + \gamma_j} \big) \\
        &= \frac{1}{N} \sum_{j=1}^{N}  \frac{\gamma_j}{z^2 - \gamma_j^2} \bu_j^\intercal \bS \bv_j \\
        &= \frac{1}{N} \sum_{j=1}^{N}  \frac{\gamma_j}{z^2 - \gamma_j^2} \Tr \bv_j \bu_j^\intercal \bS \\
        &= \frac{1}{N} \sum_{j=1}^{N}  \frac{\gamma_j}{z^2 - \gamma_j^2} \Tr  \bS^\intercal \bu_j \bv_j^\intercal \\
        &= \frac{1}{N} \Tr  \bS^\intercal \sum_{j=1}^{N}  \frac{\gamma_j}{z^2 - \gamma_j^2}  \bu_j \bv_j^\intercal \\
        &= \frac{1}{N} \Tr  \bS^\intercal \big( z^2 \bI - \bY \bY^\intercal \big)^{-1} \bY \\
        &= \frac{1}{N} \Tr  \bG_{YY^\intercal}(z^2) \bY \bS^\intercal\\
        &= L (z)
    \end{split}
\end{equation}
Similarly, we get
\begin{equation}\label{proof-prop-step4}
    \mathcal{G}_{\tau} (z) = z G(z)
\end{equation}
Finally \eqref{exact-RND} follows from \eqref{proof-prop-step2}, \eqref{proof-prop-step3}, \eqref{proof-prop-step4}. \hfill $\square$
\subsection{Proof of Theorem \ref{trace-relation}}\label{Proof-Thm2}
For simplicity of notation, we drop the $z$-dependence of the random functions $G(z), L(z)$. Let 
\begin{equation*}
    g := \bE \, G, \quad \quad l := \bE \, L
\end{equation*}
where the expectation is over the noise matrix $\bZ$ in \eqref{model}. We will need the following lemma whose proof is deferred to  subsection \ref{concentrationlemma}.

\begin{lemma}\label{concentration-lemma}
    There is a numerical constant $c>0$ (depending on $K$) such that for any $z \in \mathbb{C} \backslash \bR$ and for any $t>0$, we have:
    \begin{equation*}
        \bP \big( | G - g | \geq t \big) \leq 2 e^{- c \big(t N \big| \im z \big|^3 \big)^2 }
    \end{equation*}
    The same is also true for $L, l$.
\end{lemma}

Consider the decomposition
\begin{equation*}
    L = l + (L - l)
\end{equation*}
from the lemma above, $L-l$ is a sub-Gaussian random variable with sub-Gaussian norm $O\Big(\frac{1}{N \big| \im z \big|^3}\Big)$. Therefore, to prove the Theorem, it suffices to show that
\begin{equation}
    l = g \big( z^2 - \frac{1}{\alpha_0} \big) - g\big(z^2 g - 1) - 1 + O\Big( \frac{1}{N \big| \im z \big|^3} \Big)
\end{equation}
Let $\bG := \bG_{Y Y^\intercal}(z^2)$. We start by expanding the following matrix products:
\begin{equation*}
    \bG \bY \bY^\intercal = \lambda \bG \bS \bS^\intercal + \sqrt{\lambda} \bG \bS \bZ^\intercal + \sqrt{\lambda} \bG \bZ \bS^\intercal + \bG \bZ \bZ^\intercal
\end{equation*}
\begin{equation*}
    \bG \bY \bS^\intercal =  \sqrt{\lambda} \bG \bS \bS^\intercal +  \bG \bZ \bS^\intercal
\end{equation*}
Using the identity $ z^2 \bG - \bI = \bG \bY \bY^\intercal$, we have:
\begin{equation}
    \bG \bY \bS^\intercal = \frac{1}{\sqrt{\lambda}} \Big( z^2 \bG - \bI - \bG \bZ \bZ^\intercal \Big) - \bG \bS \bZ^\intercal 
    \label{expansion of GYS}
\end{equation}
Taking expectation and trace of the both sides:
\begin{equation}
    \bE \, \Tr \bG \bY \bS^\intercal = \frac{1}{\sqrt{\lambda}}  z^2 \bE \, \Tr \bG - \frac{1}{\sqrt{\lambda}}  N -  \bE \, \Tr   \bG \bS \bZ^\intercal - \frac{1}{\sqrt{\lambda}}  \bE  \Tr \,  \bG \bZ \bZ^\intercal 
    \label{expansion of E-Tr-GYS}
\end{equation}
The 
next step is to compute the last two terms in \eqref{expansion of E-Tr-GYS} through a use of gaussian integration by parts. \\

$\bullet$ Expansion of $\bE \, \Tr \bG \bS \bZ^\intercal$:
Using cyclicity of the trace and the fact that $\bG$ is symmetric, we have:
\begin{equation}
    \bE \, \Tr   \bG \bS \bZ^\intercal = \bE \, \Tr   \bG \bZ \bS^\intercal = \sum_{i = 1}^N \sum_{j = 1}^M \sum_{k = 1}^N \bE \, G_{ik} Z_{kj} \, S_{ij}
    \label{expansion of E-Tr-GSZ}
\end{equation}
Gaussian integration by parts yields:
\begin{equation}
    \bE \, G_{ik} Z_{kj} = \frac{1}{N} \bE \Big( \frac{\partial \bG}{\partial Z_{kj}} \Big)_{ik}
    \label{Tr-GSZ-Stein}
\end{equation}
We have:
\begin{equation*}
    \frac{\partial \bG}{\partial Z_{ab}} = - \bG \frac{\partial \big(z^2 \bI - \bY \bY^\intercal \big)}{\partial Z_{ab}} \bG = \bG \big( \bY {\bJ^{ab}}^\intercal + \bJ^{ab} \bY^\intercal \big) \bG 
\end{equation*}
with $\bJ^{ab} \in \bR^{N \times M}, J^{ab}_{ij} = \delta\{ i=a, j=b\}$. Thus, we find:
\begin{equation}
    \begin{split}
        \Big( \frac{\partial \bG}{\partial Z_{kj}} \Big)_{ik} &= \big[ \bG \bY {\bJ^{kj}}^\intercal \bG]_{ik} + \big[ \bG \bJ^{kj} \bY^\intercal \bG]_{ik} \\
        &= \sum_{a,b,c} G_{ia} \, Y_{ab} \, J^{kj}_{cb}\,  G_{ck} + \sum_{a,b,c} G_{ia}\, J^{kj}_{ab}\, Y_{cb} \,G_{ck} \\
        &= \sum_a G_{ia}\, Y_{aj} \,G_{kk} + \sum_c G_{ik}\, Y_{cj}\, G_{ck} \\
        &= G_{kk} \,\big( \bG \bY \big)_{ij} + G_{ik} \,(\bG \bY)_{kj}
    \end{split}
    \label{Tr-GSZ-Stein-second}
\end{equation}
Joining \eqref{Tr-GSZ-Stein-second} with \eqref{Tr-GSZ-Stein}, \eqref{expansion of E-Tr-GSZ} can be written further to be:
\begin{equation}
    \begin{split}
        \bE \, \Tr   \bG \bS \bZ^\intercal &= \frac{1}{N} \sum_{i = 1}^N \sum_{j = 1}^M \sum_{k = 1}^N \Big[ G_{kk} \,\big( \bG \bY \big)_{ij} + G_{ik} \,(\bG \bY)_{kj} \Big] S_{ij} \\
        &= \frac{1}{N} \bE \Big[ \big( \Tr \bG \big) \big( \Tr \bG \bY \bS^\intercal \big) \Big] + \frac{1}{N} \bE \, \Tr \bG \bG \bY \bS^{\intercal} 
    \end{split}
    \label{expansion of E-Tr-GSZ-final}
\end{equation}

$\bullet$ Expansion of $\bE \, \Tr \bG \bZ \bZ^\intercal$:
\begin{equation}
    \bE \, \Tr   \bG \bZ \bZ^\intercal  = \sum_{i = 1}^N \sum_{j = 1}^M \sum_{k = 1}^N \bE \, G_{ik} Z_{kj} \, Z_{ij}
    \label{expansion of E-Tr-GZZ}
\end{equation}
Using again gaussian integration by parts:
\begin{equation}
    \begin{split}
        \bE \, G_{ik} Z_{kj} Z_{ij} &= \frac{1}{N} \bE \, \frac{\partial G_{ik} Z_{ij}}{\partial Z_{kj}}\\
        &= \frac{1}{N} \bE \, Z_{ij} \Big( \frac{\partial \bG}{\partial Z_{kj}} \Big)_{ik} + \delta\{i=k\} \frac{1}{N} \bE \, G_{ik} \\
        &= \frac{1}{N} \bE \, Z_{ij} \Big( G_{kk} \,\big( \bG \bY \big)_{ij} + G_{ik} \,(\bG \bY)_{kj} \Big) + \delta\{i=k\} \frac{1}{N} \bE \, G_{ik}
    \end{split}
    \label{Tr-GZZ-Stein}
\end{equation}
Plugging in \eqref{expansion of E-Tr-GZZ}, we find
\begin{equation}
    \begin{split}
        \bE \, \Tr   \bG \bZ \bZ^\intercal &= \frac{1}{N} \sum_{i = 1}^N \sum_{j = 1}^M \sum_{k = 1}^N \bE \, Z_{ij} G_{kk} \,\big( \bG \bY \big)_{ij} +  \frac{1}{N} \sum_{i = 1}^N \sum_{j = 1}^M \sum_{k = 1}^N \bE \, Z_{ij} G_{ik} \,\big( \bG \bY \big)_{kj} +  \frac{1}{N} \sum_{i,j,k} \delta\{k=i\} \bE \, G_{ik}\\
        &= \frac{1}{N} \bE \Big[ \big( \Tr \bG \big) \big( \Tr \bG \bY \bZ^\intercal \big) \Big] + \frac{1}{N} \bE \, \Tr \bG \bG \bY \bZ^{\intercal} + \frac{M}{N} \bE \, \Tr \bG 
    \end{split}
    \label{expansion of E-Tr-GZZ-final}
\end{equation}

Replacing \eqref{expansion of E-Tr-GSZ-final} and \eqref{expansion of E-Tr-GZZ-final} in \eqref{expansion of E-Tr-GYS}, we find:
\begin{equation}
\begin{split}
    \bE \, \Tr \bG \bY \bS^\intercal &= \frac{1}{\sqrt{\lambda}} z^2 \bE \, \Tr \bG - \frac{1}{N} \bE \Big[ \big( \Tr \bG \big) \big( \Tr \bG \bY \bS^\intercal \big) \Big]  - \frac{1}{\sqrt{\lambda}} \frac{1}{N} \bE \Big[ \big( \Tr \bG \big) \big( \Tr \bG \bY \bZ^\intercal \big) \Big] \\
    &\hspace{1.5cm} - \frac{1}{N} \bE \, \Tr \bG \bG \bY \bS^{\intercal} - \frac{1}{\sqrt{\lambda}} \frac{1}{N} \bE \, \Tr \bG \bG \bY \bZ^{\intercal} - \frac{1}{\sqrt{\lambda}} \frac{M}{N} \bE \, \Tr \bG - \frac{1}{\sqrt{\lambda}} N \\
    &\hspace{-11.5pt}= \frac{1}{\sqrt{\lambda}} \big(z^2 - \frac{1}{\alpha_0} \big) \bE \, \Tr \bG - \frac{1}{\sqrt{\lambda}} \frac{1}{N} \bE \Big[ \big( \Tr \bG \big) \big( \Tr \bG \bY \bY^\intercal \big) \Big] - \frac{1}{\sqrt{\lambda}}\frac{1}{N} \bE \, \Tr \bG \bG \bY \bY^{\intercal} - \frac{1}{\sqrt{\lambda}} N \\
    &\hspace{-11.5pt}= \frac{1}{\sqrt{\lambda}} \big(z^2 - \frac{1}{\alpha_0} \big) \bE \, \Tr \bG - \frac{1}{\sqrt{\lambda}} \frac{1}{N} \bE \Big[ \big( \Tr \bG \big) \big( z^2 \Tr \bG -N \big) \Big] - \frac{1}{\sqrt{\lambda}} \frac{1}{N} \bE \, \Tr \bG \big( z^2  \bG -\bI \big) - \frac{1}{\sqrt{\lambda}} N
\end{split}
\end{equation}
Dividing by $N$ and  rearranging terms we find:
\begin{equation}
    l = \frac{1}{\sqrt{\lambda}} \big(z^2 +1- \frac{1}{\alpha_0} \big) g - \frac{1}{\sqrt{\lambda}} z^2 \bE \, G^2 - 1 - \frac{1}{N^2} z^2 \bE \, \Tr \bG^2 - \frac{1}{N} g
    \label{tr-expected-expa}
\end{equation}
Using lemma \ref{concentration-lemma},
\begin{equation}
\begin{split}
        \bE \, G^2 & = g^2 + \bE \, \big(g-G)^2 + 2 g \bE \, \big(g-G) \\
        &= g^2 + O \Big( \frac{1}{N \big|\im z \big|^3} \Big)
\end{split}
\label{tr-expan-term1}
\end{equation}
and,
\begin{equation}
    \begin{split}
        \big| \frac{1}{N^2} z^2 \bE \, \Tr \bG^2 \big| &\leq \bE \big| \frac{1}{N^2} z^2  \Tr \bG^2 \big|\\
        &\leq \frac{1}{N^2} \bE \big| \sum_{k=1}^N \frac{z^2}{\big( z^2 - \gamma_k^2 \big)^2} \big| \\
        &\leq \frac{1}{N^2} \bE \sum_{k=1}^N   \big| \frac{z^2}{\big( z^2 - \gamma_k^2 \big)^2} \big| \\ 
        &\leq \frac{1}{N^2} \bE \sum_{k=1}^N   \frac{1}{ \big( \im z \big)^2} = O \Big( \frac{1}{N  \big(\im z \big)^2} \Big)
    \end{split}
    \label{tr-expan-term2}
\end{equation}
Similarly, we have that 
\begin{equation}
    \frac{1}{N} g = O \Big( \frac{1}{N  \big(\im z \big)^2} \Big)
    \label{tr-expan-term3}
\end{equation}
Combining \eqref{tr-expected-expa}, \eqref{tr-expan-term1},\eqref{tr-expan-term2}, \eqref{tr-expan-term3}, we obtain the result:
\begin{equation*}
    l = \frac{1}{\sqrt{\lambda}} \Big[ \big(z^2 +1- \frac{1}{\alpha_0} \big) g - z^2 g^2 - 1 \Big] + O \Big( \frac{1}{N  \big|\im z \big|^3} \Big)
\end{equation*}
This completes the proof of Theorem \ref{trace-relation}.

\subsubsection{Proof of Lemma \ref{concentration-lemma}}\label{concentrationlemma}
To prove lemma \ref{concentration-lemma}, we use a Gaussian concentration inequality:
\begin{theorem}[Gaussian concentration inequality]
Let $X = (X_1 , . . . , X_n )$ be a vector of $n$ independent  Gaussian random variables of variance $\sigma^2$. Let $f : \bR^n \to \bR$ denote an $L$-Lipschitz function (w.r.t Euclidean norm in $\mathbb{R}^n$). Then, for any $t>0$,
\begin{equation*}
    \mathbb{P} \big( | f(X) - \bE f (X) | \geq t \big) \leq 2 e^{- \frac{1}{\sigma^2} \frac{t^2}{2 L^2}}
\end{equation*}
that is, $f(X) - \bE f (X)$ is sub-Gaussian with sub-Gaussian norm $O \big( \frac{1}{L \sigma} \big)$.
\end{theorem}

Using the above result, it suffices to show that $G(z), L(z)$ as functions of the noise matrix $\bZ \in \bR^{N \times M}$ are Lipschitz with constant $O \big( \frac{1}{\sqrt{N} \big| \im z \big|^3} \big)$. 

Consider the Hermitization $\cY$ of $\bY$ in \eqref{hermitization}. Given the decomposition \eqref{eigen-cY}, we have:
\begin{equation}\label{alternativeforG}
    \begin{split}
        G(z) &= \frac{1}{N} \Tr \big(z^2 \bI - \bY \bY^\intercal \big)^{-1} \\
        &= \frac{1}{N} \sum_{k =1}^N \frac{1}{z^2 - \gamma_k^2} \\
        &= \frac{1}{N} \sum_{k =1}^N \frac{1}{2 z} \big( \frac{1}{z - \gamma_k} + \frac{1}{z+\gamma_k} \big) \\
        &= \frac{1}{N} \frac{1}{2 z} \Big[ \sum_{k =1}^N  \big( \frac{1}{z - \gamma_k} + \frac{1}{z+\gamma_k} \big) + (M-N) \frac{1}{z} \Big] - \frac{M -N}{2 N} \frac{1}{z^2} \\
        &= \frac{1}{N} \frac{1}{2 z} \Tr \big(z \bI - \cY \big)^{-1} - \frac{M -N}{2 N} \frac{1}{z^2}
    \end{split}
\end{equation}
and
\begin{equation} \label{alternativeforL}
    \begin{split}
        L(z) &= \frac{1}{N} \Tr \big(z^2 \bI - \bY \bY^\intercal \big)^{-1} \bY \bS^\intercal \\
        &= \frac{1}{N} \Tr  \Big( \sum_{k=1}^{N} \frac{1}{z^2 - \gamma_k^2} \bu_k \bu_k^\intercal \Big) \bY \bS^\intercal \\
        &= \frac{1}{N} \sum_{k=1}^{N} \frac{1}{z^2 - \gamma_k^2} \Tr \bu_k \bu_k^\intercal  \bY \bS^\intercal \\
        &= \frac{1}{N} \sum_{k=1}^{N} \frac{\gamma_k}{z^2 - \gamma_k^2} \Tr \bu_k \bv_k^\intercal \bS^\intercal
    \end{split}
\end{equation}
On the other hand, letting $\cS$ be the Hertmitization of $\bS$, and $\bw_k$ the $k$-th column of $\bW$ in \eqref{eigen-cY}, we have:
\begin{equation*}
    \bw_k^\intercal \cS \bw_k = \begin{cases}
        \bu_k^\intercal \bS \bv_k & \rm{for} \hspace{5pt} 1 \leq k \leq N\\
        -\bu_{k-N}^\intercal \bS \bv_{k-N} & \hspace{5pt} N+1 \leq k \leq 2 N \\
        0 & \hspace{5pt} 2 N+1 \leq k \leq M + N
    \end{cases}
\end{equation*}
Therefore, denoting eigenvalues of $\cY$ by $\tilde{\gamma_k}$, we find
\begin{equation}\label{alternativeforL-2}
    \begin{split}
        \frac{1}{N} \Tr (z \bI - \cY)^{-1} \cS &= \frac{1}{N} \sum_{k = 1}^{M + N} \frac{1}{z - \tilde{\gamma_k}} \Tr \bw_k \bw_k^\intercal \cS \\
        &= \frac{1}{N} \Big[ \sum_{k = 1}^{N} \frac{1}{z - \gamma_k} \Tr \bu_k \bv_k^\intercal \bS - \sum_{k = 1}^{N} \frac{1}{z + \gamma_k} \Tr \bu_k \bv_k^\intercal \bS \Big] \\
        &= \frac{1}{N} \sum_{k = 1}^{N} \big( \frac{1}{z - \gamma_k} - \frac{1}{z+\gamma_k} \big)  \Tr \bu_k \bv_k^\intercal \bS^\intercal \\
        &= \frac{2}{N} \sum_{k=1}^{N} \frac{\gamma_k}{z^2 - \gamma_k^2} \Tr \bu_k \bv_k^\intercal \bS^\intercal
    \end{split}
\end{equation}
From \eqref{alternativeforL}, \eqref{alternativeforL-2}, we obtain:
\begin{equation}
    L(z) = \frac{1}{2 N} \Tr (z \bI - \cY)^{-1} \cS 
\end{equation}

Now, to show Lipschitz continuity of the functions we consider a variation of the noise matrix $\bZ\to \bZ + \mathbf{\delta}_Z$. From now on, variables evaluated at $\bZ + \mathbf{\delta}_Z$ are denoted with a "tilde" symbol, for example:
\begin{equation*}
    \tilde{G}(z) = \frac{1}{N} \Tr \big( z^2 \bI -\tilde{\bY} {\tilde{\bY}}^\intercal \big)^{-1} = \frac{1}{N} \Tr \big( z^2 \bI - (\bS + \bZ + \mathbf{\delta}_Z ) (\bS + \bZ + \mathbf{\delta}_Z)^\intercal \big)^{-1}
\end{equation*}

We have:
\begin{equation}
    \begin{split}
       |G(z) - \tilde{G}(z)| &= \frac{1}{N} \bigg| \Tr \Big[  \big(z^2 \bI - \bY \bY^\intercal \big)^{-1} - \big(z^2 \bI - \tilde{\bY} \tilde{\bY}^\intercal \big)^{-1} \Big] \bigg| \\
       &\overset{\rm (a)}{=} \frac{1}{N} \frac{1}{2 |z| } \bigg| \Tr \Big[  \big(z \bI - \cY \big)^{-1} - \big(z \bI - \tilde{\cY} \big)^{-1} \Big] \bigg| \\
       &\overset{\rm (b)}{=}  \frac{1}{N} \frac{1}{2 |z| } \bigg| \Tr \big(z \bI - \cY \big)^{-1} \big(\tilde{\cY} - \cY \big) \big(z \bI - \tilde{\cY} \big)^{-1} \bigg| \\
       &\overset{\rm (c)}{\leq} \frac{\sqrt{M+N}}{N} \frac{1}{2 |z| } \Big\| \big(z \bI - \cY \big)^{-1} \big(\tilde{\cY} - \cY \big) \big(z \bI - \tilde{\cY} \big)^{-1} \Big\|_{\rm F} \\
       &\overset{\rm (d)}{\leq} \frac{\sqrt{M+N}}{N} \frac{1}{2 |z| } \Big\|\big(z \bI - \cY \big)^{-1} \Big\|_{\rm op} \Big\| \big(z \bI - \tilde{\cY} \big)^{-1} \Big\|_{\rm op} \Big\| \big(\tilde{\cY} - \cY \big) \Big\|_{\rm F} \\
       &\overset{\rm (e)}{\leq} C \frac{1}{\sqrt{N}} \frac{1}{ \big| \im z \big|^3} \big\|  \mathbf{\delta}_Z \big\|_{\rm F}
    \end{split}
\end{equation}
where in $(a)$ we use the identity in \eqref{alternativeforG}, in $(b)$ we use the following resolvent formula, namely that for any square matrices $\bA, \bB$:
\begin{equation*}
    \big(z \bI - \bA \big)^{-1} - \big(z \bI - \bB \big)^{-1} = \big(z \bI - \bA \big)^{-1} \big( \bA - \bB \big) \big(z \bI - \bB \big)^{-1} \quad ,
\end{equation*}
in $(c)$ we use the inequality for any matrix $\bA \in \bR^{N \times N}$:
\begin{equation*}
    \big| \Tr \bA \big| \leq \sqrt{N} \| \bA \|_{\rm F},
\end{equation*}
in $(d)$ we  use a non-commutative Hölder-type inequality (see e.g. \cite{simon2005trace}, Thm 2.8), namely for any product $\bA_1  \cdots \bA_k$ of matrices with any size and any $i = 1,...,k$,
\begin{equation*}
    \big\| \bA_1 \cdots \bA_k \big\|_{\rm F} \leq \prod_{j \neq i} \| \bA_j \|_{\rm op} \, \, \| \bA_i \|_{\rm F},
\end{equation*}
and finally in $(e)$ the constant $C$ depends only on $K$ and we use that the operator norm of  $\big(z \bI - \tilde{\cY} \big)^{-1}$ is bounded by $\im z$:
\begin{equation*}
    \Big\| \big(z \bI - \tilde{\cY} \big)^{-1} \Big\|_{\rm op} = \max_{ \{ 0, \pm \gamma_1, \cdots, \pm \gamma_k \}} \frac{1}{|z - x|} \leq \frac{1}{|z - {\rm Re} \, z|} = \frac{1}{\big| \im z \big| }
\end{equation*}

Similarly for $L(z)$, we have:
\begin{equation}
    \begin{split}
        | L(z) - \tilde{L}(z) | & = \frac{1}{2 N} \bigg| \Tr \Big[ \big(z \bI - \cY \big)^{-1} \cS  - \big(z \bI - \tilde{\cY} \big)^{-1} \cS \Big] \bigg| \\
        &= \frac{1}{2 N} \bigg| \Tr \Big[ \big(z \bI - \cY \big)^{-1} - \big(z \bI - \tilde{\cY} \big)^{-1} \Big] \cS \ \bigg| \\
        &\leq \frac{1}{2 N} \Big\| \big(z \bI - \cY \big)^{-1} - \big(z \bI - \tilde{\cY} \big)^{-1} \Big\|_{\rm F} \big\| \cS \big\|_{\rm F} \\
        &\leq \frac{\sqrt{2 N}}{2 N} \Big\| \big(z \bI - \cY \big)^{-1} - \big(z \bI - \tilde{\cY} \big)^{-1} \Big\|_{\rm F} \big\| \bS \big\|_{\rm op} \\
        &\leq \frac{C'}{\sqrt{N}} \frac{1}{ \big( \im z \big)^2}  \big\| \mathbf{\delta}_Z \big\|_{\rm F}
    \end{split}
\end{equation}
with $C'$ a positive constant depending only on $K$.

\subsection{Computation of MMSE for the Gaussian Noise - Statement \ref{Gaus-MMSE-statement}}\label{com-G-MMSE}
From \eqref{MMSE-eq} and \eqref{rect-RIE-Gauss} we see that to compute the MMSE we must compute the following expectation:
\begin{equation*}
    \int \bigg( x - \frac{1-\alpha}{\alpha} \frac{1}{x} - 2 \pi \sH [\bar{\mu}_{Y}](x) \bigg)^2 \, \mu_Y(x) \, dx
\end{equation*}
In the following, using properties of the Hilbert transform, we show this integral equals:
\begin{equation}
    \int x^2 \mu_Y(x) \, dx + \big(\frac{1}{\alpha} -1 \big)^2 \int \frac{\mu_Y(x)}{x^2}  dx + \frac{\pi^2}{3} \int {\mu_Y(x)}^3 \, dx - \frac{2}{\alpha}
    \label{expect-xi*}
\end{equation}
Putting these relations together, we deduce (for Gaussian noise):
\begin{equation*}
\begin{split}
    \int {\xi^*(x)}^2 \mu_Y(x) \, dx &= \int x^2 \mu_S(x) \, dx - \frac{1}{\lambda} \Big[ \frac{1}{\alpha} \\
    &\hspace{-30pt}- \big(\frac{1}{\alpha} -1 \big)^2 \int \frac{\mu_Y(x)}{x^2}\, dx 
        - \frac{\pi^2}{3} \int {\mu_Y(x)}^3 \, dx \Big]
\end{split}
\end{equation*}
Replacing this identity in \eqref{MMSE-eq}, we get \eqref{G-MMSE-eq}.

\subsubsection{Derivation of \texorpdfstring{\eqref{expect-xi*}}{}}
For simplicity we denote $\sH [\bar{\mu}_{Y}](x)$ by $\bar{\sH}(x)$. Expanding the square in the integrand, we find
\begin{equation}
\begin{split}
    x^2 + &\big( \frac{1-\alpha}{\alpha} \big)^2 \frac{1}{x^2} -  2 \frac{1-\alpha}{\alpha}+ 4 \pi^2 \big(\bar{\sH}(x)\big)^2 - 4 \pi x \bar{\sH}(x) + 4 \pi \frac{1-\alpha}{\alpha} \frac{\bar{\sH}(x)}{x}
\end{split}
\label{expansion}
\end{equation}
To compute the expectation of the last three terms, we need the following properties of the Hilbert transform. 
\begin{lemma}\label{properties of Hilbert}
If $f : \bR \to \bR$ is compactly supported and
sufficiently regular, then one has the identities
\begin{equation}
   \int_\bR f(x)\big( \sH [f] (x) \big)^2 \, dx = \frac{1}{3} \int_\bR f^3(x) \, dx
   \label{Hilbert-1}
\end{equation}
\begin{equation}
   \int_\bR \sH [f] (x) x f(x) \, dx = \frac{1}{2 \pi} \Big( \int_\bR f(x) \, dx \Big)^2
   \label{Hilbert-2}
\end{equation}
\begin{equation}
   \int_\bR \frac{\sH [f] (x)}{x}  f(x) \, dx = -\frac{1}{2 \pi} \Big( \int_\bR \frac{f(x)}{x} \, dx \Big)^2
   \label{Hilbert-3}
\end{equation}
\label{Hilbert-iden}
\end{lemma}
\begin{proof}
The proof of the first two properties can be found in Lemma 3.1 of \cite{shlyakhtenko2020fractional}. To prove the last equality, we apply the same idea as in remark 3.2 of this paper to write:
\begin{equation*}
    \begin{split}
        \int_\bR \frac{\sH [f] (x)}{x}  f(x) \, dx&= \frac{1}{2 \pi} \iint \big( \frac{1}{x} - \frac{1}{y} \big) \frac{1}{x-y} f(x) f(y) \, dx \, dy \\
        &= - \frac{1}{2 \pi} \iint  \frac{1}{xy} f(x) f(y) \, dx \, dy \\
        &= - \frac{1}{2 \pi} \Big( \int \frac{f(x)}{x} \, dx \Big)^2
    \end{split}
\end{equation*}
\end{proof}

We remark that the Hilbert transform  of an even function is an odd function (see e.g. \cite{kschischang2006hilbert}), in other words for the symmetrized measure $\bar{\mu}_Y$ we have $\bar{\sH}(x) = - \bar{\sH}(x)$.
From \eqref{Hilbert-1} we have:
\begin{equation}
     \int \big(\bar{\sH}(x)\big)^2 \bar{\mu}_Y(x)\, dx = \frac{1}{3} \int {\bar{\mu}_Y(x)}^3 \, dx
     \label{first-term}
\end{equation}
The l.h.s can be written as:
\begin{equation*}
\begin{split}
        \frac{1}{2} \int_{\bR_+} \big(\bar{\sH}(x)\big)^2 \mu_Y(x)\, dx &+ \frac{1}{2} \int_{\bR_-} \big(\bar{\sH}(x)\big)^2 \mu_Y(-x)\, dx \\
        &= \frac{1}{2} \int_{\bR_+} \big(\bar{\sH}(x)\big)^2 \mu_Y(x)\, dx +  \frac{1}{2} \int_{\bR_+} \big(\bar{\sH}(-x)\big)^2 \mu_Y(x)\, dx \\
        &= \frac{1}{2} \int_{\bR_+} \big(\bar{\sH}(x)\big)^2 \mu_Y(x)\, dx +  \frac{1}{2} \int_{\bR_+} \big(\bar{\sH}(x)\big)^2 \mu_Y(x)\, dx \\
        &= \int_{\bR_+} \big(\bar{\sH}(x)\big)^2 \mu_Y(x)\, dx
\end{split}
\end{equation*}
The rhs in \eqref{first-term} equals $ \frac{1}{12} \int {\mu_Y(x)}^3 \, dx$. Therefore, the expectation of the fourth term in \eqref{expansion} is:
\begin{equation}
   4 \pi^2 \int  \big(\bar{\sH}(x)\big)^2 \mu_Y(x) \, dx = \frac{\pi^2}{3}  \int {\mu_Y(x)}^3 \, dx
   \label{expec1}
\end{equation}

Similarly, using symmetry properties of $\bar{H}$ and  $\bar{\mu}_Y$ we have that:
\begin{equation*}
    \int x \bar{H}(x) \bar{\mu}_Y(x) \, dx = \int x \bar{H}(x) \mu_Y(x) \, dx
\end{equation*}
Thus, by \eqref{Hilbert-2}, the expectation of the fifth term in \eqref{expansion} is:
\begin{equation}
    -4 \pi \int x \bar{H}(x) \mu_Y(x) \, dx = -2  \Big( \int_\bR \bar{\mu}_Y(x) \, dx \Big)^2 = -2.
    \label{expec2}
\end{equation}
Again, by symmetry we have:
\begin{equation*}
    \int  \frac{\bar{H}(x)}{x} \bar{\mu}_Y(x) \, dx = \int \frac{\bar{H}(x)}{x}  \mu_Y(x) \, dx
\end{equation*}
Thus, by \eqref{Hilbert-3}, the expectation of the last term in \eqref{expansion} is:
\begin{equation}
    \int \frac{\bar{H}(x)}{x}  \mu_Y(x) \, dx =  \Big( \int_\bR \frac{\bar{\mu}_Y(x)}{x} \, dx \Big)^2 = 0
    \label{expec3}
\end{equation}
where we used that $\frac{\bar{\mu}_Y(x)}{x}$ is an odd function.

Finally putting together \eqref{expansion}, \eqref{expec1}, \eqref{expec2}, \eqref{expec3} we get \eqref{expect-xi*}.
% \begin{equation*}
% \begin{split}
%     \int x^2 &\mu_Y(x) \, dx = \lim_{N \to \infty} \frac{1}{N} \bE [\Tr \bY \bY^T] \\
%     &= \lim_{N \to \infty} \frac{\lambda}{N} \bE[\Tr \bS \bS^T] + \lim_{N \to \infty} \frac{1}{N} \bE[\Tr \bZ \bZ^T] \\
%     &\hspace{20pt} + \lim_{N \to \infty} \frac{2 \sqrt{\lambda}}{N} \bE[\Tr \bS \bZ^T] \\
%     &= \lambda \int x^2 \mu_S(x) \, dx + \frac{1}{\alpha}
% \end{split}
% \end{equation*}
% \begin{equation*}
% \begin{split}
%     \int x^2 &\mu_Y(x) \, dx =  \lambda \int x^2 \mu_S(x) \, dx + \frac{1}{\alpha}
% \end{split}
% \end{equation*}

\subsection{Proof of Theorem \ref{MI-th}}
We start from the posterior distribution of the model \eqref{model} which reads (up to some constants):
\begin{equation}
    \begin{split}
        P (\bX | \bY) &\propto e^{-\frac{N}{2} \| \bY - \sqrt{\lambda}\bX \|_{\rm F}^2 } P_{S}(\bX) \\
        & \propto e^{N \Tr \big[ \sqrt{\lambda}\bX \bY^\intercal - \frac{\lambda}{2} \bX \bX^\intercal  \big] } P_{S}(\bX)
    \end{split}
    \label{post-model1}
\end{equation}
The partition function is defined as the normalizing factor of the posterior distribution \eqref{post-model1}:
\begin{equation}
    Z(\bY) = \int \, d \bX \, e^{N \Tr \big[ \sqrt{\lambda}\bX \bY^\intercal  - \frac{\lambda}{2} \bX \bX^\intercal  \big] } P_{S}(\bX)
    \label{partition-function-def}
\end{equation}
and the free energy is defined as:
\begin{equation}
    F_N (\lambda) = -\frac{1}{M N} \bE_{Y} \big[ \ln Z(\bY) \big]
    \label{free-energy-def}
\end{equation}
One can easily see that the free energy is linked to the (average) mutual information via the relation:
\begin{equation*}
    \frac{1}{M N} I_N(\bS;\bY) = F_N (\lambda) + \frac{\lambda}{2 M} \bE \big[ \Tr \bS \bS^\intercal  \big]
\end{equation*}
in which $\frac{1}{M} \bE \big[ \Tr \bS^\intercal  \bS \big]$ converges to the second moment of $\mu_S$ rescaled by the factor $\alpha$. Therefore, to prove theorem \ref{MI-th}, it is enough to show that
\begin{equation*}
     \lim_{N \to \infty} F_N (\lambda)  = \frac{\lambda}{2} \alpha \int x^2 \mu_{S}(x) \, d x - \mathcal{J}[\mu_{\sqrt{\lambda} S}, \mu_{\sqrt{\lambda} S}\boxplus_{\alpha} \mu_{\rm MP}]
\end{equation*}
To prove this limit, first, we show that this limit also holds for the free energy of a simpler model. Then, using the \textit{pseudo-Lipschitz} continuity of the free energy w.r.t. to a distance between two models which converges to $0$ as $N \to \infty$, we deduce that the same limit holds for the free energy of the original model.

\subsubsection{An independent singular value model}
Suppose $\bsig^0 \in \bR^N$ is generated with i.i.d. elements from $\mu_S$, and is ordered in non-decreasing way. Fix $\bsig^0$ once for all. Let $\tilde{\bS} \in \bR^{N \times M}$  the matrix  contructed as $\bU \tilde{\bSig} \bV^\intercal $ where $\bU \in \bR^{N \times N}, \bV \in \bR^{M \times M}$ are independent and distributed according to the Haar measure, and $\tilde{\bSig} \in \bR^{N \times M}$ with $\tilde{\bsig}$ on its main diagonal for $\tilde{\bsig} \in \bR^N$. The distribution of the matrix $\tilde{\bS}$ is :
\begin{equation}
\begin{split}
        d P_{\tilde{S}}(\tilde{\bS}) &= d \mu_N(\bU) \,  d \mu_M(\bV) d p_{\tilde{S}}(\tilde{\bsig}) = d \mu_N(\bU) \, d \mu_M(\bV) \, \prod_{i=1}^N \delta(\tilde{\sigma}_i - \sigma^0_i) \, d \tilde{\bsig}
\end{split}
\label{model 2}
\end{equation}
 
Matrix $\tilde{\bS}$ is observed through an AWGN channel as in \eqref{model}, $\tilde{\bY} = \sqrt{\lambda} \tilde{\bS} + \tilde{\bZ}$. The partition function and the free energy can be defined in the same way as in \eqref{partition-function-def},\eqref{free-energy-def} denoted by $\tilde{Z}(\tilde{\bY})$, $\tilde{F}_N(\lambda)$ respectively.
% \begin{equation}
%     \tilde{Z}(\tilde{\bY}) = \int d \bX e^{N \Tr \big[ \sqrt{\lambda}\bX \tilde{\bY}^\intercal  - \frac{\lambda}{2} \bX \bX^\intercal  \big] } P_{\tilde{S}}(\bX)
%     \label{partition-function-2}
% \end{equation}
% \begin{equation}
%     \tilde{F}_N(\lambda) = -\frac{1}{M N} \bE_{\tilde{Y}} \big[ \ln \tilde{Z} (\tilde{\bY}) \big]
%     \label{free-energy-2}
% \end{equation}
\begin{proposition}
For $\mu_S$ with compact support, and any $\lambda > 0$, we have $\mu_S$-almost surely
\begin{equation*}
    \lim_{N \to \infty} \tilde{F}_N(\lambda) = \frac{\lambda}{2} \alpha \int x^2 \mu_{S}(x) \, dx - \mathcal{J}[\mu_{\sqrt{\lambda} S}, \mu_{\sqrt{\lambda} S}\boxplus_{\alpha} \mu_{\rm MP}] 
% \label{asymp-free-energy-2}
\end{equation*}
\label{asymp-free-energy-2-proposition}
\end{proposition}
\textit{Proof } Appendix \ref{proof of prop1}.

\subsubsection{Pseudo-Lipschitz continuity of the free energy}
Consider two bi-rotationally invariant matrix ensembles $P^{(1)}$, $P^{(2)}$, i.e. for $\bS \sim P^{(1)}(\bS)$, $\tilde{\bS} \sim P^{(2)}(\tilde{\bS})$ with SVDs $\bS = \bU \bSig \bV^\intercal $, $\tilde{\bS} = \tilde{\bU} \tilde{\bSig} \tilde{\bV}^\intercal $
\begin{equation*}
\begin{split}
        d P_N^{(1)} (\bS )  &\propto d \mu_N(\bU) \, d \mu_M(\bV) \, p^{(1)} (\bsig ) \, d \bsig \\
        d P_N^{(2)} ( \tilde{\bS} )  &\propto d \mu_N(\tilde{\bU}) \, d \mu_M(\tilde{\bV})\, p^{(2)} (\tilde{\bsig} ) \, d \tilde{\bsig}
\end{split}
\end{equation*}
where  $p^{(1)}(\bsig)$, $p^{(2)}(\tilde{\bsig})$ are the joint probability density functions for the singular values, induced by the priors. Suppose each of these distributions to be the prior of an inference problem in model \eqref{model}. The free energy can be defined similarly for each of the priors, which are denoted by $F_N^{(1)}(\lambda), F_N^{(2)}(\lambda)$ respectively. Then, we have

\begin{proposition}
For all $\lambda > 0$ and $N$ :
\begin{equation}\label{pseudo-lip-ub}
        \big| F_N^{(1)}(\lambda) - F^{(2)}_N(\lambda) \big| \leq \frac{\lambda}{2 N} \Big( \sqrt{ \bE_{\bsig} \big[ \| \bsig \|^2 \big]} + \sqrt{ \bE_{\tilde{\bsig}} \big[ \| \tilde{\bsig} \|^2 \big]} \Big) \sqrt{\bE_{\bsig, \tilde{\bsig}} \big[ \| \bsig - \tilde{\bsig} \|^2 \big] }  
\end{equation}
\label{pseudo-lip}
\end{proposition}
\textit{Proof } Appendix \ref{proof of prop2}.

\subsubsection{The distance between two models}
Recall that
\begin{equation*}
\begin{split}
        d P_{S}(\bS)  \propto d \mu_N(\bU) \, d \mu_M(\bV) p_{S}(\bsig) \, d \bsig
\end{split}
\end{equation*}
where $p_{S}(\bsig)$ is the joint p.d.f. of singular values of $\bS$. Moreover, $d P_{\tilde{S}}(\tilde{\bS})$ is defined in \eqref{model 2} with  $p_{\tilde{S}}(\tilde{\bsig}) \equiv \prod_{i=1}^N \delta(\tilde{\sigma}_i - \sigma^0_i)$, where $\bsig^0$ is generated with i.i.d. elements from $\mu_S$
\begin{lemma}
Under assumptions \ref{assumptions on law}, \ref{bounded-mom}, for $\bsig \sim p_{S}(\bsig)$, $\tilde{\bsig} \sim p_{\tilde{S}}(\tilde{\bsig})$ , we have:
\begin{equation}
    \lim_{N \to \infty} \frac{1}{N} \bE_{\bsig, \tilde{\bsig}} \big[ \| \bsig - \tilde{\bsig} \|^2 \big] =0
    \label{W2-dis-eq}
\end{equation}
\label{W2-dis}
\end{lemma}

\textit{Proof } Appendix \ref{proof of lemma 2}.

\subsubsection{Concluding the Proof}
By proposition \ref{pseudo-lip},  the distance between the free energies $F_N(\lambda)$ (defined in \eqref{free-energy-def}) and $\tilde{F}_N(\lambda)$ is upper bounded by rhs in \eqref{pseudo-lip-ub}.
The term $\frac{1}{N} \| \bsig \|^2 = \frac{1}{N} \sum \sigma_i^2 $ is the second moment of the empirical spectral distribution of $\bS$, which is almost surely bounded by assumption \ref{bounded-mom}. So,  $\frac{1}{N} \bE_{\bsig} \big[ \| \bsig \|^2 \big]$ is bounded.
 Moreover, $\frac{1}{N} \bE \big[ \| \tilde{\bsig} \|^2 \big] = \frac{1}{N}  \sum {\sigma^0}_i^2$ which is bounded by $C_2^2$. By proposition \ref{W2-dis}, $\lim_{N \to \infty} \frac{1}{N} \bE_{\bsig, \tilde{\bsig}} \big[ \| \bsig - \tilde{\bsig} \|^2 \big] = 0$. Therefore $    \lim_{N \to \infty} |F_N(\lambda) - \tilde{F}_N(\lambda) | = 0$ and Proposition \ref{asymp-free-energy-2-proposition} gives the result. $\hfill \square$

\section*{Acknowledgments} 
We are thankful to Pierre Mergny for interesting discussions. The work of F. P has been supported by Swiss National Science Foundation grant no 200021-204119.

\bibliographystyle{unsrt}
\bibliography{References}

\newpage

\begin{appendices}
    \section{Reminder on Transforms in Random Matrix Theory}\label{RMT-transforms}
The ESD of $\bS$ is defined as:
\begin{equation*}
    \mu_{S}^{(N)}(x) = \frac{1}{N} \sum_{i=1}^N \delta(x - \sigma_i )
\end{equation*}
where $(\sigma_i)_{1 \leq i \leq N}$ are the singular values of $\bS$.

For a probability measure $\mu$ with support contained in $[-K, K ]$ with $K>0$, we define a generating function of (even) moments $\mathcal{M}_{\mu} : [0,K^{-2}] \to \bR_+$
as
\begin{equation*}
    \mathcal{M}_{\mu} (z) = \int \frac{1}{1-t^2 z} \mu(t) \, d t - 1
\end{equation*}
For $\alpha \in [0,1]$, set $T^{(\alpha)}(z) = (\alpha z +1)(z+1)$ and $\mathcal{H}_{\mu}^{(\alpha)}(z) = z T^{(\alpha)}\big(\mathcal{M}_{\mu} (z)\big)$. The \textit{rectangular R-transform} is then defined as:
\begin{equation*}
    \mathcal{C}_{\mu}^{(\alpha)}(z) = {T^{(\alpha)}}^{-1}\Big(\frac{z}{{\mathcal{H}_{\mu}^{(\alpha)}}^{-1}(z)} \Big)
\end{equation*}

For a probability density $\mu(x)$ on $\bR$, the \textit{Stieltjes}  (or \textit{Cauchy}) transform  is defined as 
\begin{equation*}
   \mathcal{G}_\mu (z) = \int_{\bR} \frac{1}{z - x} \mu(x) \, dx \hspace{10 pt} \text{for } z \in \mathbb{C} \backslash {\rm supp}(\mu)
\end{equation*}
By  Plemelj formulae we have for $x\in \mathbb{R}$,
\begin{equation}
    \lim_{y \to 0^+} \mathcal{G}_\mu(x - \ci y) = \pi \sH [\mu](x) + \ci\pi \mu(x) 
    \label{Plemelj formulae}
\end{equation}
with $\sH [\mu](x) = {\rm p.v.} \frac{1}{\pi} \int_{\bR} \frac{\mu(t)}{x - t}  d t$ the \textit{Hilbert} transform of $\mu$.
\section{Derivation of the resolvent relation }\label{RIE-der}
From \eqref{RIE-observ-model}, we have
\begin{equation}
\begin{split}
    \cY &=  \left[
\begin{array}{cc}
\mathbf{0} & \bS \\
\bS^\intercal  & \mathbf{0}
\end{array}
\right] + \left[
\begin{array}{cc}
\bU & \mathbf{0}  \\
\mathbf{0} & \bV
\end{array}
\right] \left[
\begin{array}{cc}
\mathbf{0} & \bZ \\
\bZ^\intercal  & \mathbf{0}
\end{array}
\right] \left[
\begin{array}{cc}
\bU^\intercal  & \mathbf{0}  \\
\mathbf{0} & \bV^\intercal 
\end{array}
\right] \\
&= \cS + \bO \cZ \bO^\intercal 
\end{split}
\label{RIE-1}
\end{equation}
Let $\bG(z) \equiv \bG_{\mathcal{Y}}(z)$. First, we express the entries of $\bG(z)$ using the Gaussian integral representation of an inverse matrix \cite{zinn2021quantum}:
\begin{equation}
\begin{split}
    G_{ij}(z) &= \sqrt{\frac{1}{( 2 \pi )^{N+M} \det \,  (z \bI - \cY) } }  \int \Big( \prod_{k=1}^{M+N}  d \eta_k  \Big) \, \eta_i \eta_j \,  \exp \Big\{ -\frac{1}{2}  \bbeta^\intercal \big( z \bI - \cY \big) \bbeta \Big\} \\
    &= \ddfrac{\int \Big( \prod_{k=1}^{M+N}  d \eta_k  \Big) \, \eta_i \eta_j \,   \exp \Big\{ -\frac{1}{2} \bbeta^\intercal \big( z \bI - \cY \big) \bbeta \Big\}}{\int \Big( \prod_{k=1}^{M+N}  d \eta_k  \Big) \,    \exp \Big\{ -\frac{1}{2} \bbeta^\intercal \big( z \bI - \cY \big) \bbeta \Big\}}
\end{split}
\label{Inverse-G-integral}
\end{equation}
For $z$ not close to the real axis, the resolvent is expected to exhibit self-averaging behavior in the limit of large N, meaning that it will not depend on the particular matrix realization. Thus, we can examine the resolvent  $\bG_{\mathcal{Y}}(z)$ by analyzing its ensemble average, denoted by $\langle . \rangle$ in the following.
\begin{equation}
    \big\langle G_{ij}(z) \big\rangle = \bigg\langle \frac{1}{\mathcal{Z}} \, \int \Big( \prod_{k=1}^{M+N}  d \eta_k  \Big) \, \eta_i \eta_j \,   \exp \Big\{ -\frac{1}{2} \bbeta^\intercal \big( z \bI - \cY \big) \bbeta \Big\} \bigg\rangle
    \label{Inverse-G-integral-average}
\end{equation}
where $\mathcal{Z}$ is the denominator in \eqref{Inverse-G-integral}. Computing the average is, in general, non-trivial. However, the replica method provides us with a technique to overcome this issue by employing the following identity:
\begin{equation}
    \begin{split}
        \big\langle G_{ij}(z) \big\rangle &=  \lim_{n \to 0} \bigg\langle \mathcal{Z}^{n-1} \, \int \Big( \prod_{k=1}^{M+N}  d \eta_k  \Big) \, \eta_i \eta_j \,   \exp \Big\{ -\frac{1}{2} \bbeta^\intercal \big( z \bI - \cY \big) \bbeta \Big\} \bigg\rangle \\
        &=  \lim_{n \to 0} \bigg\langle  \, \int \Big( \prod_{k=1}^{M+N} \prod_{\tau=1}^{n}  d \eta^{(\tau)}_k  \Big) \, \eta^{(1)}_i \eta^{(1)}_j \,   \exp \Big\{ -\frac{1}{2} \sum_{\tau =1}^n {\bbeta^{(\tau)}}^\intercal \big( z \bI - \cY \big) \bbeta^{(\tau)} \Big\} \bigg\rangle
    \end{split}
    \label{resolvent-replica}
\end{equation}
For the expression in the exponent, we have:
\begin{equation}
    \begin{split}
        -\frac{1}{2} \sum_{\tau=1}^n \sum_{k,l=1}^{M+N} \eta_k^{(\tau)} \big( z \delta_{kl} - \cY_{kl} \big) \eta_l^{(\tau)} = -\frac{1}{2} \sum_{\tau=1}^n \sum_{k,l=1}^{M+N} \eta_k^{(\tau)} \big( z \delta_{kl} - \cS_{kl} \big) \eta_l^{(\tau)} + \frac{1}{2} \sum_{\tau=1}^n \sum_{k,l=1}^{M+N} \eta_k^{(\tau)} (\bO \cZ \bO^\intercal )_{kl} \eta_l^{(\tau)} 
    \end{split}
    \label{exp1}
\end{equation}
The first term in the RHS can be written as
\begin{equation}
    -\frac{1}{2} \sum_{\tau=1}^n {\bbeta^{(\tau)} }^\intercal  \big( z \bI_{N+M} - \cS \big) \bbeta^{(\tau)}
    \label{exp2}
\end{equation}
Given the structure \eqref{RIE-1} for $\bO \cZ \bO^\intercal $, the second sum in \eqref{exp1} can be written as:
\begin{equation}
    \begin{split}
 \sum_{k=1}^N \sum_{l=N+1}^{M+N} &\eta_k^{(\tau)} \big(\bU \bZ \bV^\intercal )_{k,l-N} \eta_l^{(\tau)} + \sum_{k=N+1}^{M+N} \sum_{l=1}^{N} \eta_k^{(\tau)} \big(\bV \bZ^\intercal  \bU^\intercal )_{k-N,l} \eta_l^{(\tau)}
    \end{split}
    \label{exp3}
\end{equation}
Split each replica $\bbeta^{(\tau)}$ into two vectors $\ba^{(\tau)} \in \bR^N, \bb^{(\tau)} \in \bR^M$, $\bbeta^{(\tau)} = \left[
\begin{array}{c}
\ba^{(\tau)} \\
\bb^{(\tau)}
\end{array}
\right]$. The expression in \eqref{exp3} can be rewritten as $
 2 \Tr \bb^{(\tau)} {\ba^{(\tau)}}^\intercal  \bU \bZ \bV^\intercal  
 $.
So,  we have (dropping the limit term for brevity):
\begin{equation}
    \begin{split}
        \langle  G_{ij}(z) \rangle = \int \bigg( \prod_{k=1}^{M+N} \prod_{\tau =1}^n d \eta_k^{(\tau)}  \bigg) \, \eta^{(1)}_i & \eta^{(1)}_j  \exp \Big\{-\frac{1}{2} \sum_{\tau=1}^n {\bbeta^{(\tau)} }^\intercal  \big( z \bI_{N+M} - \cS \big) \bbeta^{(\tau)} \Big\}  \\
        &\times \bigg\langle \exp \Big\{ \sum_{\tau=1}^n \Tr \bb^{(\tau)} {\ba^{(\tau)}}^\intercal  \bU \bZ \bV^\intercal   \Big\}  \bigg\rangle_{\bU,\bV}
    \end{split}
\end{equation}
 Using the formula for the rectangular spherical integral \cite{benaych2011rectangular} (reviewed as Theorem \ref{rank-one-rect-sphericla-integral} in appendix  \ref{spherical integral app}) for the last term we find:
\begin{equation*}
\begin{split}
    \bigg\langle & \exp \Big\{ \sum_{\tau=1}^n \Tr \bb^{(\tau)} {\ba^{(\tau)}}^\intercal  \bU \bZ \bV^\intercal   \Big\}  \bigg\rangle_{\bU,\bV} \approx \exp \Big\{ \frac{N}{2} \sum_{\tau=1}^n \mathcal{Q}_{\mu_Z} \big(\frac{1}{N M} \| \ba^{(\tau)} \|^2 \| \bb^{(\tau)} \|^2 \big) \Big\}
\end{split}
\end{equation*}
where we used that for each replica, the non-zero singular value of $\bb^{(\tau)} {\ba^{(\tau)}}^\intercal $ is $\| \ba^{(\tau)} \| \| \bb^{(\tau)} \|$.

Therefore, we find
\begin{equation}
    \begin{split}
        &\langle  G_{ij}(z) \rangle = \int \bigg( \prod_{k=1}^{M+N} \prod_{\tau =1}^n d \eta_k^{(\tau)}  \bigg) \, \eta^{(1)}_i \eta^{(1)}_j \\
        &\hspace{3cm} \times \exp \bigg\{ \sum_{\tau=1}^n \Big[  -\frac{1}{2}  {\bbeta^{(\tau)} }^\intercal  \big( z \bI - \cS \big) \bbeta^{(\tau)} +   \frac{N}{2} \mathcal{Q}_{\mu_Z} \big(\frac{1}{N M} \| \ba^{(\tau)} \|^2 \| \bb^{(\tau)} \|^2 \big) \Big] \bigg\} 
    \end{split}
    \label{first-integral}
\end{equation}

Introducing delta functions $\delta \big(p_1^{(\tau)} - \frac{1}{N}\|\ba^{(\tau)}\|^2\big)$, $\delta \big(p_2^{(\tau)} - \frac{1}{M}\|\bb^{(\tau)}\|^2\big)$, \eqref{first-integral} can be written as:
\begin{equation}
    \begin{split}
        \langle  G_{ij}(z) \rangle = \int \bigg( \prod_{k=1}^{M+N} \prod_{\tau =1}^n d \eta_k^{(\tau)}  \bigg) & \bigg( \prod_{\tau =1}^n d p_1^{(\tau)} \,  d p_2^{(\tau)}  \bigg) \, \eta^{(1)}_i \eta^{(1)}_j \\
        & \times  \prod_{\tau=1}^n \delta \big(p_1^{(\tau)} - \frac{1}{N}\|\ba^{(\tau)}\|^2\big) \delta \big(p_2^{(\tau)} - \frac{1}{M}\|\bb^{(\tau)}\|^2\big) \\
        &\hspace{10pt} \times \exp \bigg\{ \sum_{\tau=1}^n \Big[  -\frac{1}{2}  {\bbeta^{(\tau)} }^\intercal  \big( z \bI - \cS \big) \bbeta^{(\tau)} +   \frac{N}{2} \mathcal{Q}_{\mu_Z} \big( p_1^{(\tau)} p_2^{(\tau)} \big) \Big] \bigg\} 
    \end{split}
    \label{second-integral}
\end{equation}
In the next step, we replace each delta with its Fourier transform $\delta \big(p_1^{(\tau)} - \frac{1}{N}\|\ba^{(\tau)}\|^2\big) \propto \int \, d \zeta_1^{(\tau)}  \exp  \Big\{ -\frac{N}{2} \zeta_1^{(\tau)} \big( p_1^{(\tau)} - \frac{1}{N}\|\ba^{(\tau)}\|^2 \big) \Big\}$. After rearranging, we find:
\begin{equation}
    \begin{split}
        \langle  G_{ij}(z) \rangle &\propto  \int \bigg( \prod d p_1^{(\tau)} \, d p_2^{(\tau)} \,  d \zeta_1^{(\tau)} \,  d \zeta_2^{(\tau)}  \bigg) \exp \Big\{ \frac{N}{2} \sum_{\tau=1}^n  \big[  \mathcal{Q}_{\mu_Z} ( p_1^{(\tau)} p_2^{(\tau)} ) - \zeta_1^{(\tau)} p_1^{(\tau)} - \frac{1}{\alpha} \zeta_2^{(\tau)} p_2^{(\tau)} \big] \Big\} \\
        &\times  \int \bigg( \prod_{k=1}^{M+N} \prod_{\tau =1}^n d \eta_k^{(\tau)}  \bigg) \, \eta^{(1)}_i \eta^{(1)}_j  \exp \bigg\{ \sum_{\tau=1}^n \Big[  -\frac{1}{2}  {\bbeta^{(\tau)} }^\intercal  \big( z \bI  - \cS \big) \bbeta^{(\tau)} \\
        &\hspace{7cm}+ \frac{1}{2} \zeta_1^{(\tau)} \|\ba^{(\tau)}\|^2 + \frac{1}{2} \zeta_2^{(\tau)} \|\bb^{(\tau)}\|^2   \Big] \bigg\} 
    \end{split}
    \label{eq15}
\end{equation}

The second integral in \eqref{eq15} is a Gaussian integral with matrix
\begin{equation}
    \bM^{(\tau)} = \left[
\begin{array}{cc}
(z-\zeta_1^{(\tau)}) \bI_N & -\bS  \\
-\bS^\intercal  & (z-\zeta_2^{(\tau)})  \bI_M 
\end{array}
\right]
\label{matrix-M}
\end{equation}
Using the formula for determinant of block matrices, we have 
\begin{equation*}
\begin{split}
    \det \bM^{(\tau)} &= \det \big[ (z-\zeta_1^{(\tau)})\bI_N - (z-\zeta_2^{(\tau)})^{-1} \bS \bS^\intercal  \big] \det \big[ (z-\zeta_2^{(\tau)}) \bI_M \big] \\
    &= (z-\zeta_2^{(\tau)})^{M-N} \prod_{k=1}^N \big[ (z-\zeta_1^{(\tau)})(z-\zeta_2^{(\tau)}) - {\sigma_k}^2 \big]
\end{split}
\end{equation*}

Except for the first replica, the Gaussian integral is (up to constants):
\begin{equation*}
    \exp \bigg\{ -\frac{1}{2} \Big[ (M-N) \ln (z-\zeta_2^{(\tau)}) + \sum_{k=1}^N \ln \big\{ (z-\zeta_1^{(\tau)})(z-\zeta_2^{(\tau)}) - {\sigma_k}^2 \big\} \Big] \bigg\}
\end{equation*}
And, the integral for the first replica is the above expression multiplied by $\big({\bM^{(1)}}^{-1}\big)_{ij}$. By Proposition 2.8.7 \cite{bernstein2009matrix}), ${\bM^{(1)}}^{-1}$ can be written as 
\begin{equation}
   {\bM^{(1)}}^{-1} = \left[
\begin{array}{cc}
\frac{1}{z-\zeta_1^{(1)}} \bI_N +\frac{1}{z-\zeta_1^{(1)}} \bS \bG_{S^\intercal  S} \big((z-\zeta_2^{(1)})(z - \zeta_1^{(1)})\big) \bS^\intercal  &  \bS \bG_{S^\intercal  S} \big((z-\zeta_2^{(1)})(z - \zeta_1^{(1)})\big)  \\
\bG_{S^\intercal  S}\big((z-\zeta_2^{(1)})(z - \zeta_1^{(1)})\big) \bS^\intercal  & (z - \zeta_1^{(1)}) \bG_{S^\intercal  S} \big((z-\zeta_2^{(1)})(z - \zeta_1^{(1)})\big)
\end{array}
\right]
\label{M-inverse}
\end{equation}
with $\bG_{S^\intercal  S}$ the resolvent of the matrix $\bS^\intercal  \bS$. 

Putting all this together, the integral in \eqref{eq15} can be written as 
\begin{equation}
    \begin{split}
        \langle  G_{ij}(z) \rangle &\propto  \int \bigg( \prod_{\tau=1}^n d p_1^{(\tau)} \, d p_2^{(\tau)} \,  d \zeta_1^{(\tau)} \,  d \zeta_2^{(\tau)}  \bigg) \big({\bM^{(1)}}^{-1}\big)_{ij}   \exp \Big\{ - \frac{N n}{2} F_0 \big( \bm{p}_1, \bm{p}_2, \boldsymbol{\zeta}_1, \bm{\zeta}_2 \big)\Big\}
    \end{split}
    \label{eq19}
\end{equation}
with
\begin{equation*}
\begin{split}
        F_0 \big( \bm{p}_1, \bm{p}_2, \bm{\zeta}_1, \bm{\zeta}_2 \big) &= \frac{1}{n} \sum_{\tau=1}^n  \Big[ \frac{1}{N} \sum_{k=1}^N \ln \big\{ (z-\zeta_1^{(\tau)})(z-\zeta_2^{(\tau)}) - {\sigma_k}^2 \big\} - \big( 1 - \frac{1}{\alpha} \big) \ln (z-\zeta_2^{(\tau)})\\
        &\quad \quad \quad \quad  - \mathcal{Q}_{\mu_Z} ( p_1^{(\tau)} p_2^{(\tau)} ) + \zeta_1^{(\tau)} p_1^{(\tau)} + \frac{1}{\alpha} \zeta_2^{(\tau)} p_2^{(\tau)} \Big] 
\end{split}
\end{equation*}
In the large $N$ limit, the integral in \eqref{eq19} can be computed using the saddle-points of the function $F_0$. In the evaluation of this integral, we use the \textit{replica symmetric} ansatz that assumes a saddle-point of the form:
\begin{equation*}
    \forall \tau \in \{1, \cdots, n\}: \quad 
        p_1^{(\tau)} = p_1, \quad p_2^{(\tau)} = p_2, \quad \zeta_1^{(\tau)} = \zeta_1, \quad \zeta_2^{(\tau)} = \zeta_2
\end{equation*}

One finds that the extremum of the function is then at:
\begin{equation}
    \begin{cases}
    p_1^* = (z - \zeta^*_2) \mathcal{G}_{\rho_S}\big( (z-\zeta^*_1)(z-\zeta^*_2) \big) \\
    p_2^* = (1-\alpha ) \frac{1}{z - \zeta^*_2}+\alpha(z - \zeta^*_1) \mathcal{G}_{\rho_S}\big( (z-\zeta^*_1)(z-\zeta^*_2) \big)\\
    \zeta^*_1 = \frac{\mathcal{C}^{(\alpha)}_{\mu_Z}( p_1^* p_2^* )} {p_1^*} \\
    \zeta^*_2 = \alpha \frac{\mathcal{C}^{(\alpha)}_{\mu_Z}( p_1^* p_2^* )} {p_2^*}
    \end{cases}
\end{equation}
where $\mathcal{G}_{\rho_S}$ is the Stieltjes transform of the matrix $\bS \bS^\intercal $,  whose limiting eigenvalue distribution is the squared transform of the limiting singular value distribution of $\bS$.

To simplify the solution, we compute the normalized trace of both sides in \eqref{eq19}. First on the r.h.s we compute the trace of the matrix $\bM^{-1}$ in \eqref{M-inverse} plugging $\zeta^*_1, \zeta^*_2$. The trace of the first block is:
\begin{equation}
    \begin{split}
 \frac{1}{N} \frac{1}{z-\zeta_1^*} \sum_{k=1}^{N} \Big[ 1 +  \frac{{\sigma_k}^2}{(z-\zeta_2^*)(z - \zeta_1^*) - {\sigma_k}^2} \Big] &=\frac{1}{N} (z-\zeta_2^*) \sum_{k=1}^{N}  \frac{1}{(z-\zeta_2^*)(z - \zeta_1^*) - {\sigma_k}^2} \\
        &\approx (z-\zeta_2^*) \mathcal{G}_{\rho_S}\big( (z-\zeta_2^*)(z - \zeta_1^*) \big) \\
        &= p_1^*
    \end{split}
    \label{trof1stblock}
\end{equation}
Similarly, the trace of the last block can be computed to be $p^*_b$.

The matrix in the l.h.s of \eqref{eq19} is $\bG_{\mathcal{Y}}(z)$, which has the blocks 
\begin{equation}
    \begin{split}
        \bG_{\mathcal{Y}}(z) &= \big(z \bI - \cY \big)^{-1} \\
&\hspace{-10pt}= \left[
\begin{array}{cc}
z^{-1} \bI_N + z^{-1} \bY \bG_{Y^\intercal Y}(z^2) \bY^\intercal  & \bY \bG_{Y^\intercal Y}(z^2) \\
\bG_{Y^\intercal Y}(z^2) \bY^\intercal  & z \bG_{Y^\intercal Y}(z^2)
\end{array}
\right]
    \end{split}
    \label{resolvent_Y}
\end{equation}
The trace of the first block is:
\begin{equation}
\begin{split}
\frac{1}{N} \frac{1}{z} \sum_{k=1}^N \big[ 1 + \frac{{\gamma_k}^2}{z^2 - {\gamma_k}^2} \big] &= \frac{1}{N} z \sum_{k=1}^N \frac{1}{z^2 - {\gamma_k}^2} \\
    &\approx z \mathcal{G}_{\rho_Y}(z^2)
\end{split}
\label{trace of first LHS}
\end{equation}
Therefore, from \eqref{trof1stblock}, we find $p_1^* = z \mathcal{G}_{\rho_Y}(z^2)$.
The trace of the last block can be evaluated to be
$\alpha z \mathcal{G}_{\rho_Y}(z^2) + (1- \alpha ) \frac{1}{z}$. So, $p_2^* = \alpha z \mathcal{G}_{\rho_Y}(z^2) + (1- \alpha ) \frac{1}{z}$.

% \noindent\rule[0.5ex]{\linewidth}{2pt}
% To proceed further, let us define some transformations. 

% For a probability measure $\mu$ with support contained in$[-K, K ]$ with $K>0$, we define a generating function of (even) moments $\mathcal{M}_{\mu^2} : [0,K^{-2}] \to \bR_+$
% as
% \begin{equation*}
%     \mathcal{M}_{\mu^2} (z) = \int \frac{1}{1-t^2 z} \, d\mu(t) - 1
% \end{equation*}
% \begin{remark}
% Note that, for a measure $\mu$ with positive support with symmetrized measure $\bar{\mu}$, we have that $\mathcal{M}_{\mu^2} (z) = \mathcal{M}_{\bar{\mu}^2} (z)$ .
% \end{remark}

% For $\tau \in [0,1]$, $T^{(\tau)}(z) = (\tau z +1)(z+1)$, and
% \begin{equation*}
%     \mathcal{H}_{\mu}^{(\tau)}(z) = z T^{(\tau)}\big(\mathcal{M}_{\mu^2} (z)\big)
% \end{equation*}
% And, the \textit{rectangular R-transform} is defined as:
% \begin{equation*}
%     \mathcal{C}_{\mu}^{(\tau)}(z) = {T^{(\tau)}}^{-1}\Big(\frac{z}{{\mathcal{H}_{\mu}^{(\tau)}}^{-1}(z)} \Big)
% \end{equation*}

% For a matrix $\bA \in \bR^{N \times M}$ with $\tau = \frac{N}{M}$ with limiting singular value distribution $\mu$, we have that:
% \begin{equation*}
%     \bg_{AA^\intercal }(z) = \frac{1}{z} \Big( \mathcal{M}_{\mu^2} \big(\frac{1}{z}\big) + 1 \Big) \equiv \frac{1}{z} \Big( \mathcal{M}_{A} \big(\frac{1}{z}\big) + 1 \Big)
% \end{equation*}

% \noindent\rule[0.5ex]{\linewidth}{2pt}

Thus we find
\begin{equation}
    \begin{cases}
    p_1^* = z \mathcal{G}_{\rho_Y}(z^2) = \frac{1}{z} \mathcal{M}_{\mu_Y} \big( \frac{1}{z^2} \big) + \frac{1}{z} \\
    p_2^* = \alpha z \mathcal{G}_{\rho_Y}(z^2) + (1- \alpha ) \frac{1}{z}= \alpha \frac{1}{z} \mathcal{M}_{\mu_Y} \big( \frac{1}{z^2} \big) + \frac{1}{z}
    \end{cases}
\end{equation}
and 
\begin{equation*}
    p_1^* p_2^* = \frac{1}{z^2} T^{(\alpha)} \Big( \mathcal{M}_{\mu_Y} \big( \frac{1}{z^2} \big) \Big),
\end{equation*}
which implies
\begin{equation}
    \zeta^*_1 = z\frac{\mathcal{C}^{(\alpha)}_{\mu_Z}\bigg(\frac{1}{z^2} T^{(\alpha)} \Big( \mathcal{M}_{\mu_Y} \big( \frac{1}{z^2} \big) \Big)\bigg)}{\mathcal{M}_{\mu_Y} \big( \frac{1}{z^2} \big) +1} ,\quad
    \zeta^*_2 =  \alpha z\frac{\mathcal{C}^{(\alpha)}_{\mu_Z}\bigg(\frac{1}{z^2} T^{(\alpha)} \Big( \mathcal{M}_{\mu_Y} \big( \frac{1}{z^2} \big) \Big)\bigg)}{\alpha \mathcal{M}_{\mu_Y} \big( \frac{1}{z^2} \big) +1}
    \label{zeta_sol-app}
\end{equation}
\section{Derivation of Rectangular Free Convolution}\label{free-add-conv}
Consider the normalized trace of the first block on each side in \eqref{resolvent-relation}. The trace of the first block of the lhs is computed in \eqref{trace of first LHS} which is $\frac{1}{z} \mathcal{M}_{\mu_Y}\big( \frac{1}{z^2} \big) + \frac{1}{z}$. The trace of the first block in rhs is computed in \eqref{trof1stblock} which is $(z-\zeta_2^*) \mathcal{G}_{\rho_S}\big( (z-\zeta_2^*)(z - \zeta_1^*) \big)$.
\begin{equation*}
    \begin{split}
        \frac{1}{z} \mathcal{M}_{\mu_Y}\big( \frac{1}{z^2} \big) + \frac{1}{z} &= (z-\zeta_2^*) \mathcal{G}_{\rho_S}\big( (z-\zeta_2^*)(z - \zeta_1^*) \big) \\
        &= (z-\zeta_2^*) \frac{1}{(z-\zeta_2^*)(z - \zeta_1^*)} \Big(\mathcal{M}_{\mu_S} \big( \frac{1}{(z-\zeta_2^*)(z - \zeta_1^*)} \big) + 1 \Big) \\
        &= \frac{1}{z - \zeta_1^*} \mathcal{M}_{\mu_S} \big( \frac{1}{(z-\zeta_2^*)(z - \zeta_1^*)} \big) + \frac{1}{z - \zeta_1^*}
    \end{split}
\end{equation*}
From which, we get:
\begin{equation*}
    (z - \zeta_1^*) \mathcal{M}_{\mu_Y}\big( \frac{1}{z^2} \big) + z - \zeta_1^* = z \mathcal{M}_{\mu_S} \big( \frac{1}{(z-\zeta_2^*)(z - \zeta_1^*)} \big) + z
\end{equation*}

Taking the $\zeta_1^*$ to the rhs, and plugging the expression for $\zeta_1^*$ from \eqref{zeta_sol-app}, after a bit of algebra we find:
% \begin{equation*}
%     z \Bigg(1 - \frac{\mathcal{C}^{(\alpha)}_{\bZ}\bigg(\frac{1}{z^2} T^{(\alpha)} \Big( \mathcal{M}_{\mu_Y}\big( \frac{1}{z^2} \big) \Big)\bigg)}{\mathcal{M}_{\mu_Y}\big( \frac{1}{z^2} \big) +1} \Bigg) \mathcal{M}_{\mu_Y}\big( \frac{1}{z^2} \big) = z \Bigg( \mathcal{M}_X \big( \frac{1}{(z-\zeta_2^*)(z - \zeta_1^*)} \big) + \frac{\mathcal{C}^{(\alpha)}_{\bZ}\bigg(\frac{1}{z^2} T^{(\alpha)} \Big( \mathcal{M}_{\mu_Y}\big( \frac{1}{z^2} \big) \Big)\bigg)}{\mathcal{M}_{\mu_Y}\big( \frac{1}{z^2} \big) +1} \Bigg)
% \end{equation*}
% which implies
% \begin{equation}{\label{eq27}}
%     \begin{split}
%             &\mathcal{M}_{\mu_Y}\big( \frac{1}{z^2} \big) \\
%         &= \mathcal{M}_{\mu_S} \big( \frac{1}{(z-\zeta_2^*)(z - \zeta_1^*)} \big) + \mathcal{C}^{(\alpha)}_{\mu_Z}\bigg(\frac{1}{z^2} T^{(\alpha)} \Big( \mathcal{M}_{\mu_Y}\big( \frac{1}{z^2} \big) \Big)\bigg)
%     \end{split}
% \end{equation}
\begin{equation} \label{eq27}
    \mathcal{M}_{\mu_Y}\big( \frac{1}{z^2} \big) = \mathcal{M}_{\mu_S} \big( \frac{1}{(z-\zeta_2^*)(z - \zeta_1^*)} \big) + \mathcal{C}^{(\alpha)}_{\mu_Z}\bigg(\frac{1}{z^2} T^{(\alpha)} \Big( \mathcal{M}_{\mu_Y}\big( \frac{1}{z^2} \big) \Big)\bigg) \notag
\end{equation}
% \begin{remark}
% If we consider the trace of the last block, we find the same relation as in \eqref{eq27}.
% \end{remark}

Let $\frac{1}{z^2} T^{(\alpha)} \Big( \mathcal{M}_{\mu_Y}\big( \frac{1}{z^2} \big) \Big) = u$. Then, $\frac{1}{z^2} =  {\mathcal{H}_{\mu_Y}^{(\alpha)}}^{-1}(u)$. Moreover, from the definition one can see that $\mathcal{M}_{\mu_Y}\Big( {\mathcal{H}_{\mu_Y}^{(\alpha)}}^{-1}(u) \Big) = \mathcal{C}^{(\alpha)}_{\mu_Y} (u)$. So, \eqref{eq27} can be written as:
\begin{equation}
    \mathcal{C}^{(\alpha)}_{\mu_Y} (u) = \mathcal{M}_{\mu_S} \big( \frac{1}{(z-\zeta_2^*)(z - \zeta_1^*)} \big) + \mathcal{C}^{(\alpha)}_{\mu_Z}(u)
    \label{relation for u}
\end{equation}

From \eqref{zeta_sol-app},
\begin{equation*}
    \begin{split}
        (z-\zeta_2^*)(z - \zeta_1^*) &= z^2 \Big( 1 - \frac{\mathcal{C}^{(\alpha)}_{\mu_Z}(u)}{\mathcal{C}^{(\alpha)}_{\mu_Y}(u) + 1} \Big) \Big( 1 - \frac{\alpha \mathcal{C}^{(\alpha)}_{\mu_Z}(u)}{\alpha \mathcal{C}^{(\alpha)}_{\mu_Y}(u) + 1} \Big) \\
        &= z^2 \bigg[ 1 - \mathcal{C}^{(\alpha)}_{\mu_Z}(u) \Big( \frac{1}{\mathcal{C}^{(\alpha)}_{\mu_Y}(u) + 1} + \frac{\alpha}{\alpha \mathcal{C}^{(\alpha)}_{\mu_Y}(u) + 1} \Big) + \frac{\alpha \big( \mathcal{C}^{(\alpha)}_{\mu_Z}(u)\big)^2}{T^{(\alpha)}\big( \mathcal{C}^{(\alpha)}_{\mu_Z}(u)\big)} \bigg] \\
        &= \frac{z^2}{T^{(\alpha)}\big( \mathcal{C}^{(\alpha)}_{\mu_Y}(u)\big)} T^{(\alpha)} \big(  \mathcal{C}^{(\alpha)}_{\mu_Y}(u) -  \mathcal{C}^{(\alpha)}_{\mu_Z}(u) \big)
    \end{split}
\end{equation*}
The first factor, using the definition of $ \mathcal{C}^{(\alpha)}_{\mu_Y}(u)$, is:
\begin{equation*}
    \begin{split}
        \frac{z^2}{T^{(\alpha)}\big( \mathcal{C}^{(\alpha)}_{\mu_Y}(u)\big)} &= \frac{1}{\frac{1}{z^2}} \frac{1}{T^{(\alpha)}\big( \mathcal{C}^{(\alpha)}_{\mu_Y}(u)\big)} \\
        &= \frac{1}{{\mathcal{H}_{\mu_Y}^{(\alpha)}}^{-1}(u)} \frac{1}{\frac{u}{{\mathcal{H}_{\mu_Y}^{(\alpha)}}^{-1}(u)}} = \frac{1}{u}
    \end{split}
\end{equation*}

So, \eqref{relation for u} can be written as
\begin{equation*}
    \mathcal{C}^{(\alpha)}_{\mu_Y}(u) -  \mathcal{C}^{(\alpha)}_{\mu_Z}(u) = \mathcal{M}_{\mu_S} \Big( \frac{u}{T^{(\alpha)} \big(  \mathcal{C}^{(\alpha)}_{\mu_Y}(u) -  \mathcal{C}^{(\alpha)}_{\mu_Z}(u) \big)} \Big)
\end{equation*}

One can see that, if the limiting singular value distribution of $\bS$, is not $\delta(x)$, the unique solution to the equation $ \mathcal{M}_{\mu_S}  \big( \frac{u}{T^{(\alpha)}(x)} \big) = x$, is $x = \mathcal{C}^{(\alpha)}_{\mu_S}(u)$ (see lemma 4.2 in \cite{benaych2011rectangular} for a particular case). Therefore, we find:
\begin{equation}
    \mathcal{C}^{(\alpha)}_{\mu_Y}(u) -  \mathcal{C}^{(\alpha)}_{\mu_Z}(u) = \mathcal{C}^{(\alpha)}_{\mu_S}(u)
\end{equation}
as we expected.
\section{Details of proof of Theorem 5}\label{proof-thm1-app}
\subsection{Proof of proposition \ref{asymp-free-energy-2-proposition}}\label{proof of prop1}
We start from the partition function,
\begin{equation}\label{Z2-comp1}
    \begin{split}
        \tilde{Z}(\tilde{\bY}) &=  \int d \bX e^{N \Tr \big[ \sqrt{\lambda}\bX^\intercal  \tilde{\bY} - \frac{\lambda}{2} \bX^\intercal  \bX \big] } P_{\tilde{S}}(\bX)  \\
        &= \iiint d \bsig \,d \mu_N(\bU) \, d \mu_M(\bV) \, \prod_{i=1}^N \delta(\tilde{\sigma}_i - \sigma^0_i) \,  e^{N \Tr [ \sqrt{\lambda} \bV \tilde{\bSig}^\intercal  \bU^\intercal  \tilde{\bY} - \frac{\lambda}{2}\tilde{\bSig} \tilde{\bSig}^\intercal  ]}\\
        &= e^{ - \frac{N}{2} \lambda \Tr {\bSig^0}^\intercal  {\bSig^0} } \iint d \mu_N(\bU) \,  d \mu_M(\bV) \, e^{N \Tr [ \sqrt{\lambda} {\bSig^0}^\intercal  \bU  \tilde{\bY} \bV^\intercal  ]} \\
        &=  e^{ - \frac{N}{2} \lambda \Tr  {\bSig^0}^\intercal  {\bSig^0}  } \mathcal{I}_{N,M} \big( \sqrt{\lambda} \bSig^0, \tilde{\bY} \big)
    \end{split}
\end{equation}

where, we change variables $\bU \to \bU^\intercal , \bV \to \bV^\intercal $ in third line to match the definition of the spherical integral.
% with $
% \tilde{\bSig^0} = \left[ \begin{array}{c|c}
% {\rm diag}(\bsig^0) & \mathbf{0}_{N \times (M-N)}
% \end{array}
% \right]
% $. 
% \begin{equation}
%     \begin{split}
%         \tilde{Z}(\tilde{\bY}) &= \int d \bX e^{\frac{N}{2} \Tr \big[ \sqrt{\gamma}\bX \tilde{\bY} - \frac{\gamma}{2} \bX^2 \big] } P_{\tilde{S},N}(\bX) \\
%         &= \int d \blam \, d \mu_N(\bU) \, \prod_{i=1}^N \delta(\lambda_i - \lambda^0_i) \, e^{\frac{N}{2} \Tr [ \sqrt{\gamma} \bU \bLam \bU^\intercal  \tilde{\bY} - \frac{\gamma}{2}\bLam^2 ]} \\
%         &= \int d \blam \prod_{i=1}^N \delta(\lambda_i - \lambda^0_i) \, e^{ - \frac{N}{4}\gamma \Tr \bLam^2 } \int d \mu_N(\tilde{\bU}) e^{\frac{N}{2}   \Tr  \bU \big(\sqrt{\gamma}\bLam \big)  \bU^\intercal  \tilde{\bY}} \\
%         &= \int d \blam \prod_{i=1}^N \delta(\gamma_i - \gamma^0_i) \, e^{ - \frac{N}{4} \gamma \Tr \bLam^2 } I_N \big( \sqrt{\gamma} \bLam, \tilde{\bY} \big) \\
%         &= e^{ - \frac{N}{4} \gamma \Tr {\bLam^0}^2 } I_N \big( \sqrt{\gamma} \bLam^0, \tilde{\bY} \big)
%     \end{split}
% \label{Z2-comp1}
% \end{equation}

Recall that $\tilde{\bY} = \sqrt{\lambda} \bU \bSig^0 \bV^\intercal  + \tilde{\bZ}$, so with $\mathcal{J}_{N,M} \big( \sqrt{\lambda} \bSig^0, \tilde{\bY} \big) \equiv \frac{1}{N M} \ln \mathcal{I}_{N,M} \big( \sqrt{\lambda} \bSig^0, \tilde{\bY} \big) $ the free energy can be written as:%\vspace{-2mm}
\begin{equation}\label{F2-com}
    \begin{split}
        \tilde{F}_N(\lambda) &= \bE_{\tilde{\bY}} \Big[ \frac{\lambda}{2 M} \Tr {\bSig^0}^\intercal  {\bSig^0}  - \mathcal{J}_{N,M} \big( \sqrt{\lambda} \bSig^0, \tilde{\bY} \big) \Big]  \\
        &=  \frac{\lambda}{2 M} \sum_{i=1}^N {\sigma^0_i}^2 - \bE_{\bU,\bV,\tilde{\bZ}} \Big[ \mathcal{J}_{N,M}\big( \sqrt{\lambda} \bSig^0, \sqrt{\lambda} \bU \bSig^0 \bV^\intercal  + \tilde{\bZ} \big) \Big]
    \end{split}
\end{equation}
By bi-rotational invariance of $\tilde{\bZ}$, the second term equals %\vspace{-2mm}
\begin{equation*}
\begin{split}
        \bE_{\bU,\bV,\tilde{\bZ}} \Big[ \mathcal{J}_{N,M}\big( \sqrt{\lambda} \bSig^0, \sqrt{\lambda} \bU \bSig^0 \bV^\intercal  + \bU \tilde{\bZ} \bV^\intercal  \big) \Big]
\end{split}
\end{equation*}
and then, both matrices $\bU, \bV$ can be absorbed into the integration in $\mathcal{J}_{N,M}$. So, the free energy equals:
\begin{equation*}
     \tilde{F}_N(\lambda) = \frac{\lambda}{2 M} \sum_{i=1}^N {\sigma^0_i}^2 - \bE_{\tilde{\bZ}} \Big[ \mathcal{J}_{N,M}\big( \sqrt{\lambda} \bSig^0, \sqrt{\lambda} \bSig^0  +  \tilde{\bZ}\big) \Big]
\end{equation*}

By the strong law of large numbers, the first term in \eqref{F2-com} converges to $\frac{\lambda}{2} \alpha \int x^2 \mu_S(x) \, d x$ almost surely, and proposition \ref{asymp-free-energy-2-proposition} follows from the following lemma.
\begin{lemma}
For any $\lambda \in \bR_+$, the sequence $\bE_{\tilde{\bZ}} \Big[ \mathcal{J}_{N,M}\big( \sqrt{\lambda} \bSig^0, \sqrt{\lambda} \bSig^0  +  \tilde{\bZ}\big) \Big]$ converges to $\mathcal{J}[\mu_{\sqrt{\lambda} S}, \mu_{\sqrt{\lambda} S}\boxplus_{\alpha} \mu_{\rm MP}]$ as $N \to \infty$, $\mu_S$-almost surely.
\label{converg-Expec}
\end{lemma}
%\vspace{-5mm}
\begin{proof}
We first show that the assumptions of Theorem 1.1 in \cite{guionnet2021large} holds a.s. for the sequence $\sqrt{\lambda} \bSig^0, \sqrt{\lambda} \bSig^0  +  \tilde{\bZ}$, so $\mathcal{J}_{N,M}$ converges to $\mathcal{J}$ a.s. .

By assumption \ref{assumptions on law}, the symmetrized ESD of $\sqrt{\lambda} \bSig^0$ converges weakly to $\bar{\mu}_{\sqrt{\lambda}S}$ by construction.  The ESD of $\tilde{\bZ}$ converges a.s. to the  MP law $\mu_{\rm MP}$, so by independence of $\bSig^0, \tilde{\bZ}$ the limiting ESD of $\sqrt{\lambda} \bSig^0  +  \tilde{\bZ}$ is the rectangular free convolution of $\mu_{\sqrt{\lambda}S}, \mu_{\rm MP}$ denoted by $\mu_{\sqrt{\lambda} S}\boxplus_{\alpha} \mu_{\rm MP}$ . Moreover, by assumptions \ref{assumptions on law}, \ref{bounded-mom}, the second moment of ESD of $\sqrt{\lambda} \bSig^0$ is boounded and $\bar{\mu}_S$ has finite non-commutative entropy $\iint \ln | x - y |  \, d \bar{\mu}_S(x) \, d \bar{\mu}_S(y) > -\infty$, and $\int \ln |x| \, d \bar{\mu}_S(x) > - \infty$. Therefore, the sequence $\mathcal{J}_{N,M}\big( \sqrt{\lambda} \bSig^0, \sqrt{\lambda} \bSig^0  +  \tilde{\bZ}\big)$ converges a.s. to $\mathcal{J}[\mu_{\sqrt{\lambda} S}, \mu_{\sqrt{\lambda} S}\boxplus_{\alpha} \mu_{\rm MP}]$.

Now, we prove that the limit also holds under the expectation $\bE_{\tilde{\bZ}}$. For simplicity of notation we denote $ \mathcal{J}_{N,M}\big( \sqrt{\lambda} \bSig^0, \sqrt{\lambda} \bSig^0  +  \tilde{\bZ}\big) $ by $\mathcal{J}_{N}$, and $\mathcal{J}[\mu_{\sqrt{\lambda} S}, \mu_{\sqrt{\lambda} S}\boxplus_{\alpha} \mu_{\rm MP}]$ by $\mathcal{J}$. By Jensen's inequality (note that the expectation is over the matrix $\tilde{\bZ}$), we have
\begin{equation}
\big| \bE [ \mathcal{J}_{N} ] - \mathcal{J} \big| \leq \bE \big[ | \mathcal{J}_{N} - \mathcal{J} | \big].
\label{Jensen-expec}
\end{equation}
Let $X_N \equiv \mathcal{J}_{N} - \mathcal{J}$. For $\epsilon > 0$ We can write
\begin{equation}
\begin{split}
    \bE \big[ |X_N| \big] &= \bE \big[ |X_N| \, \mathbb{I} \{ |X_N| \leq \epsilon \} \big] + \bE \big[ |X_N| \,\mathbb{I} \{ |X_N| > \epsilon \} \big] \\
    &\leq \epsilon + \bE \big[ |X_N| \, \mathbb{I} \{ |X_N| > \epsilon \} \big].
\label{converg-Expec-1}
\end{split}
\end{equation}
By lemma \ref{bound-on-J}, $ | \mathcal{J}_{N} | \leq \sqrt{\lambda} C_2 \tilde{\gamma}_N$, where $\tilde{\gamma}_N$ is the top singular value of $ \sqrt{\lambda} \bSig^0  +  \tilde{\bZ}$. The second term in \eqref{converg-Expec-1} can be bounded as, 
\begin{equation}
    \bE \big[ |X_N| \, \mathbb{I} \{ |X_N| > \epsilon \} \big] \leq \bE \big[ |W_N| \, \mathbb{I} \{ |X_N| > \epsilon \} \big]
\label{converg-Expec-2}
\end{equation}
where 
$$
W_N = \max \Big \{ \big| \mathcal{J} -  \sqrt{\lambda} C_2 \tilde{\gamma}_N \big|, \big| \mathcal{J} +  \sqrt{\lambda} C_2 \tilde{\gamma}_N \big| \Big\} =  \sqrt{\lambda} C_2 \tilde{\gamma}_N + {\rm sign}(\mathcal{J}) \mathcal{J} .
$$
For any positive constant $t$, we have
\begin{equation}
\begin{split}
    \bE \big[ |W_N| \, \mathbb{I} \{ |X_N| > \epsilon \} \big] &= \bE \big[ |W_N| \, \mathbb{I} \{ |X_N| > \epsilon \} \, \mathbb{I} \{ |W_N| \leq t \} \big] + \bE \big[ |W_N| \, \mathbb{I} \{ |X_N| > \epsilon \} \, \mathbb{I} \{ |W_N| > t \} \big] \\
    & \leq \bE \big[ |W_N| \, \mathbb{I} \{ |X_N| > \epsilon \} \mathbb{I} \{ |W_N| \leq t \} \big] + \bE \big[ |W_N| \, \mathbb{I} \{ |W_N| > t \} \big]
\end{split}
\label{converg-Expec-3}
\end{equation}
For the first term in \eqref{converg-Expec-3} we can write
\begin{equation}
\begin{split}
    \bE \big[ |W_N| \, \mathbb{I} \{ |X_N| > \epsilon \} \mathbb{I} \{ |W_N| \leq t \} \big] &\leq t \bE \big[ \mathbb{I} \{ |X_N| > \epsilon \} \big] \\
    &\leq t \, \mathbb{P} \big( |X_N| > \epsilon  \big)
\end{split}
\label{converg-Expec-4}
\end{equation}
and the second term in \eqref{converg-Expec-3} can be rewritten as 
\begin{equation}
    \bE \big[ |W_N| \, \mathbb{I} \{ |W_N| > t \} \big] = \bE \bigg[ |W_N| \, \mathbb{I} \Big\{ \tilde{\gamma}_N > \frac{2}{\sqrt{\lambda} C_2} \big(t-{\rm sign}(\mathcal{J}) \mathcal{J}\big) \Big\} \bigg].
\label{converg-Expec-5}
\end{equation}
From \eqref{converg-Expec-2}, \eqref{converg-Expec-3}, \eqref{converg-Expec-4}, we obtain
\begin{equation}
    \bE \big[ |X_N| \, \mathbb{I} \{ |X_N| > \epsilon \} \big] \leq t \,  \mathbb{P} \big( |X_N| > \epsilon  \big) + \bE \bigg[ |W_N| \, \mathbb{I} \Big\{ \tilde{\gamma}_N > \frac{2}{\sqrt{\lambda} C_2} \big(t-{\rm sign}(\mathcal{J}) \mathcal{J}\big) \Big\} \bigg].
\label{converg-Expec-6}
\end{equation}
Notice that $W_N$ is a polynomial function of $\tilde{\gamma}_N$, so by lemma \ref{expec-poly-r}, vanishes as $N \to \infty$ for sufficiently large constant $t$. By almost sure convergence of $\mathcal{J}_N$ to $\mathcal{J}$, $ \mathbb{P} \big( |X_N| > \epsilon  \big) \xrightarrow{N \to \infty} 0$. For a fixed $t>0$, the first term in \eqref{converg-Expec-6} goes to $0$ in the limit $ N \to \infty$. Therefore, taking the limit of both sides in \eqref{converg-Expec-1}, for any $\epsilon >0$, we find:
\begin{equation}
    \lim_{N \to \infty} \bE \big[ |X_N| \big] \leq \epsilon .
\end{equation}
From which, by \eqref{Jensen-expec}, we deduce that $\lim_{ N \to \infty} \bE [\mathcal{J}_{N}] = \mathcal{J}$.
\end{proof} 

\subsubsection{Technical Lemmas}
\begin{lemma}
For any $N$, $M\geq N$, and $N \times M$ symmetric matrices $\bA, \bB$ with top singular values $\sigma_N^{A}, \sigma_N^{B}$
\begin{equation*}
    -  \sigma_N^{A} \sigma_N^{B} \leq \mathcal{J}_{N,M} (\bA, \bB) \leq  \sigma_N^{A} \sigma_N^{B}.
\end{equation*}
\label{bound-on-J}
\end{lemma}
% \vspace{-2\baselineskip}
\begin{proof}
Let $\bA = \bU_{A} \bSig_{A} \bV^\intercal_{A}$,  $\bB = \bU_{B} \bSig_{B} \bV^\intercal_{B}$ be the SVD of $\bA$, $\bB$. We can write
\begin{equation*}
\begin{split}
    \mathcal{I}_{N,M} (\bA, \bB) = \iint D \bU \, D \bV \,  e^{N \Tr  \bSig_A^\intercal \bU \bSig_B \bV^\intercal} = \iint D \bU \, D \bV \,  e^{N \sum_{i,j=1}^N \sigma_i^A \sigma_j^B U_{ij} V_{ij}}
\end{split}
\end{equation*}
The term in the exponent can be bounded as:
\begin{equation}
    \sum_{i,j=1}^N \sigma_i^A \sigma_j^B U_{ij} V_{ij} \leq  \sum_{i,j=1}^N \big| \sigma_i^A \sigma_j^B U_{ij} V_{ij} \big| \leq \sigma_N^A \sigma_N^B  \sum_{i,j=1}^N  \big| U_{ij} V_{ij} \big| \leq \|\bU \|_{\rm F} \| \bV \|_{\rm F}  \sigma_N^A \sigma_N^B \leq \sqrt{N M} \sigma_N^A \sigma_N^B
\end{equation}
which implies $\mathcal{I}_{N,M} (\bA, \bB) \leq e^{N \sqrt{N M} \sigma_N^A \sigma_N^B }$. Similarly, we get $\mathcal{I}_{N,M} (\bA, \bB) \geq e^{-N \sqrt{N M} \sigma_N^A \sigma_N^B }$. Therefore, we obtain
\begin{equation*}
    - \sqrt{\frac{N}{M}} \sigma_N^A \sigma_N^B \leq \mathcal{J}_{N,M} (\bA, \bB) \leq \sqrt{\frac{N}{M}} \sigma_N^A \sigma_N^B
\end{equation*}
The result follows since $N \leq M$.
\end{proof}

\begin{lemma}
Let $\tilde{\gamma}_N$ be the top singular value of the matrix $\sqrt{\lambda} \bSig^0  +  \tilde{\bZ}$. For $k > 1 + \sqrt{K} + \sqrt{\lambda} C_2 $, we have
\begin{equation*}
    \begin{split}
        \mathbb{P} \{ \tilde{\gamma}_N \geq k \} \leq e^{-\frac{N}{2}(k-\sqrt{\lambda}C_2 - 1 - \sqrt{K})^2}.
    \end{split}
\end{equation*}
\label{bound-on-prob-of-spectral-radius}
\end{lemma}
\begin{proof}
By triangle inequality, we have:
\begin{equation*}
    \tilde{\gamma}_N \leq \sqrt{\lambda} \max \sigma_i + \gamma^{\tilde{Z}}_N \leq  \sqrt{\lambda}  C_2 + \gamma^{\tilde{Z}}_N
\end{equation*}
where $\gamma^{\tilde{Z}}_N$ is the top singular value of $\tilde{\bZ}$. Thus, we can write
\begin{equation*}
    \begin{split}
        \mathbb{P} \{ \tilde{\gamma}_N \geq k \} & \leq \mathbb{P} \big\{  \gamma^{\tilde{Z}}_N + \sqrt{\lambda} C_2 \geq k \big\} \\
        & = \mathbb{P} \{ \gamma^{\tilde{Z}}_N \geq k -  \sqrt{\lambda} C_2\}.
    \end{split}
\end{equation*}
By \cite{davidson2001local} (Theorem II.13), for $k -  \sqrt{\lambda} C_2 > 1 + \sqrt{K} > 1 + \sqrt{\nicefrac{M}{N}} $, we have
\begin{equation*}
    \begin{split}
        \mathbb{P} \{ \tilde{\gamma}_N \geq k -  \sqrt{\lambda} C_2 \}  \leq e^{-\frac{N}{2}(k-\sqrt{\lambda}C_2 - 1 - \sqrt{K} )^2}
    \end{split}
\end{equation*}
and therefore we get the result.
\end{proof}

\begin{lemma}
For any polynomial function $g$, and $k$ a sufficiently large constant, we have that 
$$
\lim_{n \to \infty} \bE \big[ g(\tilde{\gamma}_N) \mathbb{I}\{\tilde{\gamma}_N \geq k\} \big] = 0 .
$$
\label{expec-poly-r}
\end{lemma}
\begin{proof}
    See Lemma H.3 in \cite{pourkamali2023matrix}.
\end{proof}

\subsection{Proof of proposition \ref{pseudo-lip}}\label{proof of prop2}
Consider two matrices with the same singular vectors, $\bS = \bU \bSig \bV^\intercal $, $\tilde{\bS} = \bU \tilde{\bSig} \bV^\intercal $, where $\bU \in \bR^{N \times N}, \bV \in \bR^{M \times M}$ are Haar orthogonal matrices, and $\bsig$, $\tilde{\bsig}$ are distributed according to $P_N^{(1)}(\bsig), P_N^{(2)}(\tilde{\bsig})$, respectively. For two such matrices, we write $(\bS, \tilde{\bS}) \sim Q_{N,M}(\bU, \bV,  \bsig, \tilde{\bsig})$ which is the joint p.d.f. of $\bU, \bV, \bsig, \tilde{\bsig}$,
\begin{equation*}
d Q_{N,M}(\bU, \bV,  \bsig, \tilde{\bsig})  = d \mu_N(\bU) \, d \mu_M(\bV) \, P_N^{(1)}(\bsig) \, d \bsig \, P_N^{(2)}(\tilde{\bsig}) \, d \tilde{\bsig}  
\end{equation*}
For $t \in [0,1]$, consider the following observation model:
\begin{equation}
    \begin{cases}
  \bY_1^{(t)} = \sqrt{\lambda t}\bS + \bZ_1\\
  \bY_2^{(t)} = \sqrt{\lambda (1-t)}\tilde{\bS} + \bZ_2
\end{cases}
\label{int-model}
\end{equation}
where $\bZ_1, \bZ_2 \in \bR^{N \times M}$ are Gaussian matrices as in \eqref{model}, independent of each other. $(\bS, \tilde{\bS}) \sim Q_{N,M}(\bU, \bV,  \bsig, \tilde{\bsig})$. The free energy for this model can be written as 
\begin{equation}
\begin{split}
        F_N(t) &= -\frac{1}{M N} \bE_{\bY_1^{(t)}, \bY_2^{(t)}} \bigg[ \ln \int d Q_{N,M}(\bU, \bV,  \bsig, \tilde{\bsig}) e^{N \Tr [\sqrt{\lambda t} \bX^\intercal  \bY_1^{(t)} - \frac{\lambda t}{2} \bX^\intercal  \bX +  \sqrt{\lambda (1-t)} \tilde{\bX}^\intercal  \bY_2^{(t)} - \frac{\lambda (1-t)}{2} \tilde{\bX}^\intercal  \tilde{\bX} ]} \bigg] \\
        &= -\frac{1}{M N} \bE_{\bY_1^{(t)}, \bY_2^{(t)}} \bigg[ \ln \int d Q_{N,M}(\bU, \bV,  \bsig, \tilde{\bsig}) \\
        &\hspace{80pt} \times e^{N \Tr [\lambda t \bX^\intercal  \bS + \sqrt{\lambda t} \bX^\intercal  \bZ_1 - \frac{\lambda t}{2} \bX^\intercal  \bX + \lambda (1-t) \tilde{\bX}^\intercal  \tilde{\bS}+  \sqrt{\lambda (1-t)} \tilde{\bX}^\intercal  \bZ_2 - \frac{\lambda (1-t)}{2} \tilde{\bX}^\intercal  \tilde{\bX}]} \bigg]
\end{split}
\label{free-energy-combined}
\end{equation}
 where the singular vectors of $\bX, \tilde{\bX}$ are the same, $\bX = \bU \bSig \bV^\intercal $, $\tilde{\bX} = \bU \tilde{\bSig} \bV^\intercal $.
Note that, for $t=0$ the only term depending on $\bsig$ (in both the inner and outer expectation) is the pdf $P_N^{(1)}(\bsig)$ and we can integrate over $\bsig$ in both of the expectations, to get $F_N(0) = F_N^{(2)}(\lambda)$. Similarly, we have $F_N(1) = F_N^{(1)}(\lambda)$.

Taking the derivative w.r.t. $t$, we get 
\begin{equation}
\begin{split}
    \frac{d}{d t}F_N(t) = -\frac{1}{M} \bE \Big[ \lambda  \Tr \langle \bX^\intercal  \bS \rangle_t +& \frac{1}{2} \sqrt{\frac{\lambda}{t}} \Tr \bZ_1^\intercal  \langle \bX \rangle_t - \frac{\lambda}{2} \Tr \langle \bX^\intercal  \bX \rangle_t  \\
    &- \lambda  \Tr \langle \tilde{\bX}^\intercal  \tilde{\bS} \rangle_t - \frac{1}{2} \sqrt{\frac{\lambda}{1-t}} \Tr \bZ_2^\intercal  \langle \tilde{\bX} \rangle_t + \frac{\lambda}{2} \Tr \langle \tilde{\bX}^\intercal  \tilde{\bX} \rangle_t \Big] 
\end{split}
    \label{time-deivative}
\end{equation}
where $\langle . \rangle_t$ denotes the expectation with respect to the posterior distribution of the model \eqref{int-model}. By integration by parts, we have
\begin{equation*}
\begin{split}
    \bE \big[ \Tr \bZ_1^\intercal  \langle \bX \rangle_t ] &=  \sqrt{\lambda t} \bE \Big[ \Tr \langle \bX^\intercal  \bX \rangle_t  - \Tr \langle \bX \rangle_t^\intercal  \langle \bX \rangle_t \Big]\\
    \bE \big[ \Tr \bZ_2^\intercal  \langle \tilde{\bX} \rangle_t ] &\hspace{-2pt}= \hspace{-2pt}\sqrt{\lambda (1-t)} \bE \Big[ \Tr \langle \tilde{\bX}^\intercal  \tilde{\bX} \rangle_t  - \Tr \langle \tilde{\bX} \rangle_t^\intercal  \langle \tilde{\bX} \rangle_t \Big]
\end{split}
\end{equation*}
Therefore \eqref{time-deivative} can be written as:
\begin{equation}
\begin{split}
        \frac{d}{d t}F_N(t) &= -\frac{1}{M} \frac{\lambda}{2} \bE \Big[ 2 \Tr \langle \bX^\intercal  \bS \rangle_t - \Tr \langle \bX \rangle_t^\intercal  \langle \bX \rangle_t - 2 \Tr \langle \tilde{\bX}^\intercal  \tilde{\bS} \rangle_t + \Tr \langle \tilde{\bX} \rangle_t^\intercal  \langle \tilde{\bX} \rangle_t \Big] \\
        &= - \frac{1}{M} \frac{\lambda}{2} \bE \big[ \Tr [ \langle \bX^\intercal  \bS \rangle_t -  \langle \tilde{\bX}^\intercal  \tilde{\bS} \rangle_t ] \big] \hspace{10pt} \text{(By Nishimori)}
\end{split}
\label{time-deivative-Nishimori}
\end{equation}

We have:
\begin{equation}
    \begin{split}
        \frac{2 M }{\lambda} \Big|   \frac{d}{d t}F_N(t) \Big| &= \Bigg| \bE \bigg[ \Big \langle \Tr  \big[ \bS^\intercal  (\bX - \tilde{\bX})  -  (\tilde{\bS}^\intercal  - \bS)\tilde{\bX} \big] \Big\rangle_t  \bigg] \Bigg| \\
        & \leq   \bE \Bigg[ \bigg\langle \Big| \Tr  \big[ \bS^\intercal  (\bX - \tilde{\bX})  -  (\tilde{\bS}^\intercal  - \bS)\tilde{\bX} \big] \Big| \bigg\rangle_t  \Bigg] \hspace{5mm} \text{(By Jensen)} \\
        &  \leq \bE \Bigg[ \bigg\langle \Big| \Tr \bS^\intercal  (\bX - \tilde{\bX}) \Big| \bigg\rangle_t  \Bigg]   + \bE \Bigg[ \bigg\langle \Big| \Tr  (\tilde{\bS} - \bS)\tilde{\bX}^\intercal ] \Big| \bigg\rangle_t  \Bigg] \\
        & \leq    \bE \Big[ \| \bS \|_{\rm F}  \langle \| \bX - \tilde{\bX} \|_{\rm F}  \rangle_t  \Big]   + \bE \Big[ \| \bS - \tilde{\bS} \|_{\rm F} \langle \|\tilde{\bX}\|_{\rm F} \rangle_t  \Big] \\
        & \leq \sqrt{ \bE \Big[ \| \bS \|^2_{\rm F} \Big] \bE \Big[ \big \langle \| \bX - \tilde{\bX} \|_{\rm F}  \big \rangle_t^2  \Big] }   + \sqrt{\bE \Big[ \| \bS - \tilde{\bS} \|_{\rm F}^2 \Big] \bE \Big[ \big \langle \|\tilde{\bX}\|_{\rm F} \big \rangle_t^2  \Big] } \hspace{10pt} \text{(By Cauchy–Schwarz)} \\
        & \leq \sqrt{ \bE \big[ \| \bS \|^2_{\rm F} \big] \bE \Big[ \big \langle \| \bX - \tilde{\bX} \|_{\rm F}^2  \big \rangle_t  \Big] }   + \sqrt{\bE \big[ \| \bS - \tilde{\bS} \|_{\rm F}^2 \big] \bE \Big[ \big \langle \|\tilde{\bX}\|_{\rm F}^2 \big \rangle_t \Big] } \hspace{10pt} \text{(By Cauchy–Schwarz)} \\
        & = \sqrt{ \bE \big[ \| \bS \|^2_{\rm F} \big] \bE \big[  \| \bS - \tilde{\bS} \|_{\rm F}^2   \big] }   + \sqrt{\bE \big[ \| \bS - \tilde{\bS} \|_{\rm F}^2 \big] \bE \big[ \|\tilde{\bS}\|_{\rm F}^2  \big] } \hspace{10pt} \text{(By Nishimori)} \\
        &= \Big( \sqrt{ \bE \big[ \| \bS \|^2_{\rm F} \big]} + \sqrt{ \bE \big[ \| \tilde{\bS} \|^2_{\rm F} \big]} \Big) \sqrt{\bE \big[ \| \bS - \tilde{\bS} \|_{\rm F}^2    \big]} \\
        &= \Big( \sqrt{ \bE_{\bsig} \big[ \| \bsig \|^2 \big]} + \sqrt{ \bE_{\tilde{\bsig}} \big[ \| \tilde{\bsig} \|^2 \big]} \Big) \sqrt{\bE_{\bsig, \tilde{\bsig}} \big[ \| \bsig - \tilde{\bsig} \|^2 \big] } 
    \end{split}
    \label{bounde-on-derivative}
\end{equation}

We obtain the result by integrating \eqref{bounde-on-derivative}, over $t$ from $0$ to $1$, and using $N \leq M$.$\hfill \square$
% \end{widetext}

\subsection{Proof of lemma \ref{W2-dis}}\label{proof of lemma 2}%{$\bE_{Q_N} \big[ \| \bS - \bS' \|_{\rm F}^2 \big] = o(N) $}
First, note that by rotational invariance, $p_{S}(\bsig)$ is invariant under permutations, so without loss of generality, we can assume $\bsig$ is in non-decreasing order.

Since $p_{\tilde{S}}(\tilde{\bsig})$ is a delta distribution, we can easily write
\begin{equation}
    \bE_{\bsig, \tilde{\bsig}} \big[ \| \bsig - \tilde{\bsig} \|^2 \big] =  \bE_{\bsig} \big[ \| \bsig - \bsig^0 \|^2 \big]
    \label{W-2-1}
\end{equation}
% We have
% \begin{equation}
% \begin{split}
%         \bE_{Q_{N}} \big[ \| \bS - \tilde{\bS} \|_{\rm F}^2 \big] &= \bE_{P_{S,N}(\bsig)} \bE_{P_{\tilde{S},N}(\tilde{\bsig})} \bE_{\bU}  \big[ \| \bU \bsig \bU^\intercal  - \bU \tilde{\bsig} {\bU}^\intercal  \|_{\rm F}^2 \big] \\
%         &= \bE_{P_{S,N}(\bsig)} \bE_{P_{\tilde{S},N}(\tilde{\bsig})} \big[ \| \bsig  - \tilde{\bsig}\|_{\rm F}^2 \big] \\
%         &= \bE_{P_{S,N}(\bsig)} \bE_{P_{\tilde{S},N}(\tilde{\bsig})} \big[ \| \bsig  - \tilde{\bsig}\|^2 \big] \\
%         &= \bE_{P_{S,N}(\bsig)}  \big[ \| \bsig  - \bsig^0 \|^2 \big]
% \end{split}
% \label{W-2-1}
% \end{equation}

For a vector $\bsig$, denote the empirical distribution of its components by $\hat{\mu}_{\bsig}$. The Wasserstein-2 distance between two empirical distributions, $\hat{\mu}_{\bsig}, \hat{\mu}_{\bsig^0}$ is defined as 
\begin{equation*}
\begin{split}
    W_2(\hat{\mu}_{\bsig}, \hat{\mu}_{\bsig^0}) = \sqrt{\inf_{\gamma \in \Gamma(\hat{\mu}_{\bsig}, \hat{\mu}_{\bsig^0})} \bE_{\gamma(x, y)}\big[ ( x  - y)^2 \big]}
\end{split}
\end{equation*}
with $\Gamma(\hat{\mu}_{\bsig}, \hat{\mu}_{\bsig^0})$ is the set of couplings of $\hat{\mu}_{\bsig}, \hat{\mu}_{\bsig^0}$. By lemma H.5  in \cite{pourkamali2023matrix} ,the Wasserstein-2 distance can be written as
\begin{equation}
\begin{split}
    W_2(\hat{\mu}_{\bsig}, \hat{\mu}_{\bsig^0}) = \sqrt{\min_{\pi \in \mathcal{S}_N}  \frac{1}{N} \| \bsig  - \bsig^0_{\pi} \|^2 }
\end{split}
\label{Wasser-equiv}
\end{equation}
where $\bsig^0_{\pi}$ is the permuted version of $\bsig^0$, and $\mathcal{S}_N$ is the set of all $N$-permutations. So, for a given $\bsig$ and $\bsig^0$ (which have a non-decreasing order), we have (considering the identity permutation)
\begin{equation}
    \| \bsig - \bsig^0 \|^2 \geq N W_2(\hat{\mu}_{\bsig}, \hat{\mu}_{\bsig^0})^2
        \label{W-2-2}
\end{equation}

On the other hand, for any permutation of $\bsig^0$ (in particular, the one which achieves the minimum in \eqref{Wasser-equiv}), we have
\begin{equation*}
    \begin{split}
        \| \bsig  - \bsig^0_{\pi} \|^2 &= \| \bsig \|^2 + \| \bsig^0_{\pi} \|^2 - 2 \bsig^\intercal  \bsig^0_{\pi} \\
        & \geq \| \bsig \|^2 + \| \bsig^0_{\pi} \|^2 - 2 \bsig^\intercal  \bsig^0 = \| \bsig  - \bsig^0 \|^2
    \end{split}
\end{equation*}
where we used rearrangement inequality \cite{hardy1952inequalities} to get the inequality in the second line. So,
\begin{equation}
    \| \bsig - \bsig^0 \|^2 \leq N W_2(\hat{\mu}_{\bsig}, \hat{\mu}_{\bsig^0})^2
        \label{W-2-3}
\end{equation}

From \eqref{W-2-1}, \eqref{W-2-2},\eqref{W-2-3}, we have
\begin{equation}
  \bE_{\bsig, \tilde{\bsig}} \big[ \| \bsig - \tilde{\bsig} \|^2 \big] = \bE_{\bsig} \big[N W_2(\hat{\mu}_{\bsig}, \hat{\mu}_{\bsig^0})^2 \big]
    \label{W-2-5}
\end{equation}
Lemma \ref{E-W-0-l} concludes the proof.

% \subsubsection{Lemmas}
% \begin{lemma}\label{W-perm}
% Given two vectors $\bu, \bv \in \bR^N$, denote their empirical distributions by $\mu, \nu$ respectively. we have
% \begin{equation*}
%     W_2 (\mu, \nu ) = \sqrt{\min_{\pi \in \mathcal{S}_N} \frac{1}{N} \| \bu - \bv_{\pi} \|^2 } 
% \end{equation*}
% where $\bv_{\pi}$ is a permutation of $\bv$.
% \end{lemma}
% \begin{proof}
% By definition, we have
% \begin{equation*}
%     W_2 (\mu, \nu )^2 = \inf_{\gamma \in \Gamma(\mu, \nu) } \bE_{\gamma(x,y)} \big[ (x - y)^2 \big]
% \end{equation*}
% Any measure in  $\Gamma(\mu, \nu)$ can be represented by a doubly stochastic $N \times N$ matrix. Thus, we have
% \begin{equation*}
%     W_2 (\mu, \nu )^2 = \inf_{\bm{P} \in \mathcal{B}_N } \frac{1}{N} \sum_{i,j} P_{ij} (u_i - v_j)^2
% \end{equation*}
% where $\mathcal{B}_N$ denotes the set of doubly stochastic matrices.

% This minimization problem is a linear optimization problem on the bounded convex set $\mathcal{B}_N$. By Choquet's theorem, the solutions to this problem exist and are the extremal points of $\mathcal{B}_N$, which are permutation matrices (by Birkhoff's theorem). Therefore, the minimization can be written on the set of permutation matrices to get:
% \begin{equation*}
%     W_2 (\mu, \nu )^2 = \min_{\pi \in \mathcal{S}_N } \frac{1}{N} \sum_{i,j}  (u_i - v_{\pi(i)})^2
% \end{equation*}
% \end{proof} 

\begin{lemma}\label{E-W-0-l}
Suppose $\bsig \in \bR^N$ is distributed according to $p_{S}(\bsig)$, and $\bsig^0$ is generated with i.i.d. elements from $\mu_S$. Let $\hat{\mu}_{\bsig}, \hat{\mu}_{\bsig^0}$ be their empirical distribution. We have:
\begin{equation*}
    \lim_{N \to \infty} \bE_{\bsig} \big[ W_2(\hat{\mu}_{\bsig}, \hat{\mu}_{\bsig^0})^2 \big] = 0
\end{equation*}
\end{lemma}
\begin{proof}
By triangle inequality, we have:
\begin{equation}
    W_2(\hat{\mu}_{\bsig}, \hat{\mu}_{\bsig^0}) \leq W_2(\hat{\mu}_{\bsig}, \mu_S) + W_2(\hat{\mu}_{\bsig^0}, \mu_S)
\end{equation}
From the weak convergence and convergence of second moment, assumptions \ref{assumptions on law} and \ref{bounded-mom} imply that the second moment of the empirical spectral distribution converges almost surely to the one of $\mu_S$. 
Thus, by  \cite{villani2021topics}(Theorem 7.12),  the empirical singular value distribution in the Wasserstein-2 metric to  $\mu_S$. Hence, the first term approaches $0$ as $N \to \infty$ almost surely.

By law of large numbers and since the support of $\mu_S$ is bounded, the second term also converges $0$ as $N \to \infty$. Therefore, we have $W_2(\hat{\mu}_{\bsig}, \hat{\mu}_{\bsig^0}) \to 0$ almost surely. Consequently, we have that $W_2(\hat{\mu}_{\bsig}, \hat{\mu}_{\bsig^0})^2 \to 0$ almost surely.

One can see that:
\begin{equation}
    \begin{split}
        W_2(\hat{\mu}_{\bsig}, \hat{\mu}_{\bsig^0})^2 & \leq \frac{2}{N} \sum \sigma_i^2 + \frac{2}{N} \sum {\sigma^0_i}^2 \\
        &\leq 2  m^{(2)}_{\hat{\mu}_{\bsig}} + 2 C_2^2
    \end{split}
\end{equation}
with $m^{(2)}_{\hat{\mu}_{\bsig}}$ the second moment of $\hat{\mu}_{\bsig}$ which is almost surely bounded by assumption. Therefore, the result follows by using dominated convergence theorem.
\end{proof}

% \begin{lemma}\label{W-conv-emp}
% Let $X_1, \hdots, X_N$ be i.i.d. random variables distributed according to distribution $\mu$, which has finite support. Let $\mu_N$ denote their empirical distribution. We have:
% \begin{equation*}
%     \lim_{N \to \infty} W_2(\mu_N, \mu ) = 0 \hspace{5mm} \text{almost surely}
% \end{equation*}
% \end{lemma}
% \begin{proof}
% By the law of large numbers, $\mu_N \to \mu$ almost surely. Moreover, since $\mu$ has bounded support, the second moment of $\mu_N$ converges to the one of $\mu$. By Theorem 7.12 in \cite{villani2021topics}, we have the convergence in the Wasserstein-2 metric almost surely.
% \end{proof}

% \subsubsection{On Assumption \ref{E-W-0}}
% We can replace assumption \ref{E-W-0} with the assumption that the second moment of the empirical distribution of $\bsig \propto P_{S,N}(\bsig)$ converges to the second moment of $\mu_S$. Then, by Theorem 7.12 in \cite{villani2021topics}, weak convergence and convergence of the second moment implies that the empirical distribution converges in Wasserstin-2 distance towards $\mu_S$, $\lim_{N \to \infty} W_2(\hat{\mu}_{\bsig}, \mu_S) = 0$. Then, assumption \ref{E-W-0} can be justified as follows:

% By definition of Wasserstein distance, we have
% \begin{equation}
%     \begin{split}
%         \bE_{P_{S,N} (\bsig)} \big[ W_2(\hat{\mu}_{\bsig}, \hat{\bsig^0})^2 \big] &= \bE_{P_{S,N} (\bsig)} \big[ \inf_{\gamma \in \Gamma(\hat{\mu_{\bsig}} \big]
%     \end{split}
% \end{equation}
% To show that the expectation converges to 0, we use Theorem 2.6.6 in \cite{anderson2010introduction} which states a large de

\section{Rectangular spherical integrals}\label{spherical integral app}
For the matrices $\bA, \bB \in \bR^{N \times M}$ the \textit{rectangular spherical integral} is defined, , as:
\begin{equation}\label{rect-sph-int}
    \mathcal{I}_{N,M} (\bA, \bB) = \Big\langle \exp \big\{ \sqrt{N M} \Tr \bA \bU \bB \bV \big\} \Big\rangle_{\bU,\bV}
\end{equation}
where $\bU \in \bR^{ N \times N}, \bV \in \bR^{M \times M}$, and the expectation is w.r.t. the \textit{Haar} measure over orthogonal matrices of size $N\times N$ and $M\times M$. Its symmetric counterpart defined as $\Big\langle \exp \big\{ \frac{N}{2} \Tr \bA \bU \bB \bU^\intercal \big\} \Big\rangle_{\bU}$ often referred to as \textit{Harish Chandra-Itzykson-Zuber} (HCIZ) integrals in mathematical physics literature. The study of these objects dates back to the work of mathematician Harish Chandra \cite{harish1957differential} and they (in particular the symmetric case) have since been extensively studied and developed in both physics and mathematics. Here, we only focus on the rectangular case.

\subsection{Rank-one case}
Benaych-George in \cite{benaych2011rectangular} studied the asymptotic limit of \eqref{rect-sph-int} in the case where $\bA$ is a rank-one matrix.

\begin{theorem}[\textbf{Rank-one  rectangular spherical integral}, Benaych-Georges \cite{benaych2011rectangular}]\label{rank-one-rect-sphericla-integral} Let $N/M \to \alpha \in (0,1]$, and $\theta$ be the only non-zero singular value of $\bA$, and the empirical singular value distribution of $\bB$ converges weakly towards $\mu_B$. Then, for $\theta$ sufficiently small (see details in Theorem 2.2 in \cite{benaych2011rectangular}), we have:
\begin{equation}
    \lim_{N \to \infty} \frac{1}{N} \ln \mathcal{I}_{N,M} (\bA, \bB) =  \int_0^{\theta} \frac{\mathcal{C}_{\mu_B}^{(\alpha)}(t^2)} {t} \, dt = \frac{1}{2} \int_0^{\theta^2} \frac{\mathcal{C}_{\mu_B}^{(\alpha)}(t)}{t} \, dt \equiv \frac{1}{2} \mathcal{Q}_{\mu_B}^{(\alpha)}(\theta^2)
    \label{rect-spherical-limit}
\end{equation}
\end{theorem}

It is known that additional terms may be present on the rhs of \eqref{rect-spherical-limit} when the parameter $\theta$ is "large", see for example Statement 3 in \cite{pourkamali2022mismatched} or section 6.2 in \cite{mergny2022right}. In the replica calculation the order of magnitude of this parameter is determined by the solutions of the saddle point equations, but  it is difficult to fully control its order of magnitude. However the numerics show very good agreement between our explicit RIEs and the Oracle estimator, which strongly suggests it is sound to use \eqref{rect-spherical-limit}. Moreover, in our derivation, we use a generalization of this formula, namely when $\bA$ has higher (but fixed) rank, the limit is the sum over singular values of the expression on the rhs of \eqref{rect-spherical-limit}. Although we are not aware if this generalization has been proved, we believe that the ideas found in \cite{guionnet2022asymptotics} can be applied to show it holds.

\subsection{Growing-rank regime}
The asymptotic of the rectangular spherical integral, when rank of both matrices $\bA, \bB$ grows (linearly) with the dimension has been studied in \cite{guionnet2021large}. It is shown that under certain assumptions on the matrices and their limiting ESDs the log-spherical integral converges to a limit, and authors derive the asymptotic limit in terms of a variational formula. Due to the complexity of the theorem, we skip stating it and refer the reader directly to Theorem 1.1 in \cite{guionnet2021large}.

\end{appendices}

%USE THE BELOW OPTIONS IN CASE YOU NEED AUTHOR YEAR FORMAT.
% \bibliographystyle{abbrvnat}
% \bibliography{Reference}

\end{document}